\documentclass[letterpaper,11pt]{article}
\usepackage{fullpage}
\usepackage[ansinew]{inputenc}
\usepackage{verbatim} 
\usepackage{graphicx}
\usepackage{subfigure}
\usepackage{tabularx}
\usepackage{array}
\usepackage{delarray}
\usepackage{dcolumn}
\usepackage{float}
\usepackage{amsmath}
\usepackage{psfrag}
\usepackage{booktabs}
\usepackage{multirow,hhline}
\usepackage{amstext}
\usepackage{amssymb}
\usepackage{fancyhdr}
\usepackage{hyperref}
\usepackage{setspace}
\usepackage{amsthm}
\usepackage{makeidx}
\usepackage{tikz}
\usepackage{lipsum}
\usepackage{url}
\usepackage{adjustbox}
\usepackage{enumitem}
\usepackage{mathtools}
\usepackage{threeparttable} 
\usepackage{ragged2e} 

\RequirePackage[round,authoryear,longnamesfirst]{natbib}
\usetikzlibrary{arrows,calc}
\hypersetup{citecolor=true,colorlinks=true,allcolors=blue}
\makeindex
\allowdisplaybreaks
\numberwithin{equation}{section}
\newtheorem{lem}{Lemma}[section]
\newtheorem{thm}{Theorem}[section]
\newtheorem{cor}{Corollary}[section]

\newtheorem{ass}{Assumption}

\theoremstyle{definition}
\newtheorem{rem}{Remark}
\newtheorem{ex}{Example}
\renewcommand{\citep}[1]{\citeauthor{#1}, \citeyear{#1}}

\newcommand{\diag}{\text{diag}}
\newcommand{\Supp}{\text{Supp}}
\newcommand{\indep}{\perp\!\!\!\perp}

\newcommand{\convP}{\stackrel{p}{\longrightarrow}}
\newcommand{\convD}{\rightsquigarrow}

\newcommand{\N}{\mathcal{N}}
\newcommand{\eps}{\varepsilon}

\renewcommand{\epsilon}{\varepsilon}

\DeclareMathOperator*{\argmin}{arg\,min}

\usepackage{xr}
\makeatletter
\newcommand*{\rom}
[1]{\expandafter\@slowromancap\romannumeral #1@}
\makeatother

\makeatletter
\newcommand*\bigcdot{\mathpalette\bigcdot@{.5}}
\newcommand*\bigcdot@[2]{\mathbin{\vcenter{\hbox{\scalebox{#2}{$\m@th#1\bullet$}}}}}
\makeatother

\title{Adjustments with Many Regressors under Covariate-Adaptive Randomizations\thanks{Ke Miao thanks the National
		Natural Science Foundation of China for financial support under grant number 72103046. Yichong Zhang acknowledges the financial support from the Lee Kong Chian fellowship and the NSFC under the grant No. 72133002. Any and all errors are our own.\vspace{1.3mm}} 
	\\ \vspace{2mm}
}
\author{Liang Jiang\thanks{Fudan University.\ E-mail~address: jiangliang@fudan.edu.cn.} \and Liyao Li\thanks{The corresponding author. East China Normal University.\ E-mail~address:  lyli@fem.ecnu.edu.cn. } \and Ke Miao\thanks{Fudan University.\ E-mail~address: miaoke@fudan.edu.cn.}\and Yichong Zhang\thanks{%
		Singapore Management University.\ E-mail~address: yczhang@smu.edu.sg.}
	\date{}
}

\begin{document}
	\maketitle
	
	\begin{abstract}
	Our paper discovers a new trade-off of using regression adjustments (RAs) in causal inference under covariate-adaptive randomizations (CARs). On one hand, RAs can improve the efficiency of causal estimators by incorporating information from covariates that are not used in the randomization. On the other hand, RAs can degrade estimation efficiency due to their estimation errors, which are not asymptotically negligible when the number of regressors is of the same order as the sample size. Ignoring the estimation errors of RAs may result in serious over-rejection of causal inference under the null hypothesis. To address the issue, we construct a new ATE estimator by optimally linearly combining the estimators with and without RAs. We then develop a unified inference theory for this estimator under CARs. It has two features: (1) the Wald test based on it achieves the exact asymptotic size under the null hypothesis, regardless of whether the number of covariates is fixed or diverges no faster than the sample size; and (2) it guarantees weak efficiency improvement over estimators both with and without RAs.  
		


		\bigskip
		
		
	\end{abstract}
	\clearpage

\section{Introduction}
This paper studies linear regression adjustments (RAs) for the estimation and inference of the average treatment effect (ATE) under covariate-adaptive randomizations (CARs) when there are many regressors. CARs have recently seen growing use in a wide variety of randomized experiments in economic research.\footnote{See, for example, \cite{CCFNT16,greaney2016,jakiela2016,burchardi2019,anderson2021}, and etc.} Under CARs, units are first stratified using some baseline covariates, and within each stratum, the treatment status is assigned independently to achieve the balance between the numbers of treated and control units. Given data from CARs, researchers often estimate and infer various treatment effect parameters via regressions with strata dummies and baseline covariates as controls, a practice known as RAs. However, \cite{F08b,F081} showed that the usual OLS regression with covariates can actually decrease the precision of the ATE estimator. \cite{L13} found that, under complete randomization, for a linear regression with covariates to guarantee efficiency improvement upon the simple difference-in-means estimator (i.e., ``no-harm''), it must include a full set of interactions between treatment status and covariates. Under more complicated CARs, the ``no-harm'' regressions are performed stratum by stratum in a fully saturated fashion, as pointed out by \cite{BCS18} and \cite{YYS22}.


In many economics applications, the ``no-harm'' linear regressions are usually coupled with high dimensional regressors, as noted by \cite{duflo2018machinistas}.\footnote{Examples of such experiments include \cite{bursztyn2019,banerjee2020lack,dhar2022reshaping}, among others.} Moreover, when the regressors are sieve bases with a growing dimension, causal estimators with ``no-harm'' RAs can potentially achieve the semiparametric efficiency bound, as shown by \cite{JLTZ22} and \cite{BJRSZ23}. However, \cite{CJN18_ET,CJN18} demonstrated that the usual heteroskedasticity consistent inference method for linear regressions leads to misleading inference when there are many regressors. \cite{CJM18} found similar problems for a class of non-linear settings. Under CARs, this issue of dimensionality is further exacerbated with the ``no-harm'' RAs stratum by stratum.


In this paper, we consider the ``many regressors'' regime--an asymptotic regime where the number of regressors can diverge at most as fast as the sample size and discover a new efficiency trade-off of using ``no-harm'' RAs due to the dimensionality of regressors. On one hand, RAs can improve the estimation efficiency by incorporating information from covariates that are not used in the randomization; on the other hand, they can degrade estimation efficiency due to their estimation errors, which are not asymptotically negligible when the number of regressors grows at the same rate of the sample size. This trade-off poses potentially serious challenges for inference. First, ignoring the cost of RAs in the asymptotic variance estimation may lead to substantial over-rejection under the null. Second, a consistent variance estimator that accounts for the cost of RAs can potentially be larger than that of the simple estimator without RAs, even asymptotically, if the cost outweighs the benefit. Therefore, using a ``no-harm'' RA that incurs such costs can actually be harmful.  

To resolve these issues, we derive the joint asymptotic distribution for estimators both with and without RAs under the ``many regressors'' asymptotic regime and propose a consistent estimator of the asymptotic covariance matrix, which accounts for both the benefit and cost of RAs. We then use this matrix to construct a new ATE estimator by optimally linearly combining the two estimators. We show that (1) the Wald test based on this new optimal linear combination estimator achieves the exact asymptotic size under the null and (2) the optimal linear combination estimator is weakly more efficient than estimators both with and without RAs. Furthermore, we conduct a local asymptotic power analysis and establish that the Wald test based on the optimal linear combination estimator is asymptotically the uniformly most powerful (UMP) test against two-sided local alternatives over a class of unbiased tests that only depend on estimators with and without RAs. This implies our new estimator is optimal locally among not only linear but also nonlinear combination estimators. These results collectively address the challenge of the efficiency trade-off unearthed in the paper.

We further consider an alternative asymptotic regime where the dimension of regressors is fixed or moderate. Under mild regularity conditions, we show that the \textit{same} optimal linear combination estimator with the \textit{same} covariance matrix estimator achieves the \textit{same} properties (1) and (2) as above. Importantly, these two results hold even when the RAs are not approximately correctly specified.

\textbf{Relation to the literature.} Our results build on the work of \cite{CJN18}, who studied the inference of linear coefficients in OLS regression with many regressors and independent or cluster-independent observations. We extend their analysis in three new directions. First, we consider observations generated by CARs, which introduce cross-sectional dependence among treatment assignments and outcomes. Second, we estimate the ATE in a two-step procedure by linearly combining intercepts from multiple OLS regressions with many regressors. Because the first step intercept estimators are asymptotically dependent and the second step of the procedure introduces additional estimation errors, the asymptotic distribution of our ATE estimator is not a simple application of the result by \cite{CJN18}. Instead, we develop a new distributional theory for our ATE estimator using techniques from \cite{BCS17}, the Yurinskii's coupling,\footnote{We avoid the mistake mentioned in Remark 2.1 of \cite{CMU22} by formulating our result such that the Gaussian approximation error $\delta_n$ is explicitly defined as a vanishing sequence with respect to the sample size.} and an anti-concentration inequality from \cite{CCK14}. Third, we propose a new optimal linear combination estimator and discuss its efficiency improvement under the many regressors framework, which is novel in the literature. \cite{CJM18} also considered a two-step procedure, and their knife-edge order of the number of regressors is the square root of the sample size. Our paper differs from theirs because our first-step intercept estimators enter the second step linearly, resulting in a knife-edge rate of the same order as the sample size.

Our covariance matrix estimator builds on cross-fit estimators for both the variance of regression coefficients and the ``variance component'' with many regressors, as proposed by \cite{J22} and \cite{KSS2020}. To address the dependence among the first-step intercept estimators, we also extend \cite{KSS2020} by proposing a new estimator for the ``covariance component'' with many regressors, involving coefficients from two separate linear regressions. Other recent studies on inference in linear regression with many regressors include \cite{AS23} and \cite{MS23}.

This paper also relates to other works on RAs in randomized experiments. 
\cite{JLTZ22},  \cite{JPTZ22}, \cite{LTM20}, \cite{MTL20}, \cite{WDTT16}, and \cite{YYS22} considered RAs for various causal parameters when the covariates are either fixed dimensional or high-dimensional but sparse. In these settings, the estimation errors of RAs are asymptotically negligible. However, as argued by \cite{LM21}, the sparsity condition (1) may not hold in social science applications as it is unclear why the large majority of control coefficients should be very nearly zero, (2) is not invariant to linear transformations of the controls, and (3)  depends on some tuning parameter (e.g., the penalty term in Lasso), which complicates small sample interpretation of the resulting inference. \cite{WZ23} further demonstrated that both Lasso and debiased Lasso can exhibit significant omitted variable biases, even when the coefficient vector is sparse and the sample size is sufficiently large compared to the number of controls. In such cases, the ``long regression'' approach often outperforms both Lasso and debiased Lasso, utilizing the modern high-dimensional OLS-based inference methods from \cite{CJN18_ET, CJN18, CJM18}. Our approach does not need sparsity conditions, thus avoiding these issues. Under a \textit{finite-population} framework with \textit{complete random sampling}, \cite{LD21} and \cite{CMO23} proposed bias correction for the ATE while allowing the number of regressors to grow at a rate \textit{slower} than the sample size. More recently, \cite{LYW23} introduced a debiased regression-adjusted estimator under \textit{complete random sampling}, permitting the number of regressors to grow at the same rate as the sample size. In contrast, we adopt a \textit{super-population} framework with general \textit{CARs}. A more detailed comparison with \cite{LD21}, \cite{CMO23}, and \cite{LYW23} is provided in Remark \ref{rem:specification} below. \cite{T18} further considered ``optimal" stratification under CARs with a pilot experiment. 






Recent studies (\citep{Bai22}, \citep{C23optimal}, and \citep{BLST23}) have highlighted the optimality of ``finely stratified'' experiments. However, this paper focuses on ``coarse'' CARs for two main reasons. First, CARs, including simple random sampling and stratified block randomization, are widely used in empirical research as demonstrated by works such as \cite{CCFNT16}, \cite{jakiela2016}, \cite{burchardi2019}, \cite{anderson2021}. Based on a survey of selected development economists, \cite{B09} reported that about 40\% of researchers have used such a design at some point. Second, ``fine stratifying'' continuous covariates at the experimental design stage may not be applicable in some scenarios because (1) sometimes, all relevant covariates are discrete, as seen in \cite{dupas2018}, where randomization is stratified by gender, occupation, and bank branch-- all discrete variables; (2) researchers may need to analyze either past experiments or those conducted by others using CARs. Our paper thus aligns with the literature on RAs under CARs by taking the randomization scheme as given and aims to develop more efficient estimators using covariate information at the data analysis stage. Therefore, it complements the strategy of ``fine stratification'' at the experimental design stage because our method applies at the data analysis stage.

The rest of this paper is organized as follows. Section \ref{sec:setup} introduces the setup. Section \ref{sec:trade-off} identifies the efficiency trade-off in RAs with many regressors. In Section \ref{sec:optimal}, we present the optimal linear combination estimator for ATE and its asymptotic properties. Section \ref{sec:fixed} discusses the properties of the optimal linear combination estimator with fixed or moderate number of regressors.  Sections \ref{sec:simulations} and \ref{sec:app} provide simulations and an empirical application, respectively. Proofs and additional figures are given in the Online Appendix.


\textbf{Notation.} For any positive integer $m$, let $0_m$, $1_m$ and $I_m$ be the $m \times 1$ vector of zeros, $m \times 1$ vector of ones, and the $m \times m$ identity matrix, respectively. Let $||\cdot||_2$ and $||\cdot||_{op}$ denote the $\ell_2$-norm for a vector and operator norm for a matrix, respectively. For a symmetric and positive semi-definite matrix $\Upsilon$, $\lambda_{\max}(\Upsilon)$ denotes the maximum eigenvalue of $\Upsilon$. We define $W^{(n)}$ as the sample of $W$'s, i.e., $W^{(n)} = \{W_i\}_{i \in [n]}$, where $[n] = 1,\cdots,n$. We write $U \stackrel{d}{=} V$ for two random variables $U$ and $V$ if they share the same distribution.

\section{Setup}
\label{sec:setup}
Potential outcomes for treated and control groups are denoted by $Y(1)$ and $Y(0)$, respectively. Treatment status is denoted by $A$, with $A=1$ indicating treated and $A=0$ untreated. The stratum indicator is denoted by $S$, based on which the researcher implements the covariate-adaptive randomization. The support of $S$ is denoted by $\mathcal{S}$, a finite set. After randomization, the researcher can observe the data $\{Y_i,S_i,A_i,X_i\}_{i \in [n]}$ where $Y_i = Y_i(1)A_i + Y_i(0)(1-A_i)$ is the observed outcome, and $X_i$ contains covariates besides $S_i$ in the dataset. The support of $X$ is denoted as $\Supp(X)$. In this paper, we allow $X_i$ and $S_i$ to be dependent. For $s \in S$, let $n_{s} = \sum_{i \in [n]}1\{S_i = s\}$, $n_{1,s} = \sum_{i \in [n]}A_i1\{S_i=s\}$, and $n_{0,s} = n_{s} - n_{1,s}$. Let $\aleph_{a,s} = \{i \in [n]: A_i = a,S_i=s\}$ denote the set of individuals in stratum $s$ with treatment status $a$ and $\aleph_{s} = \aleph_{1,s} \cup \aleph_{0,s} $. Our parameter of interest is the average treatment effect defined as $\tau = \mathbb{E}(Y(1)-Y(0))$.


We make the following assumptions on the data generating process (DGP) and the treatment assignment rule.
\begin{ass}
	\begin{enumerate}[label=(\roman*)]
		\item $\{Y_i(1),Y_i(0),S_i,X_i\}_{i \in [n]}$ are independent and identically distributed (i.i.d.). 
		\item $\{Y_i(1),Y_i(0),X_i\}_{i \in [n]} \indep \{A_i\}_{i \in [n]}|\{S_i\}_{i \in [n]}$.
		\item Suppose $p_s = \mathbb{P}(S_i = s)$ is fixed with respect to (w.r.t.) $n$ and is positive for every $s \in \mathcal{S}$.
		\item  Let $\pi_s$ denote the target fraction of treatment for stratum $s$. Then, $c<\min_{s \in \mathcal{S}}\pi_s \leq \max_{s \in \mathcal{S}}\pi_s<1-c$ for some constant $c \in (0,0.5)$ and $\frac{D_{n,s}}{n_{s}} = o_P(1)$ for $s \in \mathcal{S}$, where $D_{n,s} = \sum_{i \in [n]} (A_i-\pi_s)1\{S_i = s\}$.
	\end{enumerate}
	\label{ass:assignment1}
\end{ass}

\begin{rem}
	\label{rem:ass1}
	Several remarks are in order. First, Assumption \ref{ass:assignment1}(i) assumes $(Y(1),Y(0),S,X)^{(n)}$ are independent. This assumption is essential for all the theoretical results presented in this paper, which implies that our findings primarily apply to cross-sectional data.    However, we still allow for $A^{(n)}$, and thus,  $Y^{(n)}$ to be cross-sectionally dependent, which will be the case under CARs. The identical distribution assumption can be relaxed by using a set of more complex notations. Second, Assumption \ref{ass:assignment1}(ii)
	implies that the treatment assignment $A^{(n)}$ are generated only based on
	strata indicators. Third, Assumption \ref{ass:assignment1}(iii) imposes that the number of strata is bounded and 
	the strata sizes are approximately balanced. It is also interesting to explore settings where the number of strata increases with the sample size. For instance, in the matched pairs design (see, e.g., \cite{BRS19}), the number of strata grows proportionally to $n$. More generally, it could scale as $n^\alpha$  for some $\alpha \in (0,1]$. Investigating these scenarios is an avenue for future research. Fourth, \cite{BCS17} show that
	Assumption \ref{ass:assignment1}(iv) holds under several specific covariate-adaptive
	treatment assignment rules such as simple random sampling (SRS), biased-coin
	design (BCD), adaptive biased-coin design (WEI) and stratified block
	randomization (SBR). For completeness, we provide brief descriptions
	below. Note that the requirement $D_{n,s}/n_s=o_P(1)$ is weaker than the assumption imposed by \cite{BCS17}, but it is the same as that imposed by \cite{BCS18} and \cite{ZZ20}.
\end{rem}

\begin{ex}
	[SRS] \label{ex:srs} Let $\{A_{i}\}_{i=1}^{n}$ be drawn independently across
	$i$ and of $\{S_{i}\}_{i=1}^{n}$ as Bernoulli random variables with success
	rate $\pi_s$, i.e., for $k=1,\ldots,n$,
	\[
	\mathbb{P}\left(  A_{k}=1\big|\{S_{i}\}_{i=1}^{n},\{A_{j}\}_{j=1}%
	^{k-1}\right)  =\mathbb{P}(A_{k}=1|S_k)=\pi_{S_k}.
	\]
\end{ex}

\begin{ex}
	[WEI] \label{ex:wei}  This design was first proposed by \cite{W78}. Let
	$n_{k-1}(S_{k}) = \sum_{i=1}^{k-1}1\{S_{i} = S_{k}\}$, $B_{k-1}(S_{k}) =
	\sum_{i=1}^{k-1}\left( A_{i} - \frac{1}{2} \right)  1\{S_{i} = S_{k}\}$, and
	\begin{align*}
	\mathbb{P}\left( A_{k} = 1\big| \{S_{i}\}_{i=1}^{k},\{A_{i}\}_{i=1}%
	^{k-1}\right)  = f\biggl(\frac{2B_{k-1}(S_{k})}{n_{k-1}(S_{k})}\biggr),
	\end{align*}
	where $f(\cdot):[-1,1] \mapsto[0,1]$ is a pre-specified non-increasing
	function satisfying $f(-x) = 1- f(x)$ and $f(x)$ is differentiable at $x=0$. Here, $\frac{B_{0}(S_{1})}{n_{0}%
		(S_{1})}$ and $B_{0}(S_{1})$ are understood to be zero.
\end{ex}

\begin{ex}
	[BCD] \label{ex:bcd}  The treatment status is determined sequentially for $1
	\leq k \leq n$ as
	\begin{align*}
	\mathbb{P}\left( A_{k} = 1| \{S_{i}\}_{i=1}^{k},\{A_{i}\}_{i=1}^{k-1}\right)
	=
	\begin{cases}
	\frac{1}{2} & \text{if }B_{k-1}(S_{k}) = 0\\
	\lambda & \text{if }B_{k-1}(S_{k}) < 0\\
	1-\lambda & \text{if }B_{k-1}(S_{k}) > 0,
	\end{cases}
	\end{align*}
	where $B_{k-1}(s)$ is defined as above and $\frac{1}{2}< \lambda\leq1$.
\end{ex}

\begin{ex}
	[SBR] \label{ex:sbr}  For each stratum, $\lfloor\pi_s n_s \rfloor$ units are
	assigned to treatment and the rest are assigned to control at random.
\end{ex}

\begin{rem}\label{rem:strong_balance}
	We note that SRS and SBR allow for the target faction of treatment to be different for different strata. For WEI and BCD, we have $\pi_s = 1/2$ for $s \in \mathcal{S}$. \cite{BCS17} further shows BCD and SBR achieve strong balance in the sense that $D_{n,s} = o_P(n^{1/2})$.    
\end{rem}

\section{Efficiency Trade-off in RAs with Many Regressors}
\label{sec:trade-off}
This section derives the joint asymptotic distribution of the fully saturated adjusted and unadjusted estimators under the assumption that the number of regressors grows no faster than the sample size. We then identify the efficiency trade-off in RAs based on the comparison of two asymptotic variances. The variance-covariance matrix will also be used to construct the optimal linear combination estimator in Section \ref{sec:optimal}.

The unadjusted estimator, denoted by $\hat \tau^{unadj}$, is the fully saturated estimator proposed by \cite{BCS18} and defined as
\begin{align}\label{eq:unadj_def}
\hat \tau^{unadj} = \frac{1}{n} \sum_{i \in [n]} \frac{A_iY_i}{\hat \pi_{S_i}} - \frac{1}{n} \sum_{i \in [n]} \frac{(1-A_i)Y_i}{1-\hat \pi_{S_i}} = \sum_{s \in \mathcal S} \hat p_s (\hat \tau_{1,s}^{unadj} - 
\hat \tau_{0,s}^{unadj}),
\end{align}
where for $s \in S$ and $a = 0,1$, $\hat \pi_s = n_{1,s}/n_s$ and $\hat \tau_{a,s}^{unadj} = \frac{1}{n_{a,s}} \sum_{i \in \aleph_{a,s}}Y_i$. 

Following their lead, the fully saturated regression adjusted estimator considered in this paper is denoted by $\hat \tau^{adj}$ and defined as 
\begin{align*}
\hat \tau^{adj} = \sum_{s \in \mathcal{S}} \hat p_s (\hat \tau_{1,s} - 
\hat \tau_{0,s}),
\end{align*}
where $\hat \tau_{a,s}$ is computed as the intercept in the OLS regression of $Y_i$ on $(1,\breve X_i)$ using all observations in $\aleph_{a,s}$, $\hat p_s = n_{s}/n$, and $\breve X_i = X_i - \overline{X}_{S_i}$ with $\overline{X}_s = \frac{1}{n_{s}}\sum_{i \in \aleph_s} X_i$. We further denote $\hat \beta_{a,s}$ as the OLS coefficient of $\breve X_i$ in the OLS regression. We note that $Y_i$ should not be demeaned in the above regression because our key parameter of interest is the intercept rather than the slope coefficient of $\breve X_i$. In fact, by the Frisch-Waugh-Lovell Theorem, the regression coefficient $\hat \beta_{a,s}$ is the same as if $X_i$ is demeaned by the cluster-treatment level mean ($\frac{1}{n_{a,s}}\sum_{i \in \aleph_{a,s}} X_i$) or not demeaned at all. Here, by demeaning $X_i$ with the cluster level mean, we effectively just change the intercept estimator $\hat \tau_{a,s}$ in the linear regression. 


\begin{rem} \label{rem:aipw}
	We consider both adjusted and unadjusted estimators for several reasons. First, \cite{BCS18} demonstrated that the unadjusted estimator remains consistent even when the treatment assignment probability, $\pi(s)$, varies across strata, and achieves the semiparametric efficiency bound in datasets without covariates, $X_i$. Second, the adjusted estimator, $\hat \tau^{adj}$, can be expressed as an augmented inverse propensity score weighted (AIPW) estimator:
	\begin{align}\label{eq:adj_def}
	\hat \tau^{adj}  & = \frac{1}{n} \sum_{i \in [n]} \frac{A_i(Y_i - X_i^\top \hat \beta_{1,S_i})}{\hat \pi_{S_i}} - \frac{1}{n} \sum_{i \in [n]} \frac{(1-A_i)(Y_i -  X_i^\top \hat \beta_{0,S_i})}{1-\hat \pi_{S_i}} + \frac{1}{n}\sum_{i \in [n]}X_i^\top (\hat \beta_{1,S_i} - \hat \beta_{0,S_i}),
	\end{align}
	where $\hat \pi_s = n_{1,s}/n_{s}$. It's important to note that if an individual, $i$, belongs to stratum $s$, then $\breve X_i$ is defined as $X_i$ demeaned by the stratum-level mean $\frac{1}{n_s}\sum_{i \in \aleph_s}X_i$, not the stratum-treatment level mean $\frac{1}{n_{a,s}}\sum_{i \in \aleph_{a,s}}X_i$. This distinction is crucial for the AIPW representation to hold. Comparing \eqref{eq:unadj_def} with \eqref{eq:adj_def}, we see that $\hat \tau^{adj}$ is a natural extension of $\hat \tau^{unadj}$ because the latter also follows the AIPW representation but with an empty set of covariates. This also implies that $\hat \tau^{unadj}$ does not use the information of $X_i$ at all, and therefore, neither benefits nor suffers from RAs. Third, while we follow \cite{BCS18} and fully saturate the regression, the violation of the overlapping support condition is less of a concern here because the treatment is assigned by the experimenter, and CARs can force strong balance of treated and control units in each stratum, as discussed in Remark \ref{rem:strong_balance}. 
\end{rem}

\begin{rem}
	In empirical applications, it is common to regress the outcome on strata fixed effects (SFEs), treatment status, and potentially additional covariates without full interaction between strata dummies and other regressors. However, we do not consider this regression for several reasons. First, even without additional covariates, \cite{BCS18} have shown that the SFE regression is inconsistent when the treatment assignment probability $\pi(s)$ is heterogeneous across strata. Moreover, when it is consistent, it is generally less efficient than the unadjusted estimator $\hat \tau^{unadj}$ unless the CAR achieves strong balance. Second, with a fixed number of covariates, the SFE regression is not always ``no-harm,'' whereas the adjusted estimator $\hat \tau^{adj}$ is, as demonstrated by \cite{YYS22}. Last, without full interactions, the SFE regression is unlikely to be (approximately) correctly specified, even when all covariates are discrete, unless the average treatment effects are homogeneous across strata. This introduces additional complexity in our analysis, particularly when the number of covariates is proportional to the sample size.
\end{rem}



\subsection{Asymptotic Properties}

Denote the dimension of $X$ as $k$. In the following, we derive the asymptotic properties of $\hat \tau^{adj}$ and $\hat \tau^{unadj}$ jointly in the case that $k=k_n$ increases no faster than the sample size $n$.

To clearly state our assumptions, we need to introduce more notation. Let $\breve X_{\aleph_{a,s}}$ be an $n_{a,s} \times k_n$ matrix which is constructed by stacking $\breve X_i^\top$ for $i \in \aleph_{a,s}$. We define $A_{\aleph_{a,s}} \in \Re^{n_{a,s}}$ similarly. Then, let $P_{a,s} = \breve X_{\aleph_{a,s}}  (\breve X_{\aleph_{a,s}}^\top\breve X_{\aleph_{a,s}})^{-1}\breve X_{\aleph_{a,s}}^\top$ be the projection matrix of $\breve X_{\aleph_{a,s}}$, and  $M_{a,s} = I_{n_{a,s}} - P_{a,s}$. Let $ M_{a,s, i,j}$ and $ P_{a,s, i,j}$ be the $(i,j)$th entry of $M_{a,s}$ and $ P_{a,s}$, respectively. Define
\begin{align*}
& \eps_{i}(a) = Y_i(a) - \mathbb{E}(Y_i(a)|X_i,S_i), \\
& \gamma_{a,s,n} = (1_{n_{a,s}}^\top M_{a,s} 1_{n_{a,s}})/n_{a,s},	\\
& \sigma^2_{a,s,n} = \frac{1}{n_{a,s}}\sum_{i \in \aleph_{a,s}} \left[ \left(\sum_{j \in \aleph_{a,s}} M_{a,s, i,j}\right)^2 \mathbb{E}(\eps_{i}^2(a)|X_i,S_i=s) \right], \quad \text{and}\\
& \rho_{a,s,n} = \frac{1}{n_{a,s}}\sum_{i \in \aleph_{a,s}} \left[ \left(\sum_{j \in \aleph_{a,s}} M_{a,s, i,j}\right) \mathbb{E}(\eps_{i}^2(a)|X_i,S_i=s) \right].
\end{align*}

\begin{ass}
	\begin{enumerate}[label=(\roman*)]
		\item $\max_{i \in [n]}\mathbb{E}\left[\eps_{i}^4(a)|X_i,S_i\right] =O_P(1).$ 
		\item There exists a constant $b>0$ such that $\min_{a =0,1, s \in \mathcal{S}, i \in [n]}\mathbb{E}\left[\eps_{i}^2(a)|X_i,S_i=s\right] \geq b.$
		\item $k_n/n_{a,s} \rightarrow \kappa_{a,s} \in \left[0, 1\right)$ and $\limsup_{n \rightarrow \infty}\max_{a=0,1, s\in \mathcal{S}, i \in [n]}P_{a,s,i,i} \leq 1-\delta$ almost surely for some constant $\delta \in (0,1)$.
		\item For $a=0,1$ and $s \in \mathcal{S}$, we have
		\begin{align*}
		\mathbb{E}(Y_i(a)|X_i,S_i=s) = \alpha_{a,s} + X_i^\top \beta_{a,s} + e_{i,s}(a),
		\end{align*}
		such that $\mathbb{E}(e_{i,s}^2(a)|S_i=s) = o(n^{-1})$. 
		\item $ \max_{a=0,1,s \in \mathcal{S}}\max_{i \in \aleph_{a,s}}\left|\sum_{j \in \aleph_{a,s}} M_{a,s, i,j}\right| = o_P(n^{1/2})$.
	\end{enumerate}
	\label{ass:linear}
\end{ass}

\begin{rem}\label{rem:ass2}
	Assumption \ref{ass:linear}(i) and \ref{ass:linear}(ii) are mild regularity conditions. Assumption \ref{ass:linear}(iii) implies that we allow the number of covariates to diverge at the rate of the sample size $n$. Because we run the stratum-treatment level regression with an effective sample size of $n_{a,s}$, we require $\kappa_{a,s} < 1$ to avoid multicollinearity. Assumption \ref{ass:linear}(v) is the same as \citet[Assumption 3]{CJN18} and \citet[Assumption 3]{J22},\footnote{In fact, $\hat v_{i,n}$ in both \citet[Assumption 3]{CJN18} and \citet[Assumption 3]{J22} is just $\sum_{j \in \aleph_{a,s}} M_{a,s, i,j}$.} which we recommend readers refer to for more discussion. 

\end{rem}

\begin{rem}
	Assumption \ref{ass:linear}(iv) implies that the linear regression is approximately correctly specified, a common assumption in analyses with many regressors, as seen in \citet[Assumption 3]{CJN18} and \citet[Assumption 3]{J22}. \cite{KSS2020} further require Assumption \ref{ass:linear}(iv) to hold with no approximation error, i.e., $e_{i,s}(a) = 0$. This condition is reasonable because the dimension $k_n$ is allowed to diverge to infinity. For example, if baseline covariates $Z_i$ are fixed-dimensional and continuous, then Assumption \ref{ass:linear}(iv) holds when $X_i$ contains sieve basis functions of $Z_i$, and $\mathbb{E}(Y_i(a)|Z_i=z,S_i=s)$ is sufficiently smooth in $z$.\footnote{In this case, we have $\mathbb E(Y_i(a)|Z_i,S_i=s) = \alpha_{a,s} + X_i^\top \beta_{a,s} + \tilde e_{i,s}(a)$ such that $\mathbb E (\tilde e_{i,s}^2(a)|S_i=s) = o(n^{-1})$. Then, we have $\mathbb E(Y_i(a)|X_i,S_i=s) = \alpha_{a,s} + X_i^\top \beta_{a,s} + e_{i,s}(a)$ where $e_{i,s}(a) = \mathbb E(\tilde e_{i,s}(a)|X_i,S_i=s)$. Then, by Jensen's inequality, we have $\mathbb E (e_{i,s}^2(a)|S_i=s) \leq \mathbb E (\tilde e_{i,s}^2(a)|S_i=s) = o(n^{-1})$.} We further detail the sieve bases and smoothness requirement implied by Assumption \ref{ass:linear}(iv) in the next remark. If all baseline covariates are categorical, Assumption \ref{ass:linear}(iv) holds when $X_i$ contains the fully saturated dummies for all categories. The approximately correct specification is also commonly assumed in the literature on causal inference with ultra-high-dimensional data, where researchers further assume the specification is sparse. See, for example, Assumption 4.2 in the seminal work by \cite{BCFH13}. While we do not consider the ultra-high dimensionality (the dimension is at most a fraction of the sample size), we do not require sparsity either. In cases when the identity of $X_i$ is predetermined by the researcher and its dimension is fixed, the true specification might not be well approximated by a linear function of $X_i$, thus potentially violating Assumption \ref{ass:linear}(iv). However, in Section \ref{sec:fixed}, we show that in this case with a fixed or moderately diverging $k$, our estimator and inference methods remain valid even if the linear regression is misspecified and the approximation error is not asymptotically negligible. Remark \ref{rem:specification} provides more discussion on this point. 
\end{rem}

\begin{rem}
	Here we provide more details on sieve bases and smoothness requirement implied by Assumption \ref{ass:linear}(iv). Suppose we have a finite $d_z$-dimensional regressor $Z \in \Re^{d_z}$ and $X = (b_{1,n}(Z), \cdots, b_{h_n,n}(Z))^\top$ is a $h_n$-dimensional vector that contains sieve bases of $Z$, where $\{b_{h,n }(\cdot)\}_{h \in [h_n]}$ are $h_n$ basis functions of a linear sieve space, denoted as $\mathbb{B}$.  Given that all $d_z$ elements of $Z$ are continuously distributed, the sieve space $\mathbb{B}$ can be constructed as follows.
	\begin{enumerate}
		\item For each element $Z^{(l)}$ of $Z$, $l=1,\cdots,d_z$, let $\mathcal{B}_l$ be the univariate sieve space of dimension $J_n$. 
		\item Let $\mathbb{B}$ be the tensor product of $\{\mathcal{B}_l\}_{l=1}^{d_x}$, which is defined as a linear space spanned by the functions $\prod_{l=1}^{d_x} b_l$, where $b_l \in \mathcal{B}_l$. The dimension of $\mathbb{B}$ is then $k_n \equiv d_z J_n$.
	\end{enumerate}
	
	We provide two examples of sieve space $\mathcal B_l$ below. More examples can be found in Section 2.3 of \cite{C07_sieve}.  
	\begin{enumerate}
		\item Polynomials. 
		$$\mathbb{B}_l = \biggl\{\sum_{j=0}^{J_n}\alpha_j z^j, z \in \Supp(Z^{(l)}), \alpha_j \in \Re \biggr\};$$
		\item Splines. 
		$$\mathbb{B}_l = \biggl\{\sum_{t=0}^{r-1}\alpha_t z^t + \sum_{j=1}^{J_n}c_j[\max(z-q_j,0)]^{r-1}, z \in \Supp(Z^{(l)}), \alpha_t, c_j \in \Re \biggr\},$$
		where the grid $-\infty=q_0 \leq q_1 \leq \cdots \leq q_{J_n} \leq q_{J_n+1} = \infty$ partitions $\Supp(Z^{(l)})$ into $J_n+1$ subsets $I_j = [q_j,q_{j+1}) \cap \Supp(Z^{(l)})$, $j=1,\cdots,J_n-1$, $I_{0} = (q_{0},q_{1}) \cap \Supp(Z^{(l)})$, and $I_{J_{n}} = (q_{J_n},q_{J_n+1}) \cap \Supp(Z^{(l)})$. 
	\end{enumerate}

	Suppose the function $f_{a,s}(z) = \mathbb{E}(Y(a) \mid Z=z, S=s)$ is $p$-th order smooth and satisfies other regularity conditions. Then, the approximation error satisfies $\mathbb{E}(e_{i,s}^2(a) \mid S_i=s) = O(k_n^{-2p/d_z})$ (see Section 2.3 of \cite{C07_sieve}). In the existing literature, the number of sieve bases typically follows $k_n = o(n^{1/2})$. For the approximation error to meet $\mathbb{E}(e_{i,s}^2(a) \mid S_i=s) = o(n^{-1})$, the smoothness order $p$ must satisfy $p > d_z$. In our case, we allow $k_n$ to be proportional to $n$, which relaxes the smoothness requirement to $p > d_z/2$. A similar observation was made by \cite{CJN18} in Section 4.
\end{rem}

\begin{rem}
	\cite{LD21} and \cite{CMO23} consider the inference of ATE under complete randomization in the finite-population setup without assuming approximate correct specification. However, for the validity of their inference methods, they require $k_n = o(n^{2/3})$, and thus, $\kappa_{a,s} = 0$. We provide more detail discussion about the dimensionality of covariates and the correct specification assumption in Remark \ref{rem:specification} in Section \ref{sec:fixed}. 
\end{rem}

\begin{ass}
	Suppose $\gamma_{a,s,n} \convP \gamma_{a,s,\infty}>0$,  $\gamma_{a,s,n}^{-2}\sigma^2_{a,s,n} \convP  \omega^2_{a,s,\infty}>0$, and $\gamma_{a,s,n}^{-1}\rho_{a,s,n} \convP \varpi_{a,s,\infty}$ for some deterministic constants  $(\gamma_{a,s,\infty},\omega_{a,s,\infty},\varpi_{a,s,\infty})$.
	\label{ass:omega}
\end{ass}

\begin{rem}\label{rem:ass3}
	If $k_n \log k_n = o(n)$ (which means $\kappa_{a,s} = 0$ for all $(a,s)$), then under general conditions on the distribution of $X_i$, it is possible to show that Assumption \ref{ass:omega} holds with $\gamma_{a,s,\infty} = 1$ and $ \omega^2_{a,s,\infty} = \varpi_{a,s,\infty} = \mathbb E (\eps_i^2(a)|S_i=s)$. See \cite{belloni2015} and \cite{CFF20} for more discussion and examples. If $\kappa_{a,s}>0$, the eigenvalues of Gram matrix $n^{-1}\breve X_{\aleph_{a,s}}^\top\breve X_{\aleph_{a,s}}$ do not converge to those of its expectation. Instead, they will converge to some spectral density which determines the quantities $(\gamma_{a,s,\infty},\omega_{a,s,\infty},\varpi_{a,s,\infty})$. Specifically, in Section \ref{sec:app_omega} in the Online Supplement, we provide detailed calculation of $\gamma_{a,s,\infty}$ when the $k_n \times 1$ vector $X_i$ given $S_i=s$ is jointly Gaussian with a covariance matrix $\Sigma_{s,n}$. Based on the Mar\v{c}enko-Pastur theorem (\citep{MP67}), we can show that 
	\begin{align*}
	\gamma_{a,s,\infty} = \frac{1}{1+ (a(1-\pi_s) + (1-a)\pi_s)\zeta_{a,s}}< 1,
	\end{align*}
	where $\zeta_{a,s} = \int_{\lambda_-}^{\lambda_+} \frac{\sqrt{(\lambda_+ - \lambda)(\lambda - \lambda_-)}}{2\pi\lambda^2} d\lambda$ and $\lambda_{\pm} = (1 \pm \sqrt{\kappa_{a,s}})^2$. If we further assume homoskedasticity in the sense that $\mathbb E (\eps_i^2(a)|X_i,S_i=s) = \mathbb E (\eps_i^2(a)|S_i=s)$, then, we have 
	\begin{align*}
	\omega^2_{a,s,\infty} = \varpi_{a,s,\infty} =  \gamma_{a,s,\infty}^{-1} \mathbb E (\eps_i^2(a)|S_i=s)>\mathbb E (\eps_i^2(a)|S_i=s).
	\end{align*}
	We conjecture that Assumption \ref{ass:omega} holds for non-Gaussian distributions of $X_i$ and heteroskedastic errors. However, establishing general preliminary conditions for this assumption requires the use of advanced tools from random matrix theory.\footnote{In random matrix theory, many delicate asymptotic properties of Gaussian ensembles were later shown to hold for much broader classes of random matrices with non-Gaussian entries. For example, see the Mar\v{c}enko-Pastur theorem (\cite{MP67}), the semicircular law (\citet[Theorem 2.4.2]{T12}), and \cite{bai08} for a comprehensive survey.} Such a discussion is beyond the scope of this paper. Finally, we emphasize that researchers do not need to know the exact values of $(\gamma_{a,s,\infty},\omega^2_{a,s,\infty},\varpi_{a,s,\infty})$ to implement our inference method.
\end{rem}

\begin{thm}
	Suppose Assumptions \ref{ass:assignment1}--\ref{ass:omega} hold. Then, we have
	\begin{align*}
	\sqrt{n}\begin{pmatrix}
	\hat \tau^{adj} - \tau \\
	\hat \tau^{unadj} - \tau
	\end{pmatrix} \convD \begin{pmatrix}
	\mathcal U^{adj}  \\
	\mathcal U^{unadj}  \\
	\end{pmatrix} + \begin{pmatrix}
	\mathcal V^{adj} \\
	\mathcal V^{unadj} \\
	\end{pmatrix}+\begin{pmatrix}
	\mathcal W \\
	\mathcal W
	\end{pmatrix},
	\end{align*}
	where  
	\begin{align*}
	\mathcal U = \begin{pmatrix}
	\mathcal U^{adj}  \\
	\mathcal U^{unadj}  \\
	\end{pmatrix} \stackrel{d}{=} \N\left(\begin{pmatrix}
	0 \\
	0
	\end{pmatrix}, \Sigma_{\mathcal U} \right),~\Sigma_{\mathcal U} = \begin{pmatrix}
	\mathbb{E}\left(\frac{\omega_{1,S_i,\infty}^2}{\pi_{S_i}}+\frac{\omega_{0,S_i,\infty}^2}{1-\pi_{S_i}}\right) & 	\mathbb{E}\left(\frac{\varpi_{1,S_i,\infty}}{\pi_{S_i}}+\frac{\varpi_{0,S_i,\infty}}{1-\pi_{S_i}}\right)\\
	\mathbb{E}\left(\frac{\varpi_{1,S_i,\infty}}{\pi_{S_i}}+\frac{\varpi_{0,S_i,\infty}}{1-\pi_{S_i}}\right) &  \mathbb{E}\left(\frac{\mathbb{E}(\eps_i^2(1)|S_i)}{\pi_{S_i}}+\frac{\mathbb{E}(\eps_i^2(0)|S_i)}{1-\pi_{S_i}}\right)
	\end{pmatrix}, 
	\end{align*}
	\begin{align*}
	\mathcal V = \begin{pmatrix}
	\mathcal V^{adj}  \\
	\mathcal V^{unadj}  \\
	\end{pmatrix}  \stackrel{d}{=}  \N\left(\begin{pmatrix}
	0 \\
	0
	\end{pmatrix},\Sigma_{\mathcal V} \right),~\Sigma_{\mathcal V} = \begin{pmatrix}
	\mathbb{E}	var(\phi_i(1) - \phi_i(0)|S_i) & \mathbb{E}	var(\phi_i(1) - \phi_i(0)|S_i)\\
	\mathbb{E}	var(\phi_i(1) - \phi_i(0)|S_i) & \mathbb{E}\left[\frac{var(\phi_i(1)|S_i)}{\pi_{S_i}} + \frac{var(\phi_i(0)|S_i)}{(1-\pi_{S_i})}\right]
	\end{pmatrix}, 
	\end{align*}
	\begin{align*}
	\mathcal W \stackrel{d}{=}\N(0, \Sigma_{\mathcal W}),~\Sigma_{\mathcal W} = var(\mathbb{E}(Y_i(1) - Y_i(0)|S_i)),
	\end{align*}
	$\phi_i(a) = \mathbb{E}(Y_i(a)|X_i,S_i)-\mathbb{E}(Y_i(a)|S_i)$ for $a=0,1$, and $(U,V,W)$ are independent.
	\label{thm:main}
\end{thm}

\begin{rem}\label{rem:main_theorem}
	Theorem \ref{thm:main} decomposes the variance of ATE estimators into $\Sigma_{\mathcal U}$, $\Sigma_{\mathcal V}$, and $\Sigma_{\mathcal W}$, which represent the variation of $Y(a)$ given $X_i$ and $S_i$, the variation of $X_i$ given $S_i$, and the variation of $S_i$, respectively. Denote  $\Sigma = \Sigma_{\mathcal U} + \Sigma_{\mathcal V} + \Sigma_{\mathcal W}1_21_2^\top$. We will focus on the variance components of $\Sigma$ to identify the efficiency trade-off in RAs herein. The whole matrix will be used to construct the optimal linear combination estimator in Section \ref{sec:optimal}.
	
\end{rem}

\begin{rem}\label{rem:rate}
	We establish a Berry-Esseen-type bound for the convergence of \(\tilde{\mathcal{U}}_n\) to \(\mathcal{U}\), where \(\tilde{\mathcal{U}}_n\) represents the part of \(\sqrt{n}(\hat{\tau}^{adj} - \tau, \hat{\tau}^{unadj} - \tau)'\) that captures the variation in \((Y(1), Y(0))\) given \((X, S)\). It is also standard to establish a Berry-Esseen-type bound for the convergence of the components $\mathcal V_n$ and $\mathcal W_n$, which accounts for the variation in \(X_i\) given \(S_i\) and the variation in \(S_i\), respectively. By combining these bounds, it is possible to determine the rate of Gaussian approximation for the distribution given in Theorem \ref{thm:main}, which is \(n^{-1/6}\) under standard moment conditions.
\end{rem}

\subsection{Efficiency Trade-off}
\label{subsec:trade-off}
We note that the $(1,1)$ and $(2,2)$ elements of $\Sigma_{\mathcal U}$ represent the asymptotic variances of $\mathcal U^{adj}$ and $\mathcal U^{unadj}$, respectively, which we denote as $Var(\mathcal U^{adj})$ and $Var(\mathcal U^{unadj})$. We denote $Var(\mathcal V^{adj})$ and $Var(\mathcal V^{unadj})$ in the same manner. Following the argument in Remark \ref{rem:ass3}, we can show that $Var(\mathcal U^{adj}) = Var(\mathcal U^{unadj})$ if $k_n \log k_n = o(n)$. Moreover, it can be shown that $Var(\mathcal V^{adj}) \leq Var(\mathcal V^{unadj})$, which indicates the benefit of RA from incorporating information of $X_i$. Therefore, when $k_n \log k_n = o(n)$, $\hat \tau^{adj}$ is weakly more efficient than $\hat \tau^{unadj}$, a result that the previous literature has established under more stringent conditions.

However, when $k_n$ is of the same order of $n$, $Var(\mathcal U^{adj})$ can be larger than $Var(\mathcal U^{unadj})$.  To illustrate this, let's continue with the Gaussian covariates example mentioned in Remark \ref{rem:ass3}. Suppose  for $a=0,1$ and $s\in \mathcal S$, $\pi_s = 1/2$, $\kappa_{a,s} = \kappa$, and $\eps_i(a)$ is homoskedastic in the sense that $\mathbb E (\eps_i^2(a)|X_i,S_i=s) = \mathbb E (\eps_i^2(a)|S_i=s)$, then we have
\begin{align*}
\text{Variance Inflation Factor}   \equiv \frac{\text{Var}(\mathcal U^{adj}) - \text{Var}(\mathcal U^{unadj})}{\text{Var}(\mathcal U^{unadj})} = \frac{\zeta}{2},
\end{align*}
where $\zeta = \int_{\lambda_-}^{\lambda_+} \frac{\sqrt{(\lambda_+ - \lambda)(\lambda - \lambda_-)}}{2\pi\lambda^2} d\lambda$ and $\lambda_{\pm} = (1 \pm \sqrt{\kappa})^2$.     
Figure \ref{fig:vif} plots the values of the variance inflation factor (VIF) as a function of $\kappa$. We can see that the VIF is more than $12.5\%$, $25\%$, $50\%$, and $100\%$ if $\kappa$ is greater than $1/5$, $1/3$, $1/2$, $2/3$, respectively. 

\begin{figure}[t]
	\centering
	\includegraphics[width=0.9\textwidth]{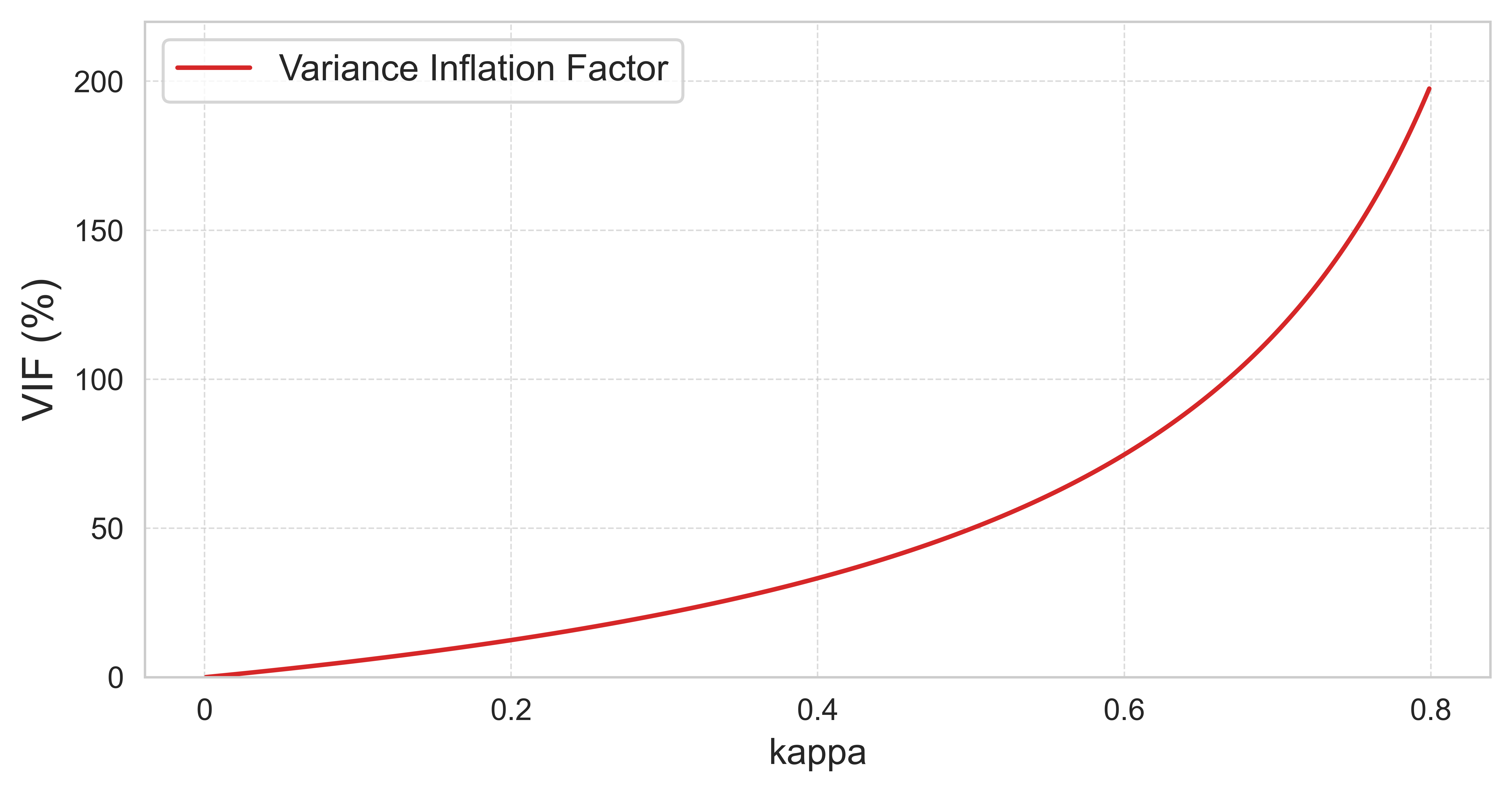}
	\caption{The Variance Inflation Factor}
	\label{fig:vif}
\end{figure}

Consequently, an efficiency trade-off in RAs between $Var(\mathcal U^{adj})$ and $Var(\mathcal V^{adj})$ can be identified in Theorem \ref{thm:main}. This trade-off can be elucidated by decomposing the  difference between the asymptotic variances of $\hat \tau^{adj}$ and $\hat \tau^{unadj}$ as 
\begin{align*}
(var(\mathcal U^{adj}) - var(\mathcal U^{unadj}))+ (var(\mathcal V^{adj}) - var(\mathcal V^{unadj})),
\end{align*}
where the second term is always non-positive but the first term can be positive when $k_n$ is of the same order of the sample size. Under some circumstances, the first term may dominate the second, resulting in the estimator with the theoretically ``no-harm'' adjustment being even less efficient than the unadjusted estimator. For example, see Model 1 of Panel (b) in Figures \ref{fig:sbr400} and \ref{fig:sbr800}. This defeats the purpose of using RAs. 

This also implies that the asymptotic variance estimator of $\hat \tau^{adj}$ in the literature, e.g., the one proposed by \cite{YYS22}, tends to under-estimate the true asymptotic variance when there are many regressors, as it neglects the VIF induced by the dimensionality. Therefore, when there are many regressors, the causal inference using such a variance estimator may over-reject under the null hypothesis, as illustrated in Section \ref{sec:simulations}.

\section{Optimal Linear Combination Estimator}
\label{sec:optimal}

As shown above, ignoring the VIF can result in size distortion in causal inference. When we account for the VIF, however, the adjusted estimator may become even less efficient than the unadjusted one due to the estimation errors incurred in RAs with many regressors. In this section, we propose an optimal linear combination estimator, denoted by $\hat \tau^*$, designed to address both issues. We show that $\hat \tau^*$ is consistent, asymptotically normal, and weakly more efficient than both $\hat \tau^{adj}$ and $\hat \tau^{unadj}$. We also present the local asymptotic power analysis for the standard Wald test based on $\hat \tau^*$. Finally, we provide a consistent estimator $\hat \Sigma$ for $\Sigma = \Sigma_{\mathcal U} + \Sigma_{\mathcal V} + \Sigma_{\mathcal W}1_21_2^\top$ in Section \ref{sec:var} below.

\subsection{Asymptotic Properties of $\hat \tau^*$}

We construct the optimal linear combination estimator as 
\begin{align}\label{eq:taustar}
\hat \tau^* = \hat w \hat \tau^{adj} + (1-\hat w)\hat \tau^{unadj},    
\end{align}
which is a weighted average of $\hat \tau^{adj}$ and $\hat \tau^{unadj}$ with the weight $\hat \omega$ minimizing the asymptotic variance. 

Denote $\hat \Sigma = \begin{pmatrix}
\hat  \Sigma_{1,1} & \hat \Sigma_{1,2} \\
\hat \Sigma_{1,2} & \hat \Sigma_{2,2}
\end{pmatrix}$  as a consistent estimator for $\Sigma \equiv \Sigma_{\mathcal U} + \Sigma_{\mathcal V} + \Sigma_{\mathcal W}1_21_2^\top = \begin{pmatrix}
\Sigma_{1,1} & \Sigma_{1,2} \\
\Sigma_{1,2} & \Sigma_{2,2}
\end{pmatrix}$  and suppose $\hat \tau^{adj}$ and $\hat \tau^{unadj}$ are not asymptotically equivalent, i.e., $ \Sigma_{1,1} +  \Sigma_{2,2} - 2  \Sigma_{1,2} > 0$. Then, the optimal weight is  
\begin{align}\label{eq:omegastar}
\hat w = \frac{\hat \Sigma_{2,2} - \hat \Sigma_{1,2}}{\hat \Sigma_{1,1} + \hat \Sigma_{2,2} - 2 \hat \Sigma_{1,2}},
\end{align}
and the corresponding estimator of the asymptotic variance for $\hat \tau^*$ is $\hat \sigma^2_* = (\hat w, 1- \hat w) \hat \Sigma (\hat w, 1- \hat w)^\top$. 

\begin{rem}\label{rem:lambda}
	When $\Sigma_{1,1} - 2\Sigma_{1,2} + \Sigma_{2,2}>0$, the adjusted estimator $\hat \tau^{adj}$ and the unadjusted estimator $\hat \tau^{unadj}$ are not asymptotically equivalent. In this case, $(w,1-w)\Sigma (w,1-w)^\top$ has a unique minimizer $w^* = \frac{\Sigma_{2,2} - \Sigma_{1,2}}{\Sigma_{1,1} - 2\Sigma_{1,2} + \Sigma_{2,2}}$, and our $\hat w$ is a consistent estimator of $w^*$. If the estimation error of the RA is asymptotically negligible and $X_i$ does not contain useful information of $(Y_i(1),Y_i(0))$ in the sense that 
	\begin{align*}
	\mathbb E(Y_i(a)|X_i,S_i) = \mathbb E(Y_i(a)|S_i), 
	\end{align*}
	then we have $\Sigma_{1,1} - 2\Sigma_{1,2} + \Sigma_{2,2}=0$ (which implies $\Sigma_{1,1} = \Sigma_{1,2} = \Sigma_{2,2}$). In this case, all linear combinations are asymptotically equivalent. If this case is a realistic concern, then we can always safeguard against the zero denominator by introducing a small positive constant $\lambda$ and constructing the final estimator as $\tilde \tau = \tilde w \hat \tau^{adj} + (1-\tilde w)\hat \tau^{unadj}$ with $ \tilde w = \frac{\hat \Sigma_{2,2} - \hat \Sigma_{1,2}}{\hat \Sigma_{1,1} - 2\hat \Sigma_{1,2} + \hat \Sigma_{2,2} + \lambda}$. It is possible to show that $\tilde \tau$ is always weakly more efficient than $\hat \tau^{unadj}$ as long as $\lambda > 0$, regardless of whether $\Sigma_{1,1} - 2\Sigma_{1,2} + \Sigma_{2,2}=0$ or $\Sigma_{1,1} - 2\Sigma_{1,2} + \Sigma_{2,2}>0$. 
\end{rem}


\begin{thm}
	Suppose Assumptions \ref{ass:assignment1}--\ref{ass:omega} hold. If $\Sigma_{1,1} - 2\Sigma_{1,2} + \Sigma_{2,2}>0$ and $\hat \Sigma$ is an consistent estimator for $\Sigma$, then we have 
	\begin{align*}
	\sqrt{n}\left[ (\hat w, 1- \hat w) \hat \Sigma (\hat w, 1- \hat w)^\top\right]^{-1/2}(\hat \tau^* - \tau) \convD \N(0, 1),
	\end{align*}
	where $\hat \tau^*$ and $\hat \omega$ are defined in \eqref{eq:taustar} and \eqref{eq:omegastar}, respectively. In this case, $\hat \tau^*$ is asymptotically weakly more efficient than both $\hat \tau^{adj}$ and $\hat \tau^{unadj}$.
	\label{thm:further}
\end{thm}


\subsection{Local Power}
Suppose we want to test the null hypothesis $H_{0}:\tau = \tau_0$ versus the two-sided alternative $H_{1}:\tau \neq \tau_0$ with level $\alpha \in (0,1)$, and we are under the local alternative in the sense that $\tau = \tau_0 + \Delta/\sqrt{n}$. Then, the hypothesis can be rewritten as $H_0: \Delta = 0$ versus $H_1: \Delta \neq 0$. In addition, by Theorem \ref{thm:main}, we have 
\begin{align*}
\sqrt{n} \Sigma^{-1/2}\begin{pmatrix}
(\hat \tau^{adj} - \tau_0) \\
(\hat \tau^{unadj} - \tau_0) \\
\end{pmatrix} \convD \N\left( \begin{pmatrix}
a\Delta \\
b\Delta
\end{pmatrix}, I_2 \right),
\end{align*}
where $\begin{pmatrix}
a \\
b
\end{pmatrix} = \Sigma^{-1/2} \begin{pmatrix}
1\\
1
\end{pmatrix}$. Let us consider a limit experiment in which a researcher observes 
\begin{align}\label{eq:N12}
\begin{pmatrix}
N_1 \\
N_2
\end{pmatrix}   \stackrel{d}{=}\N\left( \begin{pmatrix}
a\Delta \\
b\Delta
\end{pmatrix}, I_2 \right),
\end{align}
knows the values of $(a,b)$, and wants to test $\Delta = 0$ versus the two-sided alternative. Then, by the factorization criterion, it is easy to see that $N^* = \frac{aN_1+bN_2}{\sqrt{a^2+b^2}} \sim N(\sqrt{a^2+b^2}\Delta,1)$ is a sufficient statistic for $\Delta$.  
Define the usual two-sided Wald test based on $\hat \tau^*$ as
\begin{align}\label{eq:Wald}
\mathbb W_n = 1\left\{n \left[(\hat w, 1-\hat w)\hat \Sigma (\hat w, 1-\hat w)^\top\right]^{-1} ( \hat \tau^* - \tau_0)^2 \geq \mathcal{C}_\alpha \right\},    
\end{align}
where $\mathcal{C}_{\alpha}$ is the $(1-\alpha)$ quantile of the chi-squared distribution with one degree of freedom. Then, the following corollary implies $\mathbb W_n \convD 1\{ (N^*)^2 \geq \mathcal C_\alpha\}$. In addition, $1\{ (N^*)^2 \geq \mathcal C_\alpha\}$ is known as the uniformly most powerful (UMP) unbiased test (see, for example, \citet[Section 4.2]{LR06}) in the sense that it is more powerful than any level $\alpha$ unbiased test that is a (potentially nonlinear) combination of $(N_1,N_2)$. Consequently, it means $\mathbb W_n$ is asymptotically the UMP test over a class of level $\alpha$ unbiased tests that only depend on 
\begin{align}\label{eq:hatN12}
\begin{pmatrix}
\hat N_1 \\
\hat N_2
\end{pmatrix}   = \hat \Sigma^{-1/2} \begin{pmatrix}
\sqrt{n}(\hat \tau^{adj}-\tau_0) \\
\sqrt{n}(\hat \tau^{unadj}-\tau_0)  
\end{pmatrix}.
\end{align}
This implies $\hat \tau^*$ is optimal among not only \textit{linear} but also \textit{nonlinear} combination estimators. 

To formalize the statement, we need to introduce the following notation. Suppose $(N_1,N_2)$ follow the joint distribution in \eqref{eq:N12} and let  
\begin{align*}
\Psi_{\alpha}^U = \begin{Bmatrix}
\psi(\cdot): \mathbb{E}\psi(N_1,N_2) \leq \alpha \quad \text{if $\Delta = 0$}, \\
\mathbb{E}\psi(\N_1,\N_2) \geq \alpha \quad \text{if $\Delta \neq 0$}, \\
\text{the discontinuities of $\psi(\cdot)$ has zero Lebesgue measure} 
\end{Bmatrix}    
\end{align*}
be a class of unbiased tests with level $\alpha$.

\begin{cor}\label{cor:umpu}
	Suppose all conditions in Theorem \ref{thm:further} hold, and we are under the local alternative in the sense that $\tau = \tau_0 + \Delta/\sqrt{n}$. Then, we have
	\begin{align*}
	\sqrt{n}\left[ (\hat w, 1- \hat w) \hat \Sigma (\hat w, 1- \hat w)^\top\right]^{-1/2}(\hat \tau^* - \tau_0) \convD N^* \stackrel{d}{=}\N(\sqrt{a^2+b^2}\Delta,1),
	\end{align*}
	where $\begin{pmatrix}
	a \\
	b
	\end{pmatrix} = \Sigma^{-1/2} \begin{pmatrix}
	1\\
	1
	\end{pmatrix}$. In addition,       suppose $\breve{\psi}_n$ is a generic test such that $\breve{\psi}_n = \psi(\hat N_1,\hat N_2) +o_P(1)$ for some $\psi \in \Psi_{\alpha}^U$ and the sequence $\{\breve{\psi}_n \}_{n \geq 1}$ is uniformly integrable, where $(\hat N_1,\hat N_2)$ are defined in \eqref{eq:hatN12}. Then, we have
	\begin{align*}
	\lim_{n \rightarrow \infty}    \mathbb{E}\mathbb W_n = \sup_{\psi \in \Psi_{\alpha}^U } \lim_{n \rightarrow \infty} \mathbb{E}\psi(\hat N_1,\hat N_2) \geq  \lim_{n \rightarrow \infty} \mathbb{E}\breve{\psi}_n. 
	\end{align*}
\end{cor}

\begin{rem}
	In Corollary 4.1, we show that the two-sided Wald test based on the estimator $\hat \tau^*$, denoted by $\mathbb W_n$, is uniformly most powerful against two-sided local alternatives within the class of tests $\Psi_\alpha^U$. Thus, the optimality of $\mathbb W_n$ is specifically in terms of local power. The class $\Psi_\alpha^U$ consists of level $\alpha$ unbiased tests constructed using both adjusted and unadjusted estimators ($\hat \tau^{adj}$ and $\hat \tau^{unadj}$). Similar methods are used in the weak-identification robust inference literature, where combining Anderson-Rubin and Lagrange multiplier tests yields an unbiased test that maintains size control under weak identification and is uniformly most powerful against two-sided alternatives when identification is strong. For examples, see \cite{Andrews(2016)} and \cite{LWZ24}.
\end{rem}

\begin{rem}\label{rem:three}
	Our inference procedure is ``optimal'' within the class of unbiased tests constructed by both adjusted and unadjusted estimators ($\hat \tau^{adj}$ and $\hat \tau^{unadj}$). However, further improvement is possible by combining more than two estimators. For instance, researchers may have prior knowledge about a subset of adjustment regressors, denoted as $X_i^{prior}$, and use them to form another regression-adjusted estimator $\hat \tau^{prior}$.\footnote{It can be computed by replacing $X_i$ in our estimation procedure by $X_i^{prior}$.} It is feasible to optimally combine $(\hat \tau^{adj},\hat \tau^{unadj},\hat \tau^{prior})$ in the same manner. Specifically, suppose we are under the same local alternative and have 
	\begin{align*}
	\sqrt{n} \Sigma^{-1/2} \begin{pmatrix}
	\hat \tau^{adj} - \tau_0 \\
	\hat \tau^{unadj} - \tau_0 \\
	\hat \tau^{prior} - \tau_0 \\
	\end{pmatrix} \convD \begin{pmatrix}
	N_1 \\
	N_2 \\
	N_3 
	\end{pmatrix} \stackrel{d}{=} \N\left( \begin{pmatrix}
	a \Delta \\
	b \Delta \\
	c \Delta
	\end{pmatrix}, I_3 \right),
	\end{align*}
	where $(a,b,c)^\top = \Sigma^{-1/2} 1_3$, then following the same argument before Corollary \ref{cor:umpu}, the level-$\alpha$ UMP unbiased test in the limit experiment can be written as 
	\begin{align*}
	1\left\{ \frac{(aN_1 + bN_2 + cN_3)^2}{a^2+b^2+c^2} \geq \mathcal C_\alpha  \right\}.
	\end{align*}
	We can then construct a consistent cross-fit estimator $\hat \Sigma$ for the $3 \times 3$ covariance matrix $\Sigma$ following the same strategy proposed in Section \ref{sec:var} below, which leads to consistent estimators $(\hat a,\hat b, \hat c)$ for $(a,b,c)$. Further denote 
	\begin{align*}
	\sqrt{n} \hat \Sigma^{-1/2} \begin{pmatrix}
	\hat \tau^{adj} - \tau_0 \\
	\hat \tau^{unadj} - \tau_0 \\
	\hat \tau^{prior} - \tau_0 \\
	\end{pmatrix}   = \begin{pmatrix}
	\hat  N_1 \\
	\hat   N_2 \\
	\hat   N_3 
	\end{pmatrix}. 
	\end{align*}
	Then, it can be shown that 
	\begin{align*}
	1\left\{ \frac{(\hat a \hat N_1 + \hat b \hat N_2 + \hat c \hat N_3)^2}{\hat a^2+ \hat b^2+ \hat c^2} \geq \mathcal C_\alpha   \right\}
	\end{align*}
	is asymptotically UMP unbiased test, as detailed in Corollary \ref{cor:umpu}, and thus, is uniformly more powerful than the two-sided t-tests based on $(\hat \tau^{adj}, \hat \tau^{unadj},\hat \tau^{prior})$ and $\hat \tau^*$, which does not use the prior information.
	
	Theoretically, the same procedure can extend to optimally combine any fixed number of regression-adjusted estimators. However, two problems arise. Firstly, combining multiple estimators may make the covariance matrix nearly singular and cause numerical issues. This creates another interesting trade-off between the negative effect due to the estimation error of the large covariance matrix and the benefit of including more estimators. Studying the optimal number of estimators and their identities in our estimator combination procedure would be interesting. 
	
	Secondly, using data-driven methods (e.g., Lasso) to choose significant regressors may cause size distortion in the inference, especially when regressors have weak but dense effects, as in Models 1-3 in our simulation. It means even there is a well-defined globally optimal combination estimator, it may not be achievable. Due to these problems, we focus on an inference procedure that combines $\hat \tau^{adj}$ and $\hat \tau^{unadj}$ only. 
	
\end{rem}

\subsection{Covariance Matrix Estimator}
\label{sec:var}
We now propose a consistent estimator $\hat \Sigma$ for $\Sigma =\Sigma_{\mathcal U} + \Sigma_{\mathcal V} + \Sigma_{\mathcal W}1_21_2^\top$. We first note that 
$$\Sigma_{2,2} = \mathbb E\left[\frac{var(Y_i(1)|S_i)}{\pi_{S_i}} + \frac{var(Y_i(0)|S_i)}{1-\pi_{S_i}}\right] + \Sigma_{\mathcal W}$$ is the asymptotic variance of the unadjusted estimator and does not depend on $X_i$. Now let
\begin{align*}
\hat \Sigma_{2,2} =& \frac{1}{n} \sum_{s \in \mathcal{S}}
\left[
\sum_{i \in \aleph_{1,s}} \left(\frac{Y_i}{\hat{\pi}_s} - \frac{1}{n_{1,s}}\sum_{i \in \aleph_{1,s}}\frac{Y_i}{\hat{\pi}_s} \right)^2 
+
\sum_{i \in \aleph_{0,s}} \left(\frac{Y_i}{1-\hat{\pi}_s} - \frac{1}{n_{0,s}}\sum_{i \in \aleph_{0,s}}\frac{Y_i}{1-\hat{\pi}_s} \right)^2 
\right]
\\
&+ \sum_{s \in \mathcal{S}}\hat p_s\left(\frac{1}{n_{1,s}}\sum_{i \in \aleph_{1,s}}Y_i - \frac{1}{n_{0,s}}\sum_{i \in \aleph_{0,s}}Y_i -\hat \tau^{unadj}\right)^2.     
\end{align*}
This estimator is consistent as long as Assumption \ref{ass:assignment1} holds. See, for example, \cite{BCS18} for details. Therefore, it suffices to construct consistent estimators for the $(1,1)$ and $(1,2)$ elements of $ \Sigma_{\mathcal U}$ and $ \Sigma_{\mathcal V}$, and $ \Sigma_{\mathcal W}$. 

Denote the residual of the regression and its leave-one-out version are denoted as $\hat \eps_{a,s,i} = Y_i - \hat \tau_{a,s} - \breve{X}_i' \hat \beta_{a,s}$ and $\acute \eps_{a,s,i} = \hat \eps_{a,s,i}/M_{a,s,i,i}$ for $i \in \aleph_{a,s}$, respectively. Then, we construct our estimator $\hat \Sigma_{\mathcal U}$ following the cross-fit approach developed by \cite{J22}. It is also possible to construct the estimator based on the sample splitting approach proposed by \cite{CJN18}.  The cross-fit estimator, as shown below, is asymptotically valid as long as $\kappa_{a,s}<1$, but is not guaranteed to be positive definite in finite sample. In contrast, the sample splitting estimator is guaranteed to be positive semidefinite but requires $\kappa_{a,s}<1/2$. When both of them are valid, they are asymptotically equivalent. Interested readers can refer to \cite{J22} for more discussions. For the rest of the paper, we focus on cross-fit estimators of the $(1,1)$ and $(1,2)$ elements of $\hat \Sigma_{\mathcal U}$, which is defined as 
\begin{align*}
\hat \Sigma_{\mathcal U} = \begin{pmatrix}
\sum_{s \in \mathcal{S}} \sum_{a=0,1} \frac{n_s^2}{n n_{a,s}} \hat \omega_{a,s}^2	& \sum_{s \in \mathcal{S}} \sum_{a=0,1} \frac{n_s^2}{n n_{a,s}} \hat \varpi_{a,s} \\
\sum_{s \in \mathcal{S}} \sum_{a=0,1} \frac{n_s^2}{n n_{a,s}} \hat \varpi_{a,s}  & \bigcdot 
\end{pmatrix},
\end{align*}
where 
\begin{align*}
& \hat \omega_{a,s}^2 = \frac{1}{n_{a,s}}\gamma_{a,s,n}^{-2} \sum_{i \in \aleph_{a,s}} \left(\sum_{j \in \aleph_{a,s}} M_{a,s,i,j}\right)^2 Y_i \acute \eps_{a,s,i},  \\
& \hat \varpi_{a,s} = \frac{1}{n_{a,s}}\gamma_{a,s,n}^{-1} \sum_{i \in \aleph_{a,s}} \left(\sum_{j \in \aleph_{a,s}} M_{a,s,i,j}\right) Y_i \acute \eps_{a,s,i}.
\end{align*}

\begin{rem}
	\cite{CJN18_ET} pointed out that different numbers of covariates can be used for point estimation and for variance construction when the covariates possess approximation power (such as a series basis). Since $\Sigma_{\mathcal U}$ is analogous to the variance of the intercept in a linear regression with many regressors, this same idea can be applied to the construction of $ \hat \omega_{a,s}^2$ and $ \hat \varpi_{a,s}$. Moreover, if the error term $\eps_i(a)$ is homogeneous within a stratum (i.e., $\mathbb{E} (\eps_i^2(a)|X_i,S_i) = \mathbb{E} (\eps_i^2(a)|S_i)$), then following \cite{CJN18_ET}, we can define
	\begin{align*}
	\hat \Sigma_{\mathcal U}^{HO} = \begin{pmatrix}
	\sum_{s \in \mathcal{S}} \sum_{a=0,1} \frac{n_s^2}{n n_{a,s}}  \gamma_{a,s,n}^{-1} \hat {\mathbb V}_{a,s} 	& \sum_{s \in \mathcal{S}} \sum_{a=0,1} \frac{n_s^2}{n n_{a,s}}\hat {\mathbb V}_{a,s} \\
	\sum_{s \in \mathcal{S}} \sum_{a=0,1} \frac{n_s^2}{n n_{a,s}} \hat {\mathbb V}_{a,s}  & \bigcdot 
	\end{pmatrix},
	\end{align*}
	where 
	\begin{align*}
	\hat {\mathbb V}_{a,s} = \frac{1}{n_{a,s}-1 -\kappa_n} \sum_{i \in \aleph_{a,s}} \hat \eps_i^2(a), 
	\end{align*}
	and show that $\hat \Sigma_{\mathcal U}^{HO} \convP \Sigma_{\mathcal U}$. 
\end{rem}

\begin{rem}
	We can also consider the HC3 variant of $\hat \omega_{a,s}^2$, denoted as $\tilde \omega_{a,s}^2$, which is defined as:
	\begin{align*}
	\tilde \omega_{a,s}^2 = \frac{1}{n_{a,s}}\gamma_{a,s,n}^{-2} \sum_{i \in \aleph_{a,s}} \left(\sum_{j \in \aleph_{a,s}} M_{a,s,i,j}\right)^2 M_{a,s,i,i}^{-2} \hat \eps^2_{a,s,i}.
	\end{align*}
	\cite{CJN18} demonstrated that $\tilde \omega_{a,s}^2$ is asymptotically conservative, meaning:
	\begin{align*}
	\liminf_{n \rightarrow \infty} \tilde \omega_{a,s}^2 \geq \omega_{a,s,\infty}^2.  
	\end{align*}
	With this construction, along with the $(1,1)$ elements of $\hat \Sigma_{\mathcal V}$ and $\hat \Sigma_{\mathcal W}$ outlined below, we can develop an asymptotically conservative variance estimator for the adjusted estimator $\hat \tau^{adj}$.
	
\end{rem}

To define $\hat \Sigma_{\mathcal V}$, we note that $\mathbb{E}(Y_i(a)|X_i,S_i)-\mathbb{E}(Y_i(a)|S_i)$ can be approximated by $\breve{X}_i^\top \beta_{a,s}$ for $i \in \aleph_{a,s}$, where $\beta_{a,s}$ is defined in Assumption \ref{ass:linear}(iv). \cite{KSS2020} considered the estimation and inference for quadratic functions of the coefficient in one linear regression with many regressors. However, in our setting, $\Sigma_{\mathcal V}$ depends on multiple linear coefficients $\beta_{a,s}$ which are estimated from different strata of observations, and each regression model may only be approximately linear. Therefore, the estimator $\hat \Sigma_{\mathcal V}$ includes both the quadratic and cross-product terms of the estimated coefficients $\hat \beta_{a,s}$. Denote $\Gamma_{a,s} = \sum_{i \in \aleph_{a,s}}\breve{X}_i\breve{X}_i^\top$ and $\Gamma_{s} = \sum_{i \in \aleph_{s}}\breve{X}_i\breve{X}_i^\top$. Following the spirit of \cite{KSS2020}, our estimator of the $(1,1)$ and $(1,2)$ elements of $\hat \Sigma_{\mathcal V}$ is defined as 
\begin{align*}
\hat \Sigma_{\mathcal V} = \begin{pmatrix}
\hat \Sigma_{\mathcal V}^{adj} &  \hat \Sigma_{\mathcal V}^{adj} \\
\hat \Sigma_{\mathcal V}^{adj} & \bigcdot
\end{pmatrix},
\end{align*}
where 
\begin{align*}
& \hat \Sigma_{\mathcal V}^{adj} = \sum_{s \in \mathcal{S}}\hat p_s \left[\sum_{a=0,1}\left(\frac{1}{n_{a,s}} \hat \beta_{a,s}^\top \Gamma_{a,s} \hat \beta_{a,s} - \frac{1}{n_{a,s}}\sum_{i \in \aleph_{a,s}}P_{a,s,i,i} Y_i \acute \eps_{a,s,i}\right) - \frac{2}{n_s}\hat \beta_{1,s}^\top \Gamma_{s} \hat \beta_{0,s} \right].
\end{align*}

The estimator of $\Sigma_{\mathcal W}$ is standard: 
\begin{align*}
\hat \Sigma_{\mathcal W} = \sum_{s \in \mathcal{S}}\hat p_s \left(\hat \tau_{1,s} - \hat \tau_{0,s} - \hat \tau^{adj} \right)^2. 
\end{align*}

\begin{ass}
	Suppose $\max_{a=0,1,s \in \mathcal{S}} \left\Vert \Gamma_s^{1/2} \Gamma_{a,s}^{-1}\Gamma_s^{1/2}\right\Vert_{op} = o_P(n). $
	\label{ass:variance}
\end{ass}

\begin{rem}\label{rem:ass4}
	By the random matrix theory, under general regularity conditions, one can show that all the eigenvalues of $\Gamma_s/n_s$ and $\Gamma_{a,s}/n_{a,s}$ are bounded and bounded away from zero, even if $\kappa_{a,s}>0$. This then implies $\max_{a=0,1,s \in \mathcal{S}} \left\Vert \Gamma_s^{1/2} \Gamma_{a,s}^{-1}\Gamma_s^{1/2}\right\Vert_{op} = O_P(1)$, and thus, Assumption \ref{ass:variance} holds. 
\end{rem}

\begin{thm}
	Suppose Assumptions \ref{ass:assignment1}--\ref{ass:variance} hold. Let 
	$$\hat \Sigma = \begin{pmatrix}
	\sum_{s \in \mathcal{S}} \sum_{a=0,1} \frac{n_s^2}{n n_{a,s}} \hat \omega_{a,s}^2 +    \hat \Sigma_{\mathcal V}^{adj} + \hat \Sigma_{\mathcal W} & \sum_{s \in \mathcal{S}}\sum_{a=0,1} \frac{n_s^2}{n n_{a,s}} \hat \varpi_{a,s} +    \hat \Sigma_{\mathcal V}^{adj} + \hat \Sigma_{\mathcal W}   \\
	\sum_{s \in \mathcal{S}}\sum_{a=0,1} \frac{n_s^2}{n n_{a,s}} \hat \varpi_{a,s} +    \hat \Sigma_{\mathcal V}^{adj} + \hat \Sigma_{\mathcal W}  & \hat \Sigma_{2,2}
	\end{pmatrix}.$$ Then, we have
	\begin{align*}
	\hat \Sigma \convP \Sigma. 
	\end{align*}
	\label{thm:variance}
\end{thm}

Indeed, stratum-specific RAs reduce the effective sample size, exacerbating the dimensionality issue. When $k_n$ exceeds the effective sample size $n_{a,s}$, researchers can employ Ridge-regularized linear regression for adjustment. Studying the statistical properties of the corresponding regression-adjusted ATE estimator is an interesting avenue for future research.




\section{Fixed or Moderate Number of Regressors}
\label{sec:fixed}
In this section, we consider the CARs specified in Assumption \ref{ass:assignment1} with a fixed or moderately diverging dimension of $X_i$, i.e., $k_n$. We then introduce another set of regularity conditions to replace Assumptions \ref{ass:linear}--\ref{ass:variance}. Under these new conditions, the Wald test based on the \textit{same} estimator $\hat \tau^*$ - constructed using the \textit{same} $\hat \tau^{adj}$, $\hat \tau^{unadj}$, and $\hat \Sigma$ as defined above - maintains exact asymptotic size under the null and is weakly more powerful than the Wald test based on the unadjusted estimator $\hat \tau^{unadj}$. 

Particularly, in this dimension regime, there's no need to assume that the RAs are approximately correctly specified, as in Assumption \ref{ass:linear}(iv). So the results here are not merely special cases of those in Section \ref{sec:trade-off}. Additionally, they are also novel to the literature of RAs with fixed number of regressors (e.g., \cite{YYS22}). This is because the cross-fit estimator of the covariance matrix and the optimal linear combination estimator, motivated by our high-dimensional analysis, are new.


\begin{ass}	\begin{enumerate}[label=(\roman*)]
		\item 
		Denote the dimension of $X_i$ as $k_n$ and suppose $\max_{i \in [n]}||\breve X_i||_\infty = O_P(\xi_n)$ such that $k_n^2 \xi_n^2 \log(k_n) = o(n)$. In addition, there exist constants $c,C$ such that
		\begin{align*}
		0 < c \leq \lambda_{\min}\left(\frac{1}{n_{a,s}}\sum_{i \in \aleph_{a,s}}\breve X_i \breve X_i^\top\right) \leq  \lambda_{\max}\left(\frac{1}{n_{a,s}}\sum_{i \in \aleph_{a,s}}\breve X_i \breve X_i^\top\right) \leq C<\infty
		\end{align*}
		and 
		\begin{align*}
		0 < c \leq \lambda_{\min}\left(Var( X_i |S_i=s)\right) \leq  \lambda_{\max}\left(Var( X_i |S_i=s)\right) \leq C<\infty.
		\end{align*}
		\item Suppose $ \sum_{s \in \mathcal{S}}\frac{p_s}{\pi_s (1-\pi_s)} Var(X_i^\top \overline{\beta}_s^*|S_i=s) \geq c> 0$ for some constant $c$, where $\overline{\beta}_s^* = (1-\pi_s)\beta_{1,s}^* + \pi_s \beta_{0,s}^*$ and $\beta_{a,s}^*= Var(X_i\mid S_i=s)^{-1} Cov(X_i,Y_i(a)|S_i=s).$
	\end{enumerate}
	\label{ass:reg_fixed}
\end{ass}

\begin{rem}\label{rem:ass5}
	Assumption \ref{ass:reg_fixed}(i) allows the dimension of $\breve X_i$ to be either fixed or diverging with the sample size. If $\xi_n$ is bounded,\footnote{For example, the elements of the $k_n \times 1$ vector $X_i$ are independent and bounded.} then the rate requirement boils down to $k_n^2 \log(k_n) = o(n)$ or $k_n = o(n^{1/2})$ up to a logarithmic factor.       
\end{rem}

\begin{rem}\label{rem:ass52}
	Assumption \ref{ass:reg_fixed}(ii) means the covariates have non-negligible prediction power of the outcome in at least one of the strata. If this condition fails, then the adjusted and unadjusted estimators are asymptotically equivalent. The combination estimator will also be equivalent to them asymptotically if we apply the same strategy mentioned in Remark \ref{rem:lambda} to safeguard against the degeneracy.         
\end{rem}

\begin{thm}
	Suppose Assumptions \ref{ass:assignment1} and \ref{ass:reg_fixed} holds. Then, we have 
	\begin{align*}
	\Omega^{-1/2}  \begin{pmatrix}
	\sqrt{n}(\hat \tau^{adj} - \tau) \\
	\sqrt{n}(\hat \tau^{unadj} - \tau) 
	\end{pmatrix} \convD \N\left(0_2, I_2 \right), 
	\end{align*}
	\begin{align*}
	\Omega & = \left\{ \mathbb{E}\left[\frac{Var(Y_i(1)|X_i,S_i)}{\pi_{S_i}} + \frac{Var(Y_i(0)|X_i,S_i)}{1-\pi_{S_i}}\right] + \mathbb{E}(m_{1}(X_i,S_i) - m_{0}(X_i,S_i)-\tau)^2 \right\}1_2 1_2^\top \\
	& + \sum_{s \in \mathcal{S}}\frac{p_s}{\pi_s (1-\pi_s)}\begin{pmatrix}
	V_s & V_s \\
	V_s & V_s'
	\end{pmatrix},
	\end{align*}
	and 
	\begin{align*}
	\hat \Sigma^{-1}\Omega \convP I_2,
	\end{align*}
	where $m_a(x,s) = \mathbb E (Y_i(a)|X_i=x,S_i=s)$,  
	\begin{align*}
	& V_s = Var((1-\pi_s) m_1(X_i,s) + \pi_s m_0(X_i,s) - X_i^\top \overline{\beta}_s^* \mid S_i=s), \\
	& V_s' = Var((1-\pi_s) m_1(X_i,s) + \pi_s m_0(X_i,s) \mid S_i=s),
	\end{align*}

	In addition, we have
	\begin{align*}
	\sqrt{n} \left[(\hat w, 1-\hat w)\hat \Sigma (\hat w, 1-\hat w)^\top\right]^{-1/2} ( \hat \tau^* - \tau) \convD \N(0,1).
	\end{align*}
	In this case, $\hat \tau^*$ is asymptotically equivalent to $\hat \tau^{adj}$ in sense that $\hat \tau^* = \hat \tau^{adj} + o_P(n^{-1/2})$ and weakly more efficient than $\hat \tau^{unadj}$.
	\label{thm:fixed_k} 
\end{thm}

\begin{rem}
	The limit distribution of $\hat \tau^{adj}$ has already been derived by \cite{YYS22} when $k_n$ is fixed. Here, our main contributions are (1) deriving the joint distribution of $\hat \tau^{adj}$ and $\hat \tau^{unadj}$ while allowing the dimension of the covariates to diverge with the sample size and (2) establishing the consistency of our cross-fit covariance matrix estimator $\hat \Sigma$ under the moderate dimension framework. Theorem  \ref{thm:fixed_k} shows that the Wald test based on our optimal linear combination estimator $\hat \tau^*$ and the covariance matrix estimator $\hat \Sigma$ proposed in Section \ref{sec:var} still controls asymptotic size under the null and has the same power as $\hat \tau^{adj}$ under local alternatives with a moderate dimension of $X_i$. In fact, $\hat \tau^*$ is equivalent to $\hat \tau^{adj}$, and \cite{YYS22} have shown that, when $k_n$ is fixed, $\hat \tau^{adj}$ (and thus, $\hat \tau^*$) is the optimally linearly adjusted estimator in the sense that it achieves the minimum asymptotic variance among a class of estimators adjusted by linear functions of $X_i$. This optimality result naturally extends to the moderate $k_n$ case considered here as long as the rate requirement in Assumption \ref{ass:reg_fixed} is satisfied. Therefore, $\hat \tau^*$ is weakly more efficient that $\hat \tau^{unadj}$. See, for example, \cite{LTM20}, \cite{MTL20}, and \cite{YYS22} for more discussion.
\end{rem}

\begin{rem}\label{rem:specification}
	We emphasize that Theorem \ref{thm:fixed_k} holds without requiring the linear RAs to be approximately correctly specified, even when the dimension of covariates diverges. \cite{LD21} and \cite{CMO23} established similar results for their regression-adjusted ATE estimator under complete randomization and finite-population asymptotics, but they imposed the restriction $k_n = o(n^{1/2})$. Our rate requirement aligns with theirs, differing only by a logarithmic factor. When $k_n$ grows faster than $n^{1/2}$, the regression-adjusted ATE estimator may suffer from asymptotic bias if the linear regressions are not approximately correctly specified. To address this, \cite{LD21} and \cite{CMO23} proposed bias correction methods, though their approaches permit at most $k_n = o(n^{2/3})$. More recently, \cite{LYW23} introduced a debiased regression-adjusted estimator that allows $k_n$ to be of the same order as $n$ without assuming the correct specification. Our work differs from theirs in three key aspects. First, \citeauthor{LYW23}'s (\citeyear{LYW23}) regression adjustment follows a specially designed approach distinct from the original YYS proposal, which remains our main focus; indeed, \cite{LYW23} highlighted the technical challenges in analyzing the original regression adjusted estimator (i.e., \citeauthor{L13}'s (\citeyear{L13}) estimator, which is analogous to YYS in complete random sampling) when $k_n$ is proportional to $n$. Second, their results apply only to complete random sampling under a finite-population framework, whereas we consider general CARs under a superpopulation framework. Third, we propose a linear combination estimator that is guaranteed to be at least as efficient as the unadjusted estimator.
	
	Extending similar regression adjustment and bias-correction methods to our setting, which accommodates a more general randomization scheme (i.e., CAR) under a sampling-based superpopulation framework, would be an interesting avenue for future research. However, we consider this beyond the scope of the present paper. Therefore, we view our inference method as a complement rather than a substitute for those proposed by \cite{LD21}, \cite{CMO23}, and \cite{LYW23}.
	
	In summary, when treatment is assigned by CARs and either the dimension of covariates is $o(\sqrt{n})$ but the RAs may be misspecified, or the dimension of covariates is a fraction of the sample size and the RAs are approximately correctly specified (as in the cases discussed after Assumption \ref{ass:linear}), researchers can choose our inference method. However, when treatment is assigned by complete randomization, the dimension of the covariates is greater than $n^{1/2}$, and the RAs are not likely to be approximately correctly specified (as in the case where there are many original continuous regressors), researchers should turn to the bias-corrected inference methods proposed by \cite{LD21}, \cite{CMO23}, \cite{LYW23}. 
\end{rem}



\section{Simulations}
\label{sec:simulations}

We conduct a simulation study to evaluate the finite sample performance of our proposed inference method. The nominal level is $\alpha=5\%$ throughout the section.  

\subsection{Data Generating Processes}
For $a\in\{0,1\}$, we generate potential outcomes according to the equation
\begin{equation}
Y_{i}(a)=\mu_{a}+m_{a}(Z_{i})+ \sigma_{a}(Z_{i})\epsilon_{i}(a),\label{eq:simulpart01}
\end{equation}
where $\mu_{a}$, and $m_{a}\left(Z_{i}\right)$ and $\sigma_a(Z_i)$ are specified as follows. In each of the following specifications, $\{Z_i, \epsilon_{i}(1), \epsilon_{i}(0)\}$ are i.i.d. We have considered four models following the designs by \cite{CJN18}. For the first three models, the dimension of $Z_i$ is denoted as $d_n$, and we let $X_i$ be either the full set or the first $k_n$ elements of $Z_i$. For the last model, $Z_i$ is of fixed dimension, and we let $X_i$ be either the full set or the first $k_n$ elements of polynomial series of $Z_i$. The polynomial series will be specified below. By varying the choice of $k_n$, we can illustrate the uniformity of our inference method.  

\begin{description}
	\item [{Model 1: Linear model with many dummy variables.}] The dimension of $Z_i$ is set to $d_n=0.2n/|\mathcal{S}|$. The first entry of $Z_i$ is uniformly distributed in $[-1,1]$, i.e., $Z_{1i}\sim U[-1,1]$. 
	The other entries of $Z_i$ are dummies: $[Z_{2i},Z_{3i},\cdots,Z_{d_n,i}]^\top=\textbf{1}(\textbf{v}_{i}\geq \Phi^{-1}(0.8))$ with $\textbf{v}_{i} \sim \mathcal{N}(0,\textbf{I}_{d_n-1})$. We set $m_a(\cdot)$ and $\sigma_a(\cdot)$ respectively as
	$$
	m_1(Z_i)=m_0(Z_i)=Z_{1i}+2\sum_{j=2}^{d_n}Z_{ji}/\sqrt{d_n-1}
	$$
	and $$\sigma_{0}(Z_{i})=\sigma_{1}(Z_{i})=c_{\epsilon}\left[1+\left(Z_{1i}+\sum_{j=2}^{d_n}Z_{ji}/\sqrt{d_n-1}\right)^2\right]^{1/2}.$$ 
	Let $(\epsilon_{i}(1),\epsilon_{i}(0))\sim \mathcal{N}(0,\mathbf{I}_2)$ be independent of $Z_i$. We set $c_\eps$ as the normalizing constant such that $\mathbb E \sigma_a^2(Z_i) = 1$. 
	
	\item [{Model 2: Linear model with many independent continuous variables.}] This model is the same as model 1, except that $Z_{ji}$'s are $U[-1,1]$ and 
	$$
	m_1(Z_i)=m_0(Z_i)=2\sum_{j=2}^{d_n}Z_{ji}/\sqrt{d_n-1},
	$$
	where $Z_{1i}$ does not appear in $ m_1(Z_i)$ and $m_0(Z_i)$.
	
	\item  [{Model 3: Linear model with many correlated continuous variables.}] This model is the same as Model 2, except that $[Z_{2i},\cdots,Z_{d_n,i}]^\top=\Sigma_\rho^{1/2}\mathbf{v}_i$, where   $\mathbf{v}_i\in \mathbb{R}^{d_n-1}$ and $\Sigma_\rho$ is a Toeplitz matrix with
	$
	[\Sigma_\rho]_{i,j}=\rho^{|i-j|}
	$ and $\rho=0.6$. We set $\mathbf{v}_i$ to contain i.i.d. $U[-1,1]$ entries. 
	
	\item [{Model 4: Nonlinear model with many polynomial series.}] Let $d_{n}=6$, $Z_{ji}\sim U[-1,1]$ for $j=1,\cdots,6$, 
	$$m_{1}(Z_i)=m_{0}(Z_i)=2\exp\left( \sqrt{\frac{1}{6}\sum_{j=1}^{6}Z_{ji}^2} \right), $$
	and $$\sigma_{0}(Z_{i})=\sigma_{1}(Z_{i})=c_{\epsilon}\left[1+\left(Z_{1i}+\sum_{j=2}^{6}Z_{ji}/\sqrt{5}\right)^2\right]^{1/2}.$$  
	We use the following polynomials or some of them as regressors:
	\begin{itemize}
		\item First order terms: $Z_{1i},Z_{2i},Z_{3i},Z_{4i},Z_{5i},Z_{6i}$,
		\item Second order terms: $Z_{1i}^2,Z_{2i}^2,Z_{3i}^2,Z_{4i}^2,Z_{5i}^2,Z_{6i}^2$  and first-order interactions,
		\item Third-order terms: $Z_{1i}^3,Z_{2i}^3,Z_{3i}^3,Z_{4i}^3,Z_{5i}^3,Z_{6i}^3$ and second-order interactions,
		\item Fourth-order terms: $Z_{1i}^4,Z_{2i}^4,Z_{3i}^4,Z_{4i}^4,Z_{5i}^4,Z_{6i}^4$ and third-order interactions.
	\end{itemize}
	
\end{description}

For each model, strata are determined by dividing the support of $Z_{1i}$ into $|\mathcal{S}|$ intervals of equal length. Specifically, we have 
$S_i = \sum_{j = 1}^{|\mathcal{S}|} 1\{Z_{1i} \leq g_j\},
$
where $g_j=2j/|\mathcal{S}|-1$. For all strata, we set $\pi_s=1/2$. The treatment status is determined according to one of the following four CAR schemes:
\begin{enumerate}
	\item SRS: Treatment assignment is generated as in Example 1,
	\item BCD: Treatment assignment is generated as in Example 2 with $\lambda = 0.75$,
	\item WEI: Treatment assignment is generated as in Example 3 with $\phi(x) = (1-x)/2$,
	\item SBR: Treatment assignment is generated as in Example 4.
\end{enumerate}




\subsection{Simulation Results}
\label{sec:sim_results}

Rejection probabilities are computed using $10,000$ replications. We compare our methods with three existing methods recently introduced by \cite{BCS17}, \cite{YYS22}, and \cite{LTM20}. We list these methods below.

\begin{description}
	\item [$\hat{\tau}^{adj}$:] The inference method based on  the fully saturated regression adjusted estimator introduced in Section \ref{sec:setup} and the variance estimator $\hat{\Sigma}_{1,1}$ as in Theorem \ref{thm:variance}.
	\item [$\hat{\tau}^\ast$\  \ :] The inference method based on the optimal linear combination estimator introduced in Section \ref{sec:optimal} and Theorem \ref{thm:further}.
	\item[BCS:] The inference method introduced by \cite{BCS18}. The estimator for $\tau$ is exactly the unadjusted IPW estimator $\hat{\tau}^{unadj}$. 
	\item[YYS:] The inference method introduced by \cite{YYS22}. The estimator for $\tau$ is exactly the same as $\hat{\tau}^{adj}$. However, its variance estimator does not take into account the issue of many covariates.
	\item[LTM:] The inference method introduced by \cite{LTM20}, which uses Lasso to select relevant regressors in RA.  Specifically, we use the `HDM' package in R with the default choice of tuning parameter to obtain the corresponding Lasso estimates. 
\end{description}

Tables \ref{Table1} reports the size results for $n=400$ and $800$, under $H_{0}:\mu_1=\mu_0=0$ and $S=2$. In the first three models, we include all elements of $Z_i$ as regressors, i.e., $X_i=Z_i$. For the last model, we include all polynomial series of $Z_i$ as regressors. Therefore, the number of covariates used in each regression is approximately $\mathbf{40\%}$ of the effective sample size (i.e., $\kappa_{a,s} = 0.4$). We find that our methods $\hat{\tau}^{adj}$ and $\hat{\tau}^{\ast}$ have good size control under four models, improving as $n$ increases. In contrast, YYS has high rejection rates, over 10\%, in all cases, even when $n$ is as large as 800. Our theory in Section \ref{sec:trade-off} shows this size distortion is caused by ignoring the estimation errors from many regressors. Furthermore, LTM has rejection probabilities over 7\% in many cases. In particular, in Model 3 with correlated covariates, LTM has rejection rates over 9\% in all randomization schemes, even when $n$ is 800. This is because in this model, the sparsity condition does not hold. Finally, BCS does not use any regressors, and thus, is immune to the size distortion due to the many regressors. 

Table \ref{Table2} shows the power results under $H_{1}:\mu_1-\mu_0=0.2$. We find that, first, $\hat{\tau}^{\ast}$ always has higher power than $\hat{\tau}^{adj}$ and BCS, as our theory predicts. Second, YYS ignores the cost of RAs, resulting in a smaller variance estimator, and thus, higher rejection rates under both the null and alternative. Third, LTM has lower power than $\hat{\tau}^{\ast}$ in Models 1-3, where the sparsity condition for the regression coefficients fails. In fact, in this case, Lasso excludes many informative covariates in the adjustments. In model 4, LTM has better power than $\hat{\tau}^{\ast}$ since the sparsity condition is satisfied. By exploiting the sparsity, it is possible to show that the ATE estimator with the $\ell_1$ regularized adjustments achieves the semiparametric efficiency bound. However, even in this case that favors the Lasso regression adjustments, the power of our method based on $\hat \tau^*$ is very close to LTM. 

\newcolumntype{L}{>{\raggedright\arraybackslash}X} \newcolumntype{C}{>{\centering\arraybackslash}X}

\begin{table}[tbp]\caption{Rejection rate (in percent) under $H_0:\mu_{1}-\mu_{0}=0$}  \label{Table1}
	\centering
	\begin{tabularx}{1\textwidth}{CCCCCCCCCC} 
		\toprule
		\multicolumn{1}{l}{} & \multicolumn{1}{c}{} &  \multicolumn{4}{c}{n=400, $k_n=40$} &\multicolumn{4}{c}{n=800, $k_n=80$} \\ 
		
		\cmidrule(lr){3-6} \cmidrule(lr){7-10}
		
		Model                & Method               & SRS   & BCD   & WEI   & SBR   & \multicolumn{1}{c}{SRS} & \multicolumn{1}{c}{BCD} & \multicolumn{1}{c}{WEI} & \multicolumn{1}{c}{SBR} \\ 
		\midrule
		1     & $\hat{\tau}^{adj}$  & 5.78  & 5.89  & 5.67  & 5.89  & 5.43  & 5.32  & 5.46  & 5.32  \\
		& $\hat{\tau}^{\ast}$ & 5.72  & 5.96  & 5.81  & 5.96  & 5.37  & 5.41  & 5.28  & 5.41  \\
		& BCS                                     & 5.05  & 5.41  & 5.22  & 5.41  & 5.12  & 5.30   & 5.18  & 5.30   \\
		& YYS                                     & 12.04 & 12.23 & 12.56 & 12.23 & 11.73 & 11.53 & 11.53 & 11.53 \\
		& LTM                                     & 6.64  & 7.08  & 7.16  & 7.08  & 5.95  & 6.33  & 6.07  & 6.33  \\ \hline
		2     & $\hat{\tau}^{adj}$  & 5.51  & 5.57  & 5.32  & 5.57  & 5.10   & 5.65  & 5.38  & 5.65  \\
		& $\hat{\tau}^{\ast}$ & 5.42  & 5.53  & 5.18  & 5.53  & 5.06  & 5.63  & 5.51  & 5.63  \\
		& BCS                                     & 5.04  & 4.88  & 5.12  & 4.88  & 4.82  & 5.18  & 5.14  & 5.18  \\
		& YYS                                     & 11.53 & 11.49 & 11.53 & 11.49 & 11.56 & 11.63 & 11.35 & 11.63 \\
		& LTM                                     & 7.36  & 7.33  & 7.29  & 7.33  & 7.25  & 7.43  & 7.02  & 7.43  \\ \hline
		3     & $\hat{\tau}^{adj}$  & 5.90  & 6.12  & 5.79  & 5.85  & 5.44  & 5.75  & 5.48  & 5.75  \\
		& $\hat{\tau}^{\ast}$ & 5.76  & 6.23  & 5.83  & 5.98  & 5.39  & 5.75  & 5.60   & 5.75  \\
		& BCS                                     & 4.79  & 4.84  & 5.15  & 5.27  & 5.12  & 5.04  & 5.23  & 5.04  \\
		& YYS                                     & 11.47 & 11.32 & 11.33 & 11.26 & 11.38 & 11.68 & 11.3  & 11.68 \\
		& LTM                                     & 9.82  & 9.47  & 9.22  & 9.53  & 9.69  & 9.69  & 9.62  & 9.69  \\ \hline
		4     & $\hat{\tau}^{adj}$  & 6.63  & 7.11  & 7.25  & 7.09  & 6.28  & 6.45  & 6.68  & 6.45  \\
		& $\hat{\tau}^{\ast}$ & 7.21  & 7.57  & 7.57  & 7.64  & 6.29  & 6.67  & 6.52  & 6.67  \\
		& BCS                                     & 4.83  & 5.02  & 5.09  & 5.37  & 4.93  & 5.28  & 4.93  & 5.28  \\
		& YYS                                     & 12.52 & 12.87 & 13.40 & 12.62 & 12.46 & 12.86 & 13.08 & 12.86 \\
		& LTM                                     & 6.07  & 6.24  & 6.14  & 6.57  & 5.76  & 5.87  & 5.92  & 5.87  \\ 
		
		\bottomrule
		
	\end{tabularx}
	
\end{table}

Figures \ref{fig:sbr400} and \ref{fig:sbr800} show how the rejection rates of $\hat{\tau}^{adj}$, $\hat{\tau}^{\ast}$, YYS and LTM vary with the number of regressors under the null and alternative hypotheses. We consider Models 1-4 with $n=800$ and $n=400$, respectively. We set the number of regressors $k_n=0,2,\cdots,80$ ($k_n=0,1,\cdots,40)$ for $n=800$ ($n=400$), corresponding to $\kappa_{a,s} = 0,0.01,\cdots,0.4$. The x-axis is $\kappa_{a,s}$. The y-axis is the rejection probabilities. We display only the figures under the SBR scheme. Similar patterns are found under the other randomization schemes, which are shown in the Online Supplement.

Panels (a) of the figures show that our methods (i.e., $\hat \tau^{adj}$ and $\hat \tau^*$) have uniform size control over $\kappa_{a,s}$. Even when the dimensions of regressors vary from 0 to 80 (40) for $n=800$ ($n=400$), our methods have close-to-nominal rejection rates under $H_{0}$, consistent with our theories. In contrast, YYS has rejection probabilities that increase linearly with $\kappa_{a.s}$ in all models. Similarly, LTM displays this size distortion in Models 2 and 3, in which the sparsity condition fails.

Panels (b) of the figures show the power of the five methods under $H_{1}:\mu_1-\mu_0=0.2$. Several observations emerge. First, BCS's power, our benchmark, does not vary with $\kappa_{a,s}$ because it uses no regressors. Second, for all models, $\hat{\tau}^{\ast}$ is always more powerful than BCS and usually more powerful than $\hat \tau^{adj}$, as our theory predicts. Third, as the number of regressors increases, the misspecification error decreases, and $\hat{\tau}^{\ast}$ becomes more powerful. However, $\hat{\tau}^{adj}$'s power drops even below that of the unadjusted estimator as $k_n$ increases for Model 1. This reflects that many covariates do not always improve estimation accuracy, and the ``no-harm'' regression adjustments can in fact harm the estimation precision. Fourth, for Models 1-3, $\hat{\tau}^{\ast}$ is generally more powerful than LTM as the dimension of regressors increases. For Model 4, LTM slightly outperforms $\hat{\tau}^{\ast}$ when $\kappa_{a,s}$ is large for $n=800$ due to the sparsity of the model.

\begin{table}[tbp]   \caption{Rejection rate (in percent) under $H_1:\mu_{1}-\mu_{0}=0.2$} \label{Table2}
	\centering
	\begin{tabularx}{1\textwidth}{CCCCCCCCCC}
		\toprule
		\multicolumn{1}{l}{} &                     &  \multicolumn{4}{c}{n=400, $k_n=40$} &\multicolumn{4}{c}{n=800, $k_n=80$}  \\
		\cmidrule(lr){3-6} \cmidrule(lr){7-10}
		Model                & Method              & \multicolumn{1}{c}{SRS} & \multicolumn{1}{c}{BCD} & \multicolumn{1}{c}{WEI} & \multicolumn{1}{c}{SBR} & \multicolumn{1}{c}{SRS} & \multicolumn{1}{c}{BCD} & \multicolumn{1}{c}{WEI} & \multicolumn{1}{c}{SBR} \\ 
		\midrule
		1     &$\hat{\tau}^{adj}$  & 41.54 & 41.60  & 42.16 & 41.34 & 68.42 & 68.99 & 68.12 & 68.99 \\
		&$\hat{\tau}^{\ast}$& 44.66 & 44.93 & 45.37 & 44.44 & 72.31 & 72.44 & 72.15 & 72.44 \\
		& BCS                                     & 32.53 & 33.67 & 33.17 & 32.68 & 57.23 & 57.86 & 57.59 & 57.86 \\
		& YYS                                     & 55.59 & 55.82 & 56.66 & 55.33 & 80.26 & 80.62 & 80.57 & 80.62 \\
		& LTM                                     & 40.68 & 41.23 & 41.07 & 40.20 & 65.89 & 65.79 & 65.51 & 65.79 \\ \hline
		2     &$\hat{\tau}^{adj}$  & 41.88 & 41.99 & 41.99 & 42.28 & 68.33 & 68.82 & 68.64 & 68.82 \\
		&$\hat{\tau}^{\ast}$ & 43.70 & 43.80 & 44.05 & 43.87 & 70.56 & 71.22 & 70.69 & 71.22 \\
		& BCS                                     & 25.66 & 25.50  & 25.71 & 25.33 & 44.38 & 45.34 & 44.88 & 45.34 \\
		& YYS                                     & 56.10 & 55.84 & 55.35 & 55.93 & 79.83 & 80.32 & 80.66 & 80.32 \\
		& LTM                                     & 35.74 & 35.58 & 34.97 & 35.44 & 55.46 & 56.22 & 56.00 & 56.22 \\ \hline
		3     & $\hat{\tau}^{adj}$ & 42.56 & 42.14 & 42.39 & 42.59 & 68.24 & 68.30  & 68.34 & 68.30 \\
		& $\hat{\tau}^{\ast}$ & 43.13 & 43.05 & 42.99 & 42.77 & 68.60  & 68.98 & 69.04 & 68.98 \\
		& BCS                                     & 12.74 & 12.89 & 13.01 & 12.83 & 20.00  & 20.91 & 20.10 & 20.91 \\
		& YYS                                     & 55.62 & 55.80 & 55.43 & 55.63 & 79.62 & 79.80 & 80.05 & 79.80 \\
		& LTM                                     & 43.79 & 43.41 & 44.14 & 43.87 & 65.27 & 66.16 & 66.14 & 66.16 \\ \hline
		4     &$\hat{\tau}^{adj}$  & 41.02 & 42.41 & 41.59 & 41.52 & 65.65 & 65.67 & 65.11 & 65.55 \\
		&$\hat{\tau}^{\ast}$ & 45.59 & 46.78 & 45.88 & 46.14 & 72.20  & 71.57 & 71.58 & 72.21 \\
		& BCS                                     & 33.66 & 34.90 & 34.13 & 34.35 & 59.12 & 59.06 & 59.20 & 58.55 \\
		& YYS                                     & 55.25 & 56.03 & 55.64 & 55.51 & 78.73 & 78.85 & 78.31 & 78.58 \\
		& LTM                                     & 47.38 & 48.09 & 47.42 & 47.47 & 78.35 & 77.38 & 77.86 & 77.59    \\ 
		\bottomrule
	\end{tabularx}
\end{table}

\begin{figure}[H]
	\centering
	\subfigure[$H_{0}:\mu_1-\mu_0=0$]{
		\includegraphics[width=\textwidth]{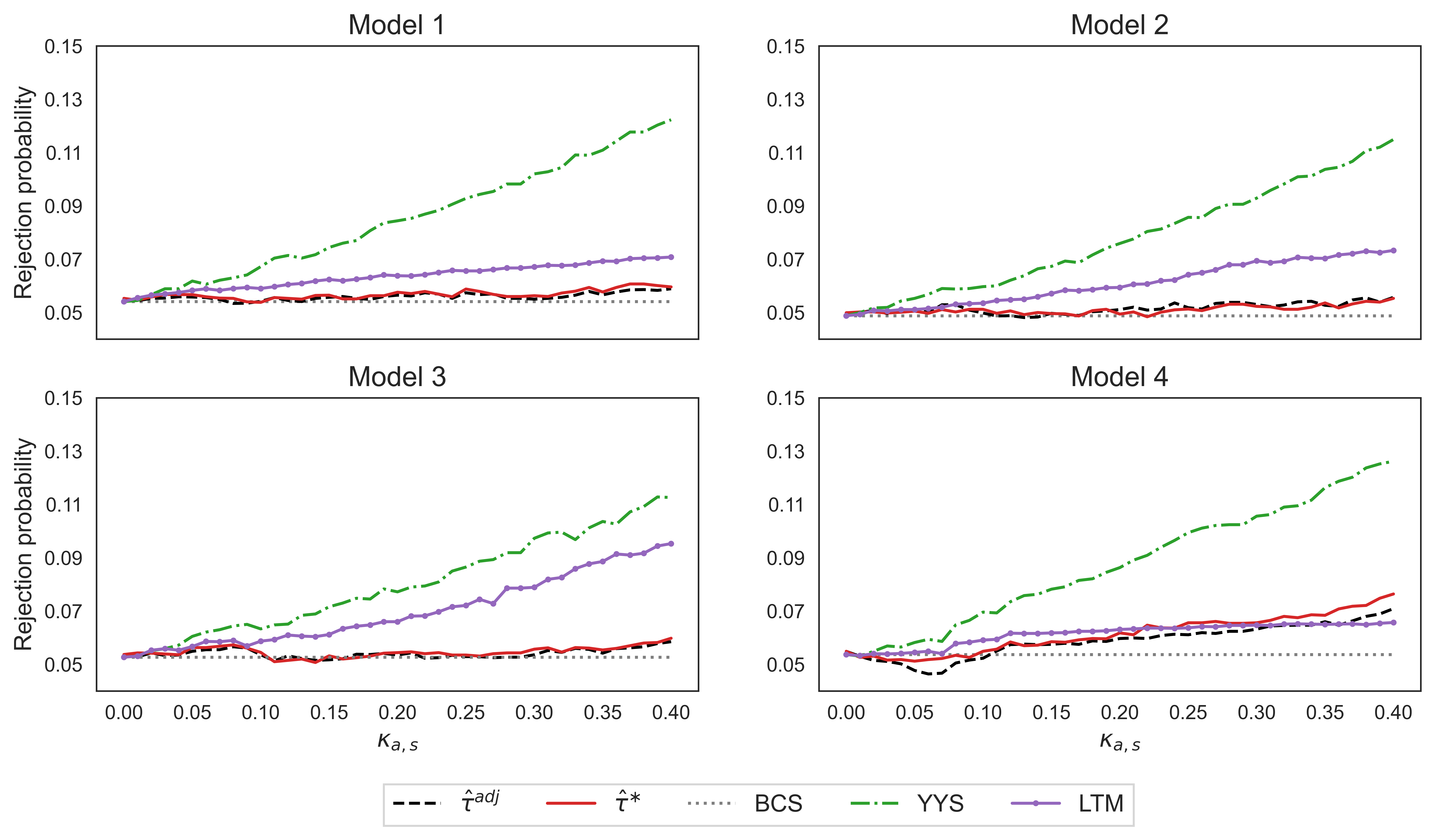}
	}
	
	\subfigure[$H_{1}:\mu_1-\mu_0=0.2$]{
		\includegraphics[width=\textwidth]{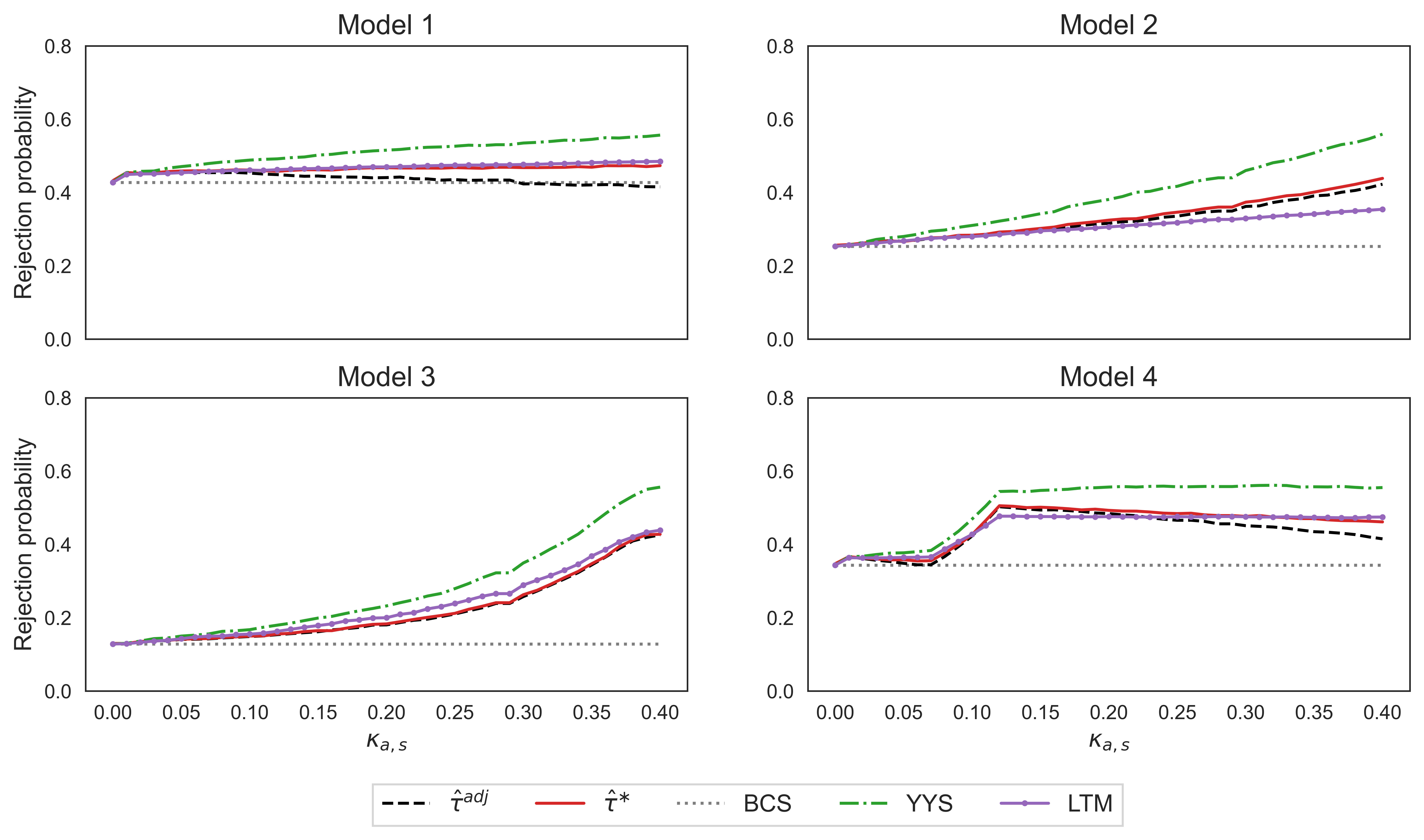}
	}
	\caption{Rejection probabilities under SBR when $n=400$}\label{fig:sbr400}
\end{figure}

\begin{figure}[H]
	\centering
	\subfigure[$H_{0}:\mu_1-\mu_0=0$]{
		\includegraphics[width=\textwidth]{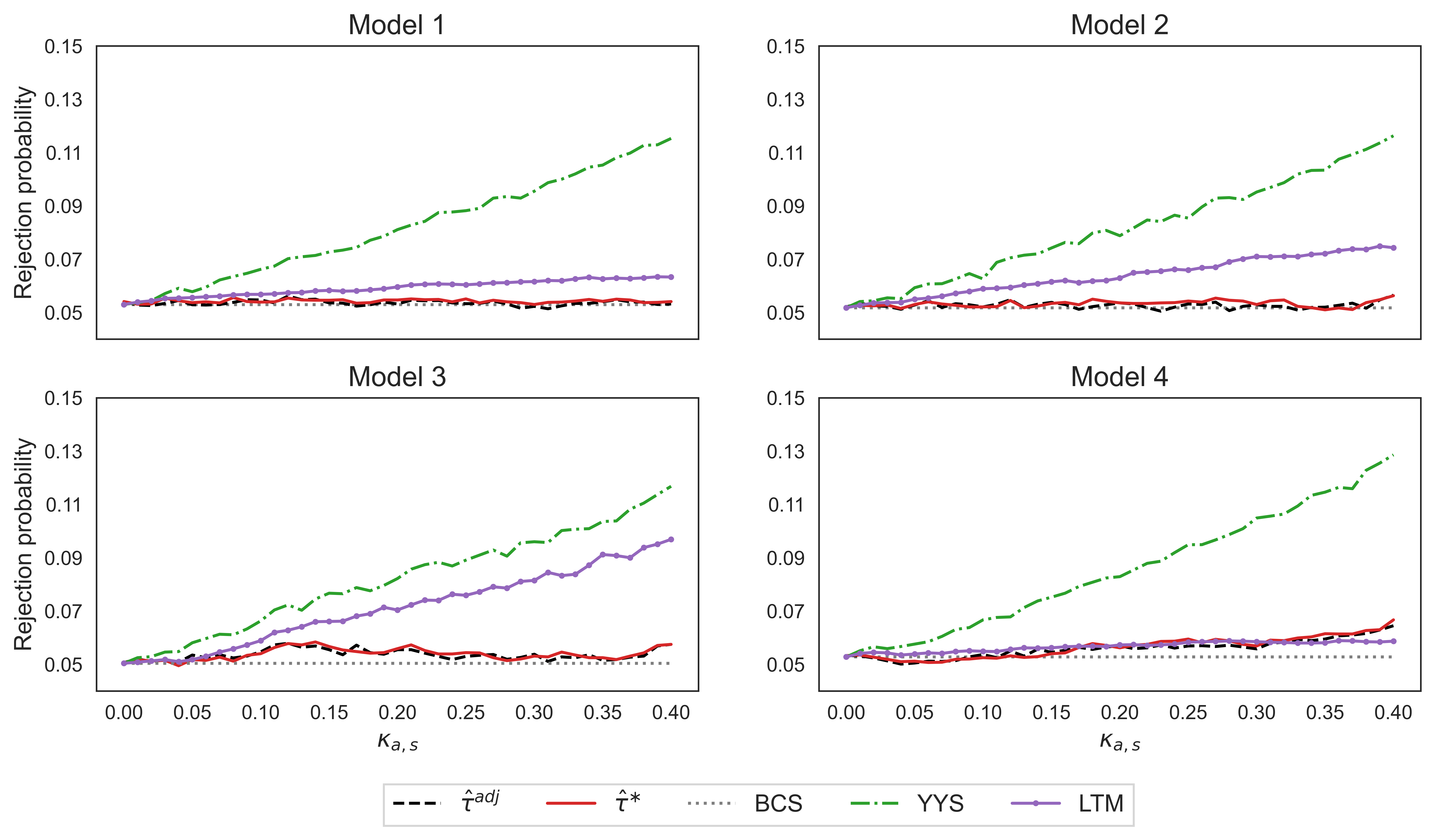}
	}
	
	\subfigure[$H_{1}:\mu_1-\mu_0=0.2$]{
		\includegraphics[width=\textwidth]{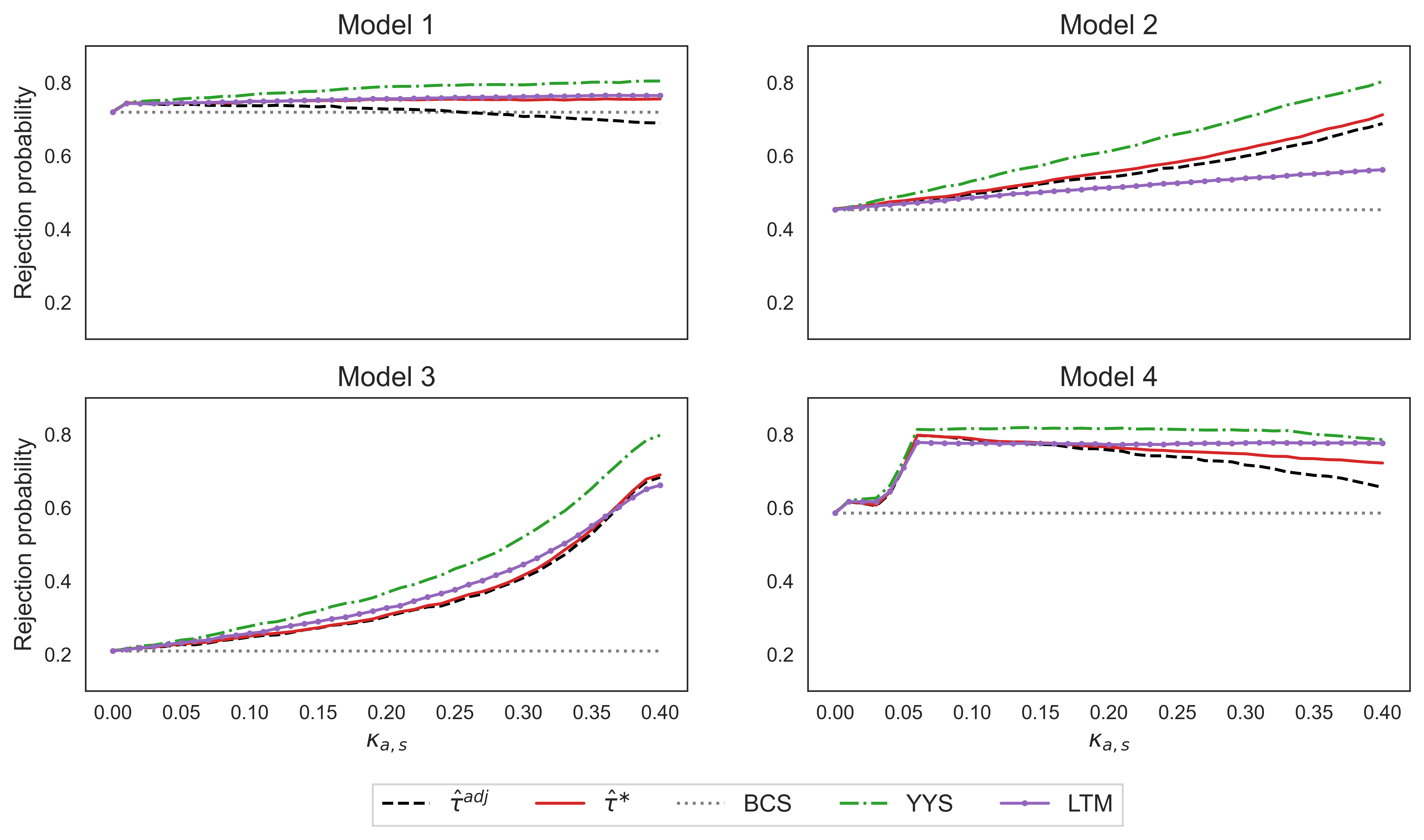}
	}
	\caption{Rejection probabilities under SBR when $n=800$}\label{fig:sbr800}
\end{figure}

\subsection{Practical Implication of Dimensionality}
When the dimension of covariates, $k_n$, is relatively small compared to $n^{1/2}$, our inference method remains valid and "no-harm," even when the linear regression adjustment is misspecified. This also holds for conventional inference methods like that of \cite{YYS22}. 

When $k_n$ exceeds $n^{1/2}$ and the linear regression is approximately correctly specified (i.e., Assumption \ref{ass:linear}(iii)), our inference method remains valid and ``no-harm'', even when $k_n$ is smaller but comparable to $n$, whereas the conventional method does not. Approximately correct specification is satisfied if the baseline covariates $Z_i$ are fixed-dimensional and continuous, with $X_i$ containing sieve basis functions of $Z_i$, or if the baseline covariates are categorical (discrete), with $X_i$ including fully saturated dummies for all categories.

When $k_n$ exceeds $n^{1/2}$ and the linear regression adjustment is not approximately correctly specified, the adjusted and linear combination estimators, $\hat{\tau}^{adj}$ and $\hat{\tau}^*$, may exhibit non-negligible biases. These biases can potentially be corrected using the methods developed by \cite{LD21} and \cite{CMO23}, provided $k_n = o(n^{2/3})$. However, the statistical properties of these bias-corrected ATE estimators, and whether they remain 'no-harm' under CARs, have not yet been studied. 

When $k_n$ exceeds $n^{2/3}$ and the linear regression adjustment is not approximately correctly specified, there is currently no known method for valid causal inference using regression-adjusted estimators under any randomization scheme, which is left for future research.

\section{Empirical Application}
\label{sec:app}

Studying the effects of child health and nutrition on educational outcomes helps us understand the connection between health and economic development in less developed countries (\citep{glewwe2007}; \citep{dupas2017}). To investigate the impact of an iron supplement program on students' schooling attainment, \cite{CCFNT16} conducted a randomized experiment in Peru with a CAR framework. 

Their experiment involved 215 students from a rural secondary school in Peru during the 2009 academic year. Students were stratified by five grade levels, with each grade having students randomly assigned to one of three groups: two treatment groups and one control group. In the treatments, one group viewed a video where a well-known soccer player endorsed iron supplements, and the other group saw a video with a doctor's endorsement. The control group watched a video that did not mention iron supplements. 


In this section, we focus on the impact of iron supplement program on the number of iron pills taken, whether a student is anemic at endline, and the student's cognitive ability. Throughout the analysis, students in the two treatment arms are grouped into the treatment group as in \cite{CCFNT16}. 

We use a fully saturated model, including two dummy covariates (gender and electricity at home) from the regressors in \cite{CCFNT16} and their interaction as the third covariate. This ensures that the model is correctly specified. Although the number of regressors is not large, it is still a significant proportion of the sample size given that the total sample size is not large either. $\kappa_{1,s}$ range from approximately 8\% to 15\%, and $\kappa_{0,s}$ range from about 16\% to 30\%.\footnote{We left out the strata where the number of observations in either treatment or control group is less than or equal to 4.}

Table \ref{tab:emp} presents the ATEs and the standard errors (in parentheses) estimated by different methods.\footnote{The description of these methods is similar to that in Section \ref{sec:sim_results}.} We make three observations from the results. First, the variance estimates of the YYS method are smaller than those of $\hat{\tau}^{adj}$ in general. For example, the standard error for pills taken is approximately about 17.6\% lower with YYS compared to $\hat{\tau}^{adj}$, despite identical ATE estimates. This indicates that YYS method may fail to fully capture estimation errors that arise when using many regressors.

Second, the standard errors of $\hat{\tau}^{\ast}$ are consistently lower than those of both $\hat{\tau}^{adj}$ and BCS for all outcomes, with some differences being substantial. For example, for cognitive ability, the standard error of $\hat{\tau}^{\ast}$ is about 8.2\% smaller than that of BCS. This result is expected because $\hat{\tau}^{\ast}$ is a weighted average of $\hat \tau^{adj}$ and the unadjusted estimator (BCS),  designed to minimize asymptotic variance. The standard errors of $\hat{\tau}^{\ast}$ are also generally lower than those of LTM, though by a smaller margin.

Third, consistent with \cite{CCFNT16}, our results indicate that the iron supplement program increases the number of iron pills taken but has no significant effect on anemia or cognitive ability.

\newcolumntype{L}{>{\raggedright\arraybackslash}X} \newcolumntype{C}{>{\centering\arraybackslash}X}

\begin{table}[t]
	\centering
	\caption{Impacts of Iron Supplements}
	\vspace{1ex}
	\begin{tabularx}{1\textwidth}{lCCCCC}
		\toprule
		Y & $\hat{\tau}^{adj}$ & $\hat{\tau}^{\ast}$ & BCS & YYS & LTM\\
		\midrule
		
		Pills Taken             & 3.844     & 4.789     & 4.773     & 3.844     & 4.471      \\ 
		& (1.5123)  & (1.3356)  & (1.3357)  & (1.2467)  & (1.2948)   \\ 
		Amenic                  & -0.047    & -0.057    & -0.081    & -0.047    & -0.070     \\ 
		& (0.0655)  & (0.0645)  & (0.0691)  & (0.0639)  & (0.0661)   \\ 
		Cognitive Ability      & 0.103     & 0.111     & 0.134     & 0.103     & 0.148      \\ 
		& (0.1419)  & (0.1402)  & (0.1527)  & (0.1363)  & (0.1454)   \\

		\bottomrule
	\end{tabularx} \\
	\vspace{-1ex}
	\justify
	Notes: The table reports the ATE estimates of the effect of the iron supplement program on the number of iron pills taken, whether a student is anemic at endline, and cognitive ability. We use gender, a dummy for electricity at home, and its product as the covariates in the regression adjustments. 
	\label{tab:emp}
\end{table}

\section{Conclusion}
This paper shows that the estimation error of the ``no-harm'' regression adjustments can contaminate the ATE estimator and degrade estimation efficiency when the number of regressors is of the same order as the sample size. We then propose a new ATE estimator which is guaranteed to be weakly more efficient than both the estimators with and without RAs. We also propose a consistent estimator of its asymptotic variance and construct the corresponding Wald-statistic.

\newpage

\appendix

\section{Verifying Assumption \ref{ass:omega} for Gaussian Covariates}\label{sec:app_omega}
For $i \in \aleph_s$, we assume that $X_i = \Sigma_{s,n}^{1/2} Z_i $ where $Z_i = (Z_{i,1},\cdots,Z_{i,k_n})$ are $k_n$ independent standard normal random variables and the covariance matrix $\Sigma_{s,n} \in \Re^{k_n \times k_n}$ is symmetric and positive definite. The assumption that $Z_i$ has zero mean is without loss of generality because the means will be canceled in the definition of $\breve X_i$. We focus on the case with $a=1$.  By an abuse of notation, we assume $\aleph_{1,s} = \{1,\cdots,n_{1,s}\}$ and $\aleph_{0,s} = \{1+n_{1,s},\cdots,n_{s}\}$ and denote $Z_{\aleph_{1,s}} = (Z_1,\cdots,Z_{n_{1,s}})^\top \in \Re^{n_{1,s} \times k_n}$ and $Z_{\aleph_{s}} = (Z_1,\cdots,Z_{n_{s}})^\top \in \Re^{n_{s} \times k_n}$. Then, we have 
$$\breve X_{\aleph_{1,s}} = X_{\aleph_{1,s}} - 1_{n_{1,s}} 1_{n_s}^\top X_{\aleph_s}/n_s \stackrel{d}{=} (Z_{\aleph_{1,s}} - 1_{n_{a,s}} 1_{n_s}^\top Z_{\aleph_s}/n_s)\Sigma_{s,n}^{1/2} = \Gamma_n Z_{\aleph_s}\Sigma_{s,n}^{1/2}$$  
and 
\begin{align*}
P_{1,s} =  \breve X_{\aleph_{1,s}} \left[\breve X_{\aleph_{1,s}}^\top \breve X_{\aleph_{1,s}}\right]^{-1} \breve X_{\aleph_{1,s}}^\top \stackrel{d}{=}    \Gamma_n Z_{\aleph_s} \left[Z_{\aleph_s}^\top \Gamma_n^\top \Gamma_n Z_{\aleph_s}\right]^{-1}   Z_{\aleph_s}^\top\Gamma_n^\top,
\end{align*}
where $\Gamma_n = (I_{n_{1,s}},0_{n_{1,s},n_{0,s}}) - 1_{n_{1,s}} 1_{n_s}^\top/n_s$ and $0_{n_{1,s},n_{0,s}}$ is a $n_{1,s} \times n_{0,s}$ matrix of zeros. For the $l$-th column of $\Gamma_nZ_{\aleph_s}$ denoted as  $\Gamma_nZ_{\aleph_s,l}$, we have 
\begin{align*}
\Gamma_nZ_{\aleph_s,l} \stackrel{d}{=} W_n^{1/2} Z_{\aleph_{1,s},l},
\end{align*}
where $Z_{\aleph_{1,s},l}$ is an $n_{1,s} \times 1$ standard normal vector and $W_n$ is a $n_{1,s} \times n_{1,s}$ matrix with $1-1/n_s$ and $-1/n_s$ being the diagonal and off-diagonal elements, respectively. Given that $Z_{\aleph_s}$ are independent across columns, we have
\begin{align*}
\Gamma_nZ_{\aleph_s} \stackrel{d}{=} W_n^{1/2} Z_{\aleph_{1,s}},
\end{align*}
which implies 
\begin{align}\label{eq:gamma_1sn}
1- \gamma_{1,s,n} & \stackrel{d}{=} \frac{1_{n_{1,s}}^\top  \Gamma_n Z_{\aleph_s}}{n_{1,s}} \left[\frac{Z_{\aleph_s}^\top \Gamma_n^\top \Gamma_n Z_{\aleph_s}}{n_{1,s}}\right]^{-1}   \frac{Z_{\aleph_s}^\top\Gamma_n^\top1_{n_{1,s}}}{n_{1,s}} \notag \\
& \stackrel{d}{=} \frac{1_{n_{1,s}}^\top  W_n^{1/2} Z_{\aleph_{1,s}}}{n_{1,s}} \left[\frac{Z_{\aleph_{1,s}}^\top W_n Z_{\aleph_{1,s}}}{n_{1,s}}\right]^{-1}   \frac{Z_{\aleph_{1,s}}^\top W_n^{1/2}1_{n_{1,s}}}{n_{1,s}} \notag \\
& \stackrel{d}{=} \frac{1_{n_{1,s}}^\top U_n D_n^{1/2} Z_{\aleph_{1,s}}}{n_{1,s}} \left[\frac{Z_{\aleph_{1,s}}^\top D_n Z_{\aleph_{1,s}}}{n_{1,s}}\right]^{-1}   \frac{Z_{\aleph_{1,s}}^\top D_n^{1/2} U_n^\top 1_{n_{1,s}}}{n_{1,s}} \notag \\
&= \frac{(1-n_{1,s}/n_s)}{n_{1,s}} Z_{n_{1,s}}^\top \left[\frac{1}{n_{1,s}}\sum_{i \in [n_{1,s}-1] }Z_iZ_i^\top + \left(\frac{1}{n_{1,s}} - \frac{1}{n_s}\right) Z_{n_{1,s}}Z_{n_{1,s}}^\top \right]^{-1} Z_{n_{1,s}} \notag \\
& = \frac{(1-n_{1,s}/n_s)}{n_{1,s}} Z_{n_{1,s}}^\top \left[\mathcal{A}_n^{-1} - \frac{\mathcal{A}_n^{-1} ( \left(\frac{1}{n_{1,s}}-\frac{1}{n_s}\right) Z_{n_{1,s}}Z_{n_{1,s}}^\top )\mathcal{A}_n^{-1} }{1+ \left(\frac{1}{n_{1,s}}-\frac{1}{n_s}\right) Z_{n_{1,s}}^\top \mathcal{A}_n^{-1}Z_{n_{1,s}}}  \right] Z_{n_{1,s}} ,
\end{align}
where $W_n = U_n D_n U_n^\top$ is the eigenvalue decomposition of $W_n$, $D_n = \diag(1,\cdots,1,1-n_{1,s}/n_s)$ is the eigenvalues, $U_n$ is the corresponding eigenvectors, $\mathcal{A}_n = \frac{1}{n_{1,s}}\sum_{i \in [n_{1,s}-1] }Z_iZ_i^\top$, the fourth equality is by $1_{n_{1,s}}^\top U_n = (0,\cdots,0,\sqrt{n_{1,s}})$, and the last equality follows the Sherman-Morrison matrix inverse formula.  Consider the eigenvalue decomposition of $\frac{n_{1,s}}{n_{1,s}-1}\mathcal{A}_n = \mathcal{U}_n \Lambda_n \mathcal{U}_n^\top$ where $\Lambda_n = \diag(\lambda_1,\cdots,\lambda_{k_n})$.  We note that $Z_{n_{1,s}}$ and $\mathcal{A}_n$, and thus, $\mathcal U_n$ and $\Lambda_n$, are independent by construction. Then, for any bounded and smooth function $f: \Re^{k_n} \mapsto \Re$, we have
\begin{align*}
\mathbb E [f(\mathcal{U}_n^\top Z_{n_{1,s}} \mid \Lambda_n] =   \mathbb E [ \mathbb E (f(\mathcal{U}_n^\top Z_{n_{1,s}}\mid \mathcal U_n, \Lambda_n) \mid \Lambda_n ] = \mathbb 
E f(Z_{n_{1,s}}) = \mathbb E f(\mathcal{U}_n^\top Z_{n_{1,s}}),
\end{align*}
where we use the fact that the distribution of $Z_{n_{1,s}}$ ($k_n \times 1$ standard normal random vector) is invariant to rotations. This implies $\mathcal{U}_n^\top Z_{n_{1,s}}$ and $\Lambda_n$ are independent. Denote $\mathcal{U}_n^\top Z_{n_{1,s}}$ as $\mathcal G = (g_1,\cdots,g_{K_n})^\top$ which are independent of $\Lambda_n$ and suppose $k_n/n_{1,s} \rightarrow \kappa_{1,s}$. Then, we have
\begin{align}\label{eq:A11}
\frac{1}{n_{1,s}} Z_{n_{1,s}}^\top \left(\frac{n_{1,s}}{n_{1,s}-1} \mathcal{A}_n\right)^{-1}Z_{n_{1,s}} & \stackrel{d}{=} \frac{1}{n_{1,s}}\sum_{l = 1}^{k_n} \lambda_l^{-1} g_l^2 \notag \\
& = \frac{1}{n_{1,s}}\sum_{l = 1}^{k_n} \lambda_l^{-1} + \frac{1}{n_{1,s}}\sum_{l = 1}^{k_n} \lambda_l^{-1} (g_l^2-1) \notag \\
& = \frac{k_n}{n_{1,s}} \frac{1}{k_n}\sum_{l = 1}^{k_n} \lambda_l^{-1} + o_P(1) \notag \\
& \convP  \int_{\lambda_-}^{\lambda_+} \frac{\sqrt{(\lambda_+ - \lambda)(\lambda - \lambda_-)}}{2\pi\lambda^2} d\lambda, 
\end{align}
where $\lambda_{\pm} = (1 \pm \sqrt{\kappa_{1,s}})^2$. We note that the third equality in \eqref{eq:A11} holds because $\lambda_{k_n} \convP 1 - \sqrt{\kappa_{1,s}}$ which is bounded away from zero, 
\begin{align*}
\mathbb E\left(\frac{1}{n_{1,s}}\sum_{l = 1}^{k_n} \lambda_l^{-1} (g_l^2-1) \mid \Lambda_n\right) = 0, 
\end{align*}
and 
\begin{align*}
Var \left(\frac{1}{n_{1,s}}\sum_{l = 1}^{k_n} \lambda_l^{-1} (g_l^2-1) \mid \Lambda_n\right) \leq \frac{k_n}{n_{1,s}^2\lambda_{k_n}^{2}}   \rightarrow 0,   
\end{align*}
and the last convergence in probability \eqref{eq:A11} holds by the Mar\v{c}enko-Pastur theorem (see, for example, \cite{bai08}). Denote $\int_{\lambda_-}^{\lambda_+} \frac{\sqrt{(\lambda_+ - \lambda)(\lambda - \lambda_-)}}{2\pi\lambda^2} d\lambda$ as $\zeta_{1,s}$, then we have
\begin{align*}
\frac{1}{n_{1,s}} Z_{n_{1,s}}^\top  \mathcal{A}_n^{-1}Z_{n_{1,s}} \convP \int_{\lambda_-}^{\lambda_+} \frac{\sqrt{(\lambda_+ - \lambda)(\lambda - \lambda_-)}}{2\pi\lambda^2} d\lambda,
\end{align*}
and thus, by \eqref{eq:gamma_1sn}, 
\begin{align*}
\gamma_{1,s,n} \convP \frac{1}{1+ (1-\pi_s)\zeta_{1,s}} \equiv \gamma_{1,s,\infty}.
\end{align*}

\section{Proof of Theorem \ref{thm:main}}
\label{sec:proof_main}
Denote $\hat \tau_s = \hat \tau_{1,s} - \hat \tau_{0,s}$, $\tau_{a,s} = \mathbb{E}(Y(a)|S=s)$ for $a=0,1$, and $\tau_s = \mathbb{E}(Y(1)-Y(0)|S=s)$. Note that 
\begin{align*}
Y_{i}(1) & = \alpha_{1,S_i} +X_i^\top \beta_{1,S_i} + e_{i,S_i}(1) + \eps_{i}(1) \\
& = \alpha_{1,S_i} + \overline{X}_{S_i}^\top\beta_{1,S_i} +  \breve X_i^\top \beta_{1,S_i} + e_{i,S_i}(1) + \eps_{i}(1), 
\end{align*}
which implies 
\begin{align*}
\hat \tau_{1,s} & = (1_{n_{1,s}}^\top M_{1,s} 1_{n_{1,s}})^{-1}(1_{n_{1,s}}^\top M_{1,s} Y_{\aleph_{1,s}}(1))\\
& = \alpha_{1,s} + \overline{X}_{s}^\top\beta_{1,s} + (1_{n_{1,s}}^\top M_{1,s} 1_{n_{1,s}})^{-1}(1_{n_{1,s}}^\top M_{1,s} (e_{\aleph_{1,s}}(1) + \eps_{\aleph_{1,s}}(1))) \\
& = \frac{1}{n_{s}}\sum_{i \in [n]}1\{S_i=s\}\left(\mathbb{E}(Y_i(1)|X_i,S_i=s)- e_{i,s}(1)\right) \\
&+ (1_{n_{1,s}}^\top M_{1,s} 1_{n_{1,s}})^{-1}(1_{n_{1,s}}^\top M_{1,s} (e_{\aleph_{1,s}}(1) + \eps_{\aleph_{1,s}}(1))).
\end{align*}


By Assumption \ref{ass:linear}(iv), we have  
\begin{align*}
& \left\vert \frac{1}{n_{s}}\sum_{i \in [n]}1\{S_i=s\} e_{i,s}(1)\right\vert \leq \left( \frac{1}{n_{s}}\sum_{i \in [n]}1\{S_i=s\} e_{i,s}^2(1)\right)^{1/2}  = o_P(n^{-1/2})
\end{align*}
and 
\begin{align*}
\left\vert n_{1,s}^{-1}(1_{n_{1,s}}^\top M_{1,s} e_{\aleph_{1,s}}(1))\right\vert& \leq n_{1,s}^{-1} ||M_{1,s}||_{op} ||1_{n_{1,s}}||_2 ||e_{\aleph_{1,s}}(1)||_2 \\
& \leq \left( \frac{1}{n_{s}}\sum_{i \in [n]}1\{S_i=s\} e_{i,s}^2(1)\right)^{1/2} = o_P(n^{-1/2}).
\end{align*}

This implies the linear expansion 
\begin{align*}
\hat \tau_{1,s} - \tau_{1,s} & = \frac{1}{n_{s}}\sum_{i \in [n]}1\{S_i=s\}\left(\mathbb{E}(Y_i(1)|X_i,S_i=s)- \mathbb{E}(Y_i(1)|S_i=s)\right) \\
&+ \gamma_{1,s,n}^{-1} n_{1,s}^{-1}(1_{n_{1,s}}^\top M_{1,s}  \eps_{\aleph_{1,s}}(1)) + o_P(n^{-1/2}).
\end{align*}

Similarly, we have 
\begin{align*}
\hat \tau_{0,s} - \tau_{0,s} & = \frac{1}{n_{s}}\sum_{i \in [n]}1\{S_i=s\}\left(\mathbb{E}(Y_i(0)|X_i,S_i=s)- \mathbb{E}(Y_i(0)|S_i=s)\right) \\
&+ \gamma_{0,s,n}^{-1} n_{0,s}^{-1}(1_{n_{0,s}}^\top M_{0,s}  \eps_{\aleph_{0,s}}(0)) + o_P(n^{-1/2}),
\end{align*}
and thus
\begin{align*}
\hat \tau_s - \tau_s & = \frac{1}{n_{s}}\sum_{i \in [n]}1\{S_i=s\}\left[\mathbb{E}(Y_i(1)|X_i,S_i=s)-\mathbb{E}(Y_i(0)|X_i,S_i=s) - \tau_s\right] \\
&+ \gamma_{1,s,n}^{-1} n_{1,s}^{-1}(1_{n_{1,s}}^\top M_{1,s}  \eps_{\aleph_{1,s}}(1)) - \gamma_{0,s,n}^{-1} n_{0,s}^{-1}(1_{n_{0,s}}^\top M_{0,s}  \eps_{\aleph_{0,s}}(0)) + o_P(n^{-1/2}).
\end{align*}

Therefore, we have
\begin{align*}
\hat \tau^{adj} -\tau & = \sum_{s\in \mathcal{S}}\hat p_s \left[\tilde \gamma_{1,s}^{-1} n_{1,s}^{-1}(1_{n_{1,s}}^\top M_{1,s}  \eps_{\aleph_{1,s}}(1)) - \tilde \gamma_{1,s}^{-1} n_{0,s}^{-1}(1_{n_{0,s}}^\top M_{0,s}  \eps_{\aleph_{0,s}}(0))\right] \\
& + \sum_{s \in \mathcal{S}}\frac{\hat p_s}{n_{s}}\sum_{i \in [n]}1\{S_i=s\} \left[\mathbb{E}(Y_i(1)|X_i,S_i=s)- \mathbb{E}(Y_i(1)|S_i=s) \right] \\
& - \sum_{s \in \mathcal{S}}\frac{\hat p_s}{n_{s}}\sum_{i \in [n]}1\{S_i=s\} \left[\mathbb{E}(Y_i(0)|X_i,S_i=s) -\mathbb{E}(Y_i(0)|S_i=s)\right] \\
& + \frac{1}{n}\sum_{i \in [n]}\left[\mathbb{E}(Y_i(1)|S_i)-\mathbb{E}(Y_i(0)|S_i) - \tau\right] + o_P(n^{-1/2}).
\end{align*}

In addition, we have 
\begin{align*}
\hat \tau^{unadj} - \tau & =  \sum_{s \in \mathcal{S}}  \frac{\hat p_s}{n_{1,s}}\sum_{i \in [n]}A_i1\{S_i=s\}\eps_{i}(1)  - \sum_{s \in \mathcal{S}}  \frac{\hat p_s}{n_{0,s}}\sum_{i \in [n]}(1-A_i)1\{S_i=s\} \eps_{i}(0) \\
& + \sum_{s \in \mathcal{S}} \frac{\hat p_s}{n_{1,s}}\sum_{i \in [n]}A_i1\{S_i=s\}(\mathbb{E}(Y_i(1)|X_i,S_i=s) - \mathbb{E}(Y_i(1)|S_i=s)) \\
&  - \sum_{s \in \mathcal{S}} \frac{\hat p_s}{n_{0,s}}\sum_{i \in [n]}(1-A_i)1\{S_i=s\}(\mathbb{E}(Y_i(0)|X_i,S_i=s) - \mathbb{E}(Y_i(0)|S_i=s)) \\
&+ \frac{1}{n}\sum_{i \in [n]}(\mathbb{E}(Y_i(1)|S_i) - \mathbb{E}(Y_i(0)|S_i) - \tau).
\end{align*}

Therefore, we have
\begin{align*}
\sqrt{n}\begin{pmatrix}
\hat \tau^{adj} - \tau \\
\hat \tau^{unadj} - \tau
\end{pmatrix} = \mathcal U_n + \mathcal V_n + \begin{pmatrix}
\mathcal W_n \\
\mathcal W_n
\end{pmatrix} + o_P(1),
\end{align*}
where 
\begin{align}
\mathcal U_{n} = \sum_{s \in \mathcal{S}} \hat p_s \sqrt{n} \begin{pmatrix}
\gamma_{1,s,n}^{-1} n_{1,s}^{-1}(1_{n_{1,s}}^\top M_{1,s}  \eps_{\aleph_{1,s}}(1)) - 	\gamma_{0,s,n}^{-1} n_{0,s}^{-1}(1_{n_{0,s}}^\top M_{0,s}  \eps_{\aleph_{0,s}}(0))\\
n_{1,s}^{-1}(1_{n_{1,s}}^\top \eps_{\aleph_{1,s}}(1)) - n_{0,s}^{-1}(1_{n_{0,s}}^\top \eps_{\aleph_{0,s}}(0))
\end{pmatrix},
\label{eq:Un}
\end{align}
\begin{align}
\mathcal V_{n} = \sum_{s \in \mathcal{S}} \hat p_s \sqrt{n} \begin{pmatrix}
\frac{1}{n_{s}}\sum_{i \in [n]}1\{S_i=s\} \left(\phi_i(1)-\phi_i(0)\right) \\
\frac{1}{n_{1,s}}\sum_{i \in \aleph_{1,s} }\phi_i(1) -
\frac{1}{n_{0,s}}\sum_{i \in \aleph_{0,s}}\phi_i(0) 
\end{pmatrix},
\label{eq:Vn}
\end{align}
\begin{align}
\mathcal W_n = \frac{1}{\sqrt{n}}\sum_{i \in [n]}(\mathbb{E}(Y_i(1)|S_i) - \mathbb{E}(Y_i(0)|S_i) - \tau),
\label{eq:Wn}
\end{align}
and $\phi_i(a) =  \mathbb{E}(Y_i(a)|X_i,S_i)-\mathbb{E}(Y_i(a)|S_i)$. Then, the desired result holds by Lemma \ref{lem:clt}.

\section{Proof of Theorem \ref{thm:further}}
We note that 
\begin{align*}
\hat w \convP \frac{\Sigma_{2,2} - \Sigma_{1,2}}{\Sigma_{1,1}+\Sigma_{2,2} - 2\Sigma_{1,2}}   = \argmin_w (w,1-w) \Sigma (w,1-w)^\top.
\end{align*}
This implies $\hat \tau^*$ is weakly more efficient than both $\hat \tau^{adj}$ and $\hat \tau^{adj}$ which correspond to $\hat w = 1$ and 0, respectively. Then, by the continuous mapping theorem, we have
\begin{align*}
\sqrt{n}\left[ (\hat w, 1- \hat w) \hat \Sigma (\hat w, 1- \hat w)^\top\right]^{-1/2}(\hat \tau^* - \tau) \convD \N(0, 1).
\end{align*}

\section{Proof of Corollary \ref{cor:umpu}}
For the weak convergence, we note that 
\begin{align*}
\sqrt n (\hat \tau^* - \tau_0) = (\hat \omega, 1-\hat \omega) \Sigma^{1/2} \begin{pmatrix}
\hat N_1 \\
\hat N_2
\end{pmatrix} \convD ( \omega^*, 1-\omega^*) \Sigma^{1/2} \begin{pmatrix}
N_1 \\
N_2
\end{pmatrix}.
\end{align*}
Recall $(a,b) = (1,1)\Sigma^{-1/2}$. Given that both $N_1^* = \frac{aN_1 + bN_2}{\sqrt{a^2+b^2}}$ and the limit distribution of \\$  \sqrt{n}\left[ (\hat w, 1- \hat w) \hat \Sigma (\hat w, 1- \hat w)^\top\right]^{-1/2}(\hat \tau^* - \tau_0)$ have variance 1, it suffices to show that $a/b = \tilde a/\tilde b$ where $(\tilde a,\tilde b) = ( \omega^*, 1- \omega^*) \Sigma^{1/2}$ and $\omega^* = (\Sigma_{2,2} - \Sigma_{1,2})/(\Sigma_{1,1}+\Sigma_{2,2}-2\Sigma_{1,2})$. To see this, we note that 
\begin{align*}
\Sigma^{1/2} = \frac{1}{t}\begin{pmatrix}
\Sigma_{1,1} + \eta &  \Sigma_{1,2} \\
\Sigma_{1,2}  &  \Sigma_{2,2} + \eta
\end{pmatrix} \quad \text{and} \quad     \Sigma^{-1/2} = \frac{t}{(\Sigma_{1,1} + \eta)(\Sigma_{2,2} + \eta)-\Sigma^{2}_{1,2} }\begin{pmatrix}
\Sigma_{2,2} + \eta & - \Sigma_{1,2} \\
- \Sigma_{1,2}  &  \Sigma_{1,1} + \eta
\end{pmatrix},
\end{align*}
where $\eta = \sqrt{\Sigma_{1,1}\Sigma_{2,2} - \Sigma_{1,2}^2}$ and $t = \sqrt{\Sigma_{1,1} + \Sigma_{2,2} + 2\eta}$. Then, it is obvious that 
\begin{align*}
a & = \frac{t(\Sigma_{2,2} - \Sigma_{1,2} + \eta)}{(\Sigma_{1,1} + \eta)(\Sigma_{2,2} + \eta)-\Sigma^{2}_{1,2} }, \\
b & = \frac{t(\Sigma_{1,1} - \Sigma_{1,2} + \eta)}{(\Sigma_{1,1} + \eta)(\Sigma_{1,1} + \eta)-\Sigma^{2}_{1,2} }, \\
\tilde a & = \frac{(\Sigma_{2,2} - \Sigma_{1,2})(\Sigma_{1,1} + \eta) + (\Sigma_{1,1} - \Sigma_{1,2})\Sigma_{1,2}}{t (\Sigma_{1,1}+\Sigma_{2,2}-2\Sigma_{1,2})} \\
& = \frac{(\Sigma_{2,2}\Sigma_{1,1} - \Sigma_{1,2}^2) + \eta(\Sigma_{2,2} - \Sigma_{1,2}) }{t (\Sigma_{1,1}+\Sigma_{2,2}-2\Sigma_{1,2})} \\
& = \frac{\eta (\eta + \Sigma_{2,2} - \Sigma_{1,2}) }{t (\Sigma_{1,1}+\Sigma_{2,2}-2\Sigma_{1,2})}, \quad \text{and} \\
\tilde b & = \frac{\eta (\eta + \Sigma_{1,1} - \Sigma_{1,2}) }{t (\Sigma_{1,1}+\Sigma_{2,2}-2\Sigma_{1,2})},
\end{align*}
which implies $a/b = \tilde a/ \tilde b$. 

For the optimality result, note that 
\begin{align*}
(\mathbb W_n,\psi(\hat N_1,\hat N_2)) \convD (1\{(N^*)^2 \geq \mathcal C_{\alpha}\},\psi(N_1,N_2)),    
\end{align*}
where $(N_1,N_2)^\top \stackrel{d}{=}\N\left( \begin{pmatrix}
a\Delta \\
b\Delta
\end{pmatrix}, I_2 \right)$ and $N^* = \frac{aN_1 + bN_2}{\sqrt{a^2 + b^2}}$. Consider a limit experiment in which a researcher observes $(N_1,N_2)$,
knows the values of $(a,b)$, and wants to test $\Delta = 0$ versus the two-sided alternative. \citet[Section 4.2]{LR06} show $1\{(N^*)^2 \geq \mathcal C_{\alpha}\}$ is the UMP unbiased test for this limit experiment, implying that 
\begin{align*}
\mathbb E 1\{(N^*)^2 \geq \mathcal C_{\alpha}\} \geq \mathbb E \psi(N_1,N_2)
\end{align*}
for any $\psi \in \Psi_{\alpha}^U.$ This implies 
\begin{align*}
\lim_{n \rightarrow \infty}    \mathbb{E}\mathbb W_n =\mathbb E 1\{(N^*)^2 \geq \mathcal C_{\alpha}\}
= \sup_{\psi \in \Psi_{\alpha}^U } \mathbb E \psi(N_1,N_2) = \sup_{\psi \in \Psi_{\alpha}^U } \lim_{n \rightarrow \infty} \mathbb{E}\psi(\hat N_1,\hat N_2).
\end{align*}
The inequality that 
\begin{align*}
\sup_{\psi \in \Psi_{\alpha}^U } \lim_{n \rightarrow \infty} \mathbb{E}\psi(\hat N_1,\hat N_2) \geq  \lim_{n \rightarrow \infty} \mathbb{E}\breve{\psi}_n
\end{align*}
holds due to the uniform integrability of $\breve \psi_n$, the fact that $\breve \psi_n = \psi(\hat N_1,\hat N_2) + o_P(1)$ for some $\psi \in \Psi^U_\alpha$, and the dominated convergence theorem.

\section{Proof of Theorem \ref{thm:variance}}
We derive the limits of $\hat \Sigma_{\mathcal U}$, $\hat \Sigma_{\mathcal V}$, and $\hat \Sigma_{\mathcal W}$ in the following three steps. 

\textbf{Step 1: Limit of $\hat \Sigma_{\mathcal U}$}. Following \cite{BCS17}, we define $\{(X_i^s,\eps_i^s(1),\eps_i^s(0)): 1\leq i \leq n\}$ as a sequence of i.i.d. random variables with marginal distributions equal to the distribution of $(X_i,\eps_i(1),\eps_i(0))|S_i = s$ and $N_s = \sum_{i =1}^n1\{S_i <s\}$. We order units by strata and then by $A_i = 1$ first and $A_i = 0$ second within each stratum. This means for the $s$th stratum, units indexed from $N_s+1$ to $N_s + n_{1,s}$ are treated and units indexed from $N_s + n_{1,s}+1$ to $N_s + n_s$ are untreated. Then, we have $(Y_i(a)\mid S_i=s) \stackrel{d}{=} Y_i^s(a)$ and $(\tilde e_i(a)\mid S_i=s) \stackrel{d}{=} e_{i,s}(a)$ where 
\begin{align*}
Y_i^s(a) = \alpha_{a,s} + (X_i^s)^\top \beta_{a,s} + \tilde e_i^s(a) + \eps_i^s(a) \quad \text{and} \quad \tilde e_i^s(a) =  \mathbb{E}(Y_i(a)|X_i=x,S_i=s)  - \alpha_{a,s} - (X_i^s)^\top \beta_{a,s}. 
\end{align*}
We further define $\tilde M_{a,s,i,j}$ as the $(i,j)$th entry of the $n_{a,s} \times n_{a,s}$ matrix,  
$\tilde M_{a,s} =I_{n_{a,s}}- \Xi_{a,s} (\Xi_{a,s}^\top  \Xi_{a,s})^{-1} \Xi_{a,s}^\top $, $\Xi_{1,s} = ((\breve X_{N_s+1}^s)^\top,\cdots,( \breve X_{N_s+n_{1,s}}^s)^\top)^\top$ is an $n_{1,s} \times k_n$ matrix, $\Xi_{0,s} = ((\breve X_{N_s+1+n_{1,s}}^s)^\top,\cdots,(\breve X_{N_s+n_{s}}^s)^\top)^\top$ is an $n_{0,s} \times k_n$ matrix,
\begin{align*}
\breve X_i^s = 	X_i^s - \frac{1}{n_{s}}\sum_{j=N_s+1}^{N_s+n_{s}}X_j^s, \quad \text{if}\quad N_s+1 \leq i \leq N_s+n_{s},
\end{align*}
$\tilde \gamma_{a,s,n} = 1_{n_{a,s}}^\top \tilde M_{a,s} 1_{n_{a,s}}$, and $\phi_i^s(a) = \mathbb{E}(Y_i(a)|X_i,S_i=s)- \mathbb{E}(Y_i(a)|S_i=s)$.

\textbf{Step 1.1: Limit of $\hat \omega_{a,s}^2$.} We consider the case with $a=1$. The result for $a=0$ can be established in the same manner. Given the above definitions, we see that, conditionally on $(A^{(n)},S^{(n)})$, and thus, unconditionally,  
\begin{align*}
\hat \omega_{1,s}^2 \stackrel{d}{=} \frac{1}{n_{1,s}}\tilde \gamma_{1,s,n}^{-2}\sum_{i=N_s+1}^{N_s + n_{1,s}} \left(\sum_{j=N_s+1}^{N_s + n_{1,s}} \tilde M_{1,s,i,j} \right)^2 Y_i^s(1) \tilde {\acute {\eps}}_{1,s,i},
\end{align*}
where $\tilde {\acute {\eps}}_{1,s,i} = \tilde {\eps}_{1,s,i}/ \tilde M_{1,s,i,i}$ and $\tilde {\eps}_{1,s,i}$ is the residual from the linear regression of $Y_i^s(1)$ on $(1,\breve X_i^s)$ with observations $N_s+1 \leq i \leq N_s + n_{1,s}$. Because $\{(X_i^s,Y_i^s(1)): 1\leq i \leq n\}$ are i.i.d. and independent of $(A^{(n)},S^{(n)})$ we can directly apply \citet[Theorem 1]{J22} and conclude that  
\begin{align*}
\hat \omega_{1,s}^2 - \tilde \gamma_{1,s,n}^{-2} \sigma_{1,s,n}^2 = o_P(1). 
\end{align*}
Then, because $\tilde \gamma_{a,s,n} \stackrel{d}{=}\gamma_{a,s,n}$ and Assumption \ref{ass:omega} holds, we have
\begin{align*}
\hat \omega_{1,s}^2 \convP \omega_{1,s,\infty}^2.
\end{align*}

\textbf{Step 1.2: Limit of $\hat \varpi_{a,s}$.} Again, we focus on $\hat \varpi_{1,s}$. Following the argument in \textbf{Step 1.1}, we have
\begin{align*}
\hat \varpi_{1,s} \stackrel{d}{=} \frac{1}{n_{1,s}}\tilde \gamma_{1,s,n}^{-1}\sum_{i=N_s+1}^{N_s + n_{1,s}} \left(\sum_{j=N_s+1}^{N_s + n_{1,s}} \tilde M_{1,s,i,j} \right) Y_i^s(1) \tilde {\acute {\eps}}_{1,s,i}.
\end{align*}
Note that 
\begin{align*}
\tilde {\acute {\eps}}_{1,s,i}  & = \frac{\sum_{k=N_s+1}^{N_s + n_{1,s}} \tilde M_{1,s,i,k} (Y_k^s(1)-\tilde \tau_{1,s})}{\tilde M_{1,s,i,i}} \\
& = \frac{\sum_{k=N_s+1}^{N_s + n_{1,s}} \tilde M_{1,s,i,k} (\tilde e_k^s(1) +  \alpha_{1,s} + \widetilde X_s^\top \beta_{1,s} -\tilde \tau_{1,s} + \eps_k^s(1))}{\tilde M_{1,s,i,i}},
\end{align*}
where $\tilde \tau_{1,s}$ is the intercept of the OLS regression of $Y_i^s(1)$ on $(1,\breve X_i^s)$ using observations $N_s + 1 \leq i \leq N_s + n_{1,s}$ and $\widetilde X_s = \frac{1}{n_{s}} \sum_{i=N_s+1}^{N_s + n_{s}}X_i^s$. By construction, we have $\tilde \tau_{1,s} \stackrel{d}{=} \hat \tau_{1,s}$. Let $\theta_i = \sum_{k=N_s+1}^{N_s + n_{1,s}} \tilde M_{1,s,i,k}$. 
Therefore, we have 
\begin{align*}
& \frac{1}{n_{1,s}}\tilde \gamma_{1,s,n}^{-1}\sum_{i=N_s+1}^{N_s + n_{1,s}} \left(\sum_{j=N_s+1}^{N_s + n_{1,s}} \tilde M_{1,s,i,j} \right) Y_i^s(1) \tilde {\acute {\eps}}_{1,s,i} \\
& = \frac{1}{n_{1,s}}\tilde \gamma_{1,s,n}^{-1} \sum_{i=N_s+1}^{N_s + n_{1,s}}\sum_{k=N_s+1}^{N_s + n_{1,s}}  \frac{\theta_i Y_i^s(1)  \tilde  M_{1,s,i,k} \eps_k^s(1)}{\tilde M_{1,s,i,i}} + \frac{1}{n_{1,s}}\tilde \gamma_{1,s,n}^{-1} \sum_{i=N_s+1}^{N_s + n_{1,s}}\sum_{k=N_s+1}^{N_s + n_{1,s}}  \frac{\theta_i Y_i^s(1)  \tilde  M_{1,s,i,k} \tilde e_k^s(1)}{\tilde M_{1,s,i,i}} \\
& +  \frac{1}{n_{1,s}}\tilde \gamma_{1,s,n}^{-1}\sum_{i=N_s+1}^{N_s + n_{1,s}}     \frac{\theta_i^2 Y_i^s(1)}{\tilde M_{1,s,i,i}}(\alpha_{1,s} + \widetilde X_s^\top \beta_{1,s} -\tilde \tau_{1,s})) \\
& \equiv T_1 + T_2 + T_3. 
\end{align*}
For $T_2$, we have
\begin{align*}
|T_2| & \leq \tilde \gamma_{1,s,n}^{-1} \left[\frac{1}{n_{1,s}} \sum_{i=N_s+1}^{N_s + n_{1,s}} \frac{\theta_i^2 (Y_i^s(1))^2}{\tilde M_{1,s,i,i}^2}  \right]^{1/2} \left[\frac{1}{n_{1,s}} \sum_{i=N_s+1}^{N_s + n_{1,s}} (\tilde e_i^s(1))^2  \right]^{1/2} \\
& \leq  \gamma_{1,s,n}^{-1} \frac{\max_{ N_s+1 \leq i \leq N_s + n_{1,s}} |\theta_i| }{\min_{ N_s+1 \leq i \leq N_s + n_{1,s}} \tilde M_{1,s,i,i} } \left[\frac{1}{n_{1,s}} \sum_{i=N_s+1}^{N_s + n_{1,s}} (Y_i^s(1))^2 \right]^{1/2} \left[\frac{1}{n_{1,s}} \sum_{i=N_s+1}^{N_s + n_{1,s}} (\tilde e_i^s(1))^2  \right]^{1/2}\\
& = o_P(1),
\end{align*}
where the first inequality is by $||\tilde M_{1,s}||_{op} \leq 1$ and the last equality holds because  by Assumption \ref{ass:linear}, we have 
\begin{align*}
& \min_{ N_s+1 \leq i \leq N_s + n_{1,s} } \tilde M_{1,s,i,i} \stackrel {d}{=}  \min_{ N_s+1 \leq i \in \aleph_{1,s} }M_{1,s,i,i} \geq \delta>0, \\
& \frac{1}{n_{1,s}} \sum_{i=N_s+1}^{N_s + n_{1,s}} (\tilde e_i^s(1))^2 \stackrel{d}{=}\frac{1}{n_{1,s}} \sum_{i \in \aleph_{1,s}} e_{i,s}^2(1) = o_P(n^{-1}), \quad \text{and}\\
& \max_{ N_s+1 \leq i \leq N_s + n_{1,s}}|\theta_i| \stackrel{d}{=} \max_{i \in \aleph_{1,s}} |\sum_{j \in \aleph_{1,s}} M_{1,s,i,j}| = o_P(n^{1/2}).
\end{align*}
For $T_3$, we have
\begin{align*}
|T_3| & \leq \tilde \gamma_{1,s,n}^{-1} \left[\frac{1}{n_{1,s}}\sum_{i=N_s+1}^{N_s + n_{1,s}}     \frac{\theta_i^4}{\tilde M_{1,s,i,i}^2}\right]^{1/2} \left[\frac{1}{n_{1,s}}\sum_{i=N_s+1}^{N_s + n_{1,s}}     (Y_i^s(1))^2\right]^{1/2} |\alpha_{1,s} + \widetilde X_s^\top \beta_{1,s} -\tilde \tau_{1,s}| \\
& \leq \gamma_{1,s,n}^{-1} \frac{\max_{ N_s+1 \leq i \leq N_s + n_{1,s}} |\theta_i| }{\min_{ N_s+1 \leq i \leq N_s + n_{1,s}} \tilde M_{1,s,i,i} } \left[\frac{1}{n_{1,s}}\sum_{i=N_s+1}^{N_s + n_{1,s}}  \theta_i^2\right]^{1/2} \\
& \times  \left[\frac{1}{n_{1,s}}\sum_{i=N_s+1}^{N_s + n_{1,s}}     (Y_i^s(1))^2\right]^{1/2} |\alpha_{1,s} + \widetilde X_s^\top \beta_{1,s} -\tilde \tau_{1,s})| \\
& = o_P(1),
\end{align*}
where the last equality holds because 
\begin{align*}
\frac{1}{n_{1,s}}\sum_{i=N_s+1}^{N_s + n_{1,s}}  \theta_i^2 = \frac{1}{n_{1,s}}\sum_{i=N_s+1}^{N_s + n_{1,s}}\sum_{j=N_s+1}^{N_s + n_{1,s}}  \tilde M_{1,s,i,j} 
= \frac{1}{n_{1,s}}1_{n_{1,s}}^\top \tilde M_{1,s} 1_{n_{1,s}} \leq ||\tilde M_{1,s}||_{op} \leq 1
\end{align*}
and 
\begin{align*}
|\alpha_{1,s} + \widetilde X_s^\top \beta_{1,s} -\tilde \tau_{1,s})| \stackrel{d}{=}  |\alpha_{1,s} + \overline X_s^\top \beta_{1,s} -\hat  \tau_{1,s})| & = |(1_{n_{1,s}}^\top M_{1,s} 1_{n_{1,s}})^{-1}(1_{n_{1,s}}^\top M_{1,s}\eps_{\aleph_{1,s}}(1))| + o_P(n^{-1/2})\\ 
&= O_P(n^{-1/2}).   
\end{align*}

Next, we focus on $T_1$. Denote $\mu_i^s(1) = \mathbb{E}(Y_i^s(1)|X_i^s) \stackrel{d}{=} \mathbb{E}(Y_i(1)|X_i,S_i=s)$. Then, we have
\begin{align*}
T_1 & = \frac{1}{n_{1,s}}\tilde \gamma_{1,s,n}^{-1} \sum_{i=N_s+1}^{N_s + n_{1,s}}\sum_{k=N_s+1}^{N_s + n_{1,s}}  \frac{\theta_i \mu_i^s(1)    M_{1,s,i,k} \eps_k^s(1)}{\tilde M_{1,s,i,i}} \\
& + \frac{1}{n_{1,s}}\tilde \gamma_{1,s,n}^{-1} \sum_{i=N_s+1}^{N_s + n_{1,s}}\sum_{N_s+1 \leq k \leq N_s + n_{1,s},k\neq i}  \frac{\theta_i \eps_i^s(1)    M_{1,s,i,k} \eps_k^s(1)}{\tilde M_{1,s,i,i}} \\
& + \frac{1}{n_{1,s}}\tilde \gamma_{1,s,n}^{-1} \sum_{i=N_s+1}^{N_s + n_{1,s}}  \theta_i [(\eps_i^s(1))^2 - \tilde V_{1,s,i}^2] + \frac{1}{n_{1,s}}\tilde \gamma_{1,s,n}^{-1} \sum_{i=N_s+1}^{N_s + n_{1,s}}  \theta_i \tilde V_{1,s,i}^2 \\
& \equiv T_{1,1}+T_{1,2}+T_{1,3}+T_{1,4},
\end{align*}
where $\tilde V_{1,s,i}^2 = \mathbb{E}[(\eps_i^s(1))^2|X_i^s] \stackrel{d}{=} \mathbb{E}(\eps_i^2(1)|X_i,S_i=s)$. Recall $\Xi_{1,s} = ((\breve X_{N_s+1}^s)^\top,\cdots,( \breve X_{N_s+n_{1,s}}^s)^\top)^\top$. We have $\mathbb{E} (T_{1,1}\mid \Xi_{1,s}) = 0$ and 
\begin{align*}
Var(T_{1,1}\mid \Xi_{1,s}) & = \frac{1}{n_{1,s}^2 \tilde \gamma_{1,s,n}^2} \sum_{k=N_s+1}^{N_s + n_{1,s}} \tilde V_{1,s,k}^2 \left[ \sum_{i=N_s+1}^{N_s + n_{1,s}}\frac{\theta_i \mu_i^s(1)    M_{1,s,i,k}}{\tilde M_{1,s,i,i}} \right]^2 \\
& \leq  \frac{\max_{i \in [n]}\tilde V_{1,s,i}^2}{n_{1,s}^2 \tilde \gamma_{1,s,n}^2} \sum_{k=N_s+1}^{N_s + n_{1,s}}  \left[ \sum_{i=N_s+1}^{N_s + n_{1,s}}\frac{\theta_i \mu_i^s(1)    M_{1,s,i,k}}{\tilde M_{1,s,i,i}} \right]\left[ \sum_{j=N_s+1}^{N_s + n_{1,s}}\frac{\theta_j \mu_j^s(1)    M_{1,s,j,k}}{\tilde M_{1,s,j,j}} \right] \\
& =  \frac{\max_{i \in [n]}\tilde V_{1,s,i}^2}{n_{1,s}^2 \tilde \gamma_{1,s,n}^2} \sum_{i=N_s+1}^{N_s + n_{1,s}} \sum_{j=N_s+1}^{N_s + n_{1,s}} \left[ \frac{\theta_i \mu_i^s(1) \theta_j \mu_j^s(1)   M_{1,s,i,j}}{\tilde M_{1,s,i,i} \tilde M_{1,s,j,j}} \right] \\
& \leq \frac{\max_{i \in [n]}\tilde V_{1,s,i}^2}{n_{1,s}^2 \tilde \gamma_{1,s,n}^2}\sum_{i=N_s+1}^{N_s + n_{1,s}} \left[ \frac{\theta_i \mu_i^s(1)}{\tilde M_{1,s,i,i}} \right]^2 \\
& \leq \left[\frac{\max_{i \in [n]}\tilde V_{1,s,i}^2}{n_{1,s} \tilde \gamma_{1,s,n}^2 \min_{ N_s+1 \leq i \leq N_s + n_{1,s} } \tilde M_{1,s,i,i}^2 }\sum_{i=N_s+1}^{N_s + n_{1,s}} (\mu_i^s(1))^2 \right] \left[\frac{ \max_{ N_s+1 \leq i \leq N_s + n_{1,s}}\theta_i^2 }{n_{1,s} }\right] \\
& = o_P(1),
\end{align*}
where the second equality holds because the matrix $\tilde M_{1,s}$ is idempotent, the second inequality holds because $||\tilde M_{1,s}||_{op} \leq 1$, and the last equality holds because by Assumption \ref{ass:linear},  we have
\begin{align*}
& \min_{ N_s+1 \leq i \leq N_s + n_{1,s} } \tilde M_{1,s,i,i} \stackrel {d}{=}  \min_{ N_s+1 \leq i \in \aleph_{1,s} }M_{1,s,i,i} \geq \delta>0, \quad    \max_{i \in [n]}\tilde V_{1,s,i}^2 = O_P(1), \quad \text{and}\\
& \max_{ N_s+1 \leq i \leq N_s + n_{1,s}}|\theta_i| \stackrel{d}{=} \max_{i \in \aleph_{1,s}} \left|\sum_{j \in \aleph_{1,s}} M_{1,s,i,j}\right| = o_P(n^{1/2}).
\end{align*}
This implies $T_{1,1} = o_P(1)$.

For $T_{1,2}$, we have $\mathbb{E} (T_{1,2}\mid \Xi_{1,s}) = 0$ and 
\begin{align*}
Var(T_{1,2}\mid \Xi_{1,s}) & = \frac{1}{n_{1,s}^2}\tilde \gamma_{1,s,n}^{-2} \sum_{i=N_s+1}^{N_s + n_{1,s}}\sum_{N_s+1 \leq k \leq N_s + n_{1,s},k\neq i}   \frac{\theta_i^2  \tilde M_{1,s,i,k}^2 \tilde V_{1,s,i}^2\tilde V_{1,s,k}^2 }{\tilde M_{1,s,i,i}^2} \\
& + \frac{1}{n_{1,s}^2}\tilde \gamma_{1,s,n}^{-2} \sum_{i=N_s+1}^{N_s + n_{1,s}}\sum_{N_s+1 \leq k \leq N_s + n_{1,s},k\neq i}   \frac{|\theta_i \theta_k|  \tilde M_{1,s,i,k}^2 \tilde V_{1,s,i}^2\tilde V_{1,s,k}^2 }{\tilde M_{1,s,i,i}^2} \\
& \leq \frac{\max_{i \in [n]}\tilde V_{1,s,i}^4}{n_{1,s}^2 \tilde \gamma_{1,s,n}^2} \sum_{i=N_s+1}^{N_s + n_{1,s}}\frac{\theta_i^2 \sum_{k=N_s+1}^{N_s + n_{1,s}} \tilde M_{1,s,i,k}^2}{\tilde M_{1,s,i,i}^2} \\
& + \frac{\max_{i \in [n]}\tilde V_{1,s,i}^4}{n_{1,s}^2 \tilde \gamma_{1,s,n}^2} \sum_{i=N_s+1}^{N_s + n_{1,s}} \sum_{k=N_s+1}^{N_s + n_{1,s}}\frac{|\theta_i \theta_k| \tilde M_{1,s,i,k}^2}{\tilde M_{1,s,i,i}^2} \\
& \leq 2\left[\frac{\max_{ N_s+1 \leq i \leq N_s + n_{1,s} }\tilde V_{1,s,i}^4 \max_{ N_s+1 \leq i \leq N_s + n_{1,s} }\theta_i^2 }{n_{1,s} \tilde \gamma_{1,s,n}^2 \min_{ N_s+1 \leq i \leq N_s + n_{1,s} } \tilde M_{1,s,i,i} } \right]  \\
& = o_P(1),
\end{align*}
where the last equality is by $\max_{ N_s+1 \leq i \leq N_s + n_{1,s} }\theta_i^2 = o_P(n)$.

For $T_{1,3}$, we have $\mathbb{E} (T_{1,3}\mid \Xi_{1,s}) = 0$ and
\begin{align*}
Var(T_{1,3} \mid  \Xi_{1,s}) & = \frac{\tilde \gamma_{1,s,n}^{-2}}{n_{1,s}^2} \sum_{i = N_s+1}^{N_s+n_{1,s}} \theta_i^2 \mathbb{E}\left([(\eps_i^s(1))^2 - \tilde V_{1,s,i}^2]^2\mid \Xi_{1,s}\right) \\
& \leq \frac{\tilde \gamma_{1,s,n}^{-2} \max_{N_s + 1 \leq i \leq N_s+n_{1,s}}\mathbb{E}\left([(\eps_i^s(1))^2 - \tilde V_{1,s,i}^2]^2\mid \Xi_{1,s}\right)}{n_{1,s}^2} \sum_{i = N_s+1}^{N_s+n_{1,s}} \theta_i^2 = O_P(n^{-1}), 
\end{align*}
where we use the fact that $\sum_{i = N_s+1}^{N_s+n_{1,s}}\theta_i^2 = 1_{n_{1,s}}^\top \tilde M_{1,s} 1_{n_{1,s}} \leq n_{1,s}$. This implies $T_{1,3} = o_P(1)$. 

Last, by Assumption \ref{ass:omega}, we have
\begin{align*}
T_{1,4} \stackrel{d}{=}  \frac{1}{n_{1,s}} \gamma_{1,s,n}^{-1} \sum_{i \in \aleph_{1,s}}  (\sum_{j \in \aleph_{1,s}}M_{1,s,i,j})  \mathbb{E}(\eps_i^2(1)|X_i,S_i=s) \convP  \varpi_{a,s,\infty}.
\end{align*}


This concludes the proof of \textbf{Step 1.2.}


\textbf{Step 1.3: Consistency of $\hat \Sigma_{\mathcal U}$.} We note that $n_s^2/(nn_{1,s}) \convP p_s/\pi_s$ and $n_s^2/(nn_{0,s}) \convP p_s/(1-\pi_s)$. Therefore, we have

\begin{align*}
\hat \Sigma_{\mathcal U} \convP \sum_{s \in \mathcal{S}}p_s \begin{pmatrix}
\frac{\omega_{1,s,\infty}^2}{\pi_s} + \frac{\omega_{0,s,\infty}^2}{1-\pi_s}  &  \frac{\omega_{1,s,\infty}^2}{\pi_s} + \frac{\omega_{0,s,\infty}^2}{1-\pi_s} \\
\frac{\omega_{1,s,\infty}^2}{\pi_s} + \frac{\omega_{0,s,\infty}^2}{1-\pi_s} &  \bigcdot 
\end{pmatrix}.
\end{align*}


\textbf{Step 2: Limit of $\hat \Sigma_{\mathcal V}^{adj}$}. 
It suffices to consider the limits of $(\hat \beta_{a,s}^\top \Gamma_{a,s} \hat \beta_{a,s} - \sum_{i \in \aleph_{a,s}}P_{a,s,i,i} Y_i \acute \eps_{a,s,i})/n_{a,s}$ and $\hat \beta_{1,s}^\top \Gamma_s \hat \beta_{0,s}/n_s$, which are derived in the following two steps. 

\textbf{Step 2.1: Limit of $(\hat \beta_{a,s}^\top \Gamma_{a,s} \hat \beta_{a,s} - \sum_{i \in \aleph_{a,s}}P_{a,s,i,i} Y_i \acute \eps_{a,s,i})/n_{a,s}$.} We note that
\begin{align}
\hat \beta_{a,s} & = \Gamma_{a,s}^{-1} \sum_{i \in \aleph_{a,s}}\Breve{X}_i (Y_i-\hat \tau_a (s))    \notag \\
& = \Gamma_{a,s}^{-1}\sum_{i \in \aleph_{a,s}}\Breve{X}_i \left[(\alpha_{a,s} + \overline X_s^\top \beta_{a,s}-\hat \tau_a (s)) + \breve X_i^\top \beta_{a,s} + e_{i,s}(a) + \eps_{i}(a) \right] \notag \\
& = \beta_{a,s} + \Gamma_{a,s}^{-1}\sum_{i \in \aleph_{a,s}}\Breve{X}_i\eps_{i}(a)+ \Gamma_{a,s}^{-1}\sum_{i \in \aleph_{a,s}}\Breve{X}_i \left[(\alpha_{a,s} + \overline X_s^\top \beta_{a,s}-\hat \tau_a (s)) +  e_{i,s}(a) \right].
\label{eq:betahat}
\end{align}

By \eqref{eq:betahat}, we have 
\begin{align}
&  \hat \beta_{a,s}^\top \Gamma_{a,s} \hat \beta_{a,s} - \sum_{i \in \aleph_{a,s}}P_{a,s,i,i} Y_i \acute \eps_{a,s,i} \notag \\
& =  \beta_{a,s}^\top \Gamma_{a,s} \beta_{a,s} + \left( \eps_{\aleph_{a,s}}^\top(a) P_{a,s}\eps_{\aleph_{a,s}}(a) - \sum_{i \in \aleph_{a,s}}P_{a,s,i,i} Y_i \acute \eps_{a,s,i}\right) + 2 \sum_{i \in \aleph_{a,s}} \Breve{X}_i^\top \beta_{a,s} \eps_{i}(a) \notag \\
&+ 2 \sum_{i \in \aleph_{a,s}} \Breve{X}_i^\top \beta_{a,s} (\alpha_{a,s} + \overline X_s^\top \beta_{a,s}-\hat \tau_a (s)) + \eps_{\aleph_{a,s}}^\top(a) P_{a,s} \left[1_{n_{a,s}}(\alpha_{a,s} + \overline X_s^\top \beta_{a,s}-\hat \tau_a (s)) + e_{\aleph_{a,s}}(a)\right] \notag \\
& + \left[1_{n_{a,s}}(\alpha_{a,s} + \overline X_s^\top \beta_{a,s}-\hat \tau_a (s)) + e_{\aleph_{a,s}}(a)\right]^\top P_{a,s} \left[1_{n_{a,s}}(\alpha_{a,s} + \overline X_s^\top \beta_{a,s}-\hat \tau_a (s)) + e_{\aleph_{a,s}}(a)\right] \notag \\
& \equiv \beta_{a,s}^\top \Gamma_{a,s} \beta_{a,s} + \sum_{l\in [5]}I_l,
\label{eq:bhat21}
\end{align}
where $e_{\aleph_{a,s}}(a)$ and $\eps_{\aleph_{a,s}}(a)$ are $n_{a,s} \times 1$ vectors of $\{e_{i,s}(a)\}_{i \in \aleph_{a,s}}$ and $\{\eps_{i,s}(a)\}_{i \in \aleph_{a,s}}$, respectively. We note that 
\begin{align}\label{eq:xb}
\sum_{i \in \aleph_{a,s}} (\breve X_i^\top \beta_{a,s})^2 & =     \sum_{i \in \aleph_{a,s}} \left[\mathbb
E(Y_i(a)|X_i,S_i=s) - e_{i,s}(a) - \frac{1}{n_s}\sum_{j \in \aleph_s}\left(\mathbb
E(Y_j(a)|X_j,S_j=s) - e_{j,s}(a)\right) \right]^2 \notag \\
& \lesssim    \sum_{i \in \aleph_{a,s}} \left[\mathbb
E(Y_i(a)|X_i,S_i=s) - e_{i,s}(a)\right]^2 +  \frac{n_{a,s}}{n_s}\sum_{j \in \aleph_s}\left[\left(\mathbb
E(Y_j(a)|X_j,S_j=s) - e_{j,s}(a)\right) \right]^2 \notag \\
& = O_P(n),
\end{align}
\begin{align*}
& \left| \left[\frac{1}{n_{a,s}}\beta_{a,s}^\top \Gamma_{a,s} \beta_{a,s}\right]^{1/2} - \left[\frac{1}{n_{a,s}}\sum_{i \in \aleph_{a,s}}(\mathbb E(Y_i(a)|X_i,S_i=s) - \mathbb E(Y_i(a)|S_i=s))^2 \right]^{1/2}\right| \\
& = \left| \left[\frac{1}{n_{a,s}}\sum_{i \in \aleph_{a,s}}(\breve X_i^\top \beta_{a,s})^2\right]^{1/2} - \left[\frac{1}{n_{a,s}}\sum_{i \in \aleph_{a,s}}(\mathbb E(Y_i(a)|X_i,S_i=s) - \mathbb E(Y_i(a)|S_i=s))^2 \right]^{1/2}\right|  \\
& \leq \left[ \frac{1}{n_{a,s}}\sum_{i \in \aleph_{a,s}}\left(e_{i,s}(a) - \frac{1}{n_s}\sum_{j \in \aleph_s} e_{j,s}(a) \right)^2\right]^{1/2} + \left|\frac{1}{n_s}\sum_{j \in \aleph_s}\mathbb
E(Y_j(a)|X_j,S_j=s) - \mathbb E(Y_i(a)|S_i=s)\right| \\
& = o_P(1),
\end{align*}
and 
\begin{align*}
\frac{1}{n_{a,s}}\sum_{i \in \aleph_{a,s}}(\mathbb E(Y_i(a)|X_i,S_i=s) - \mathbb E(Y_i(a)|S_i=s))^2 \convP var(\phi_i(a)|S_i=s),    
\end{align*}
where $\phi_i(a) = \mathbb E(Y_i(a)|X_i,S_i=s) - \mathbb E(Y_i(a)|S_i=s)$. This implies 
\begin{align*}
\frac{1}{n_{a,s}}\beta_{a,s}^\top \Gamma_{a,s} \beta_{a,s} \convP     var(\phi_i(a)|S_i=s).
\end{align*}

Next, we show $I_1,\cdots,I_5$ are $o_P(n)$. Lemma \ref{lem:I1} shows $I_1 = o_P(n)$. For $I_2$, we note that, conditional on $(A^{(n)},S^{(n)},X^{(n)})$, $\{\eps_{i}(a)\}_{i \in \aleph_{a,s}}$ is independent across $i$ and mean zero, which implies 
\begin{align*}
& Var \left(\sum_{i \in \aleph_{a,s}} \Breve{X}_i^\top \beta_{a,s} \eps_{i}(a) \mid A^{(n)},S^{(n)},X^{(n)}\right) \\
& \leq \left[ \max_{i \in \aleph_{a,s}}\mathbb E( \eps_{i}^2(a) \mid A^{(n)},S^{(n)},X^{(n)}) \right] \left[ \sum_{i \in \aleph_{a,s}} (\Breve{X}_i^\top \beta_{a,s} )^2\right] = O_P(n),
\end{align*}
where the equality is by \eqref{eq:xb}. Therefore, we have $I_2 = O_P(n^{1/2}) = o_P(n)$. Because $\alpha_{a,s} + \overline X_s^\top \beta_{a,s}-\hat \tau_a (s) \convP o_P(1)$ by the proof of Theorem \ref{thm:main}, $\sum_{i \in \aleph_{a,s}}e_{i,s}^2(a) = o(n)$, and $\lambda_{\max}(P_{a,s}) \leq 1$, we have $I_l = o_P(n)$ for $l = 3,4,5$. This means 
\begin{align*}
\frac{1}{n_{a,s}} \left( \hat \beta_{a,s}^\top \Gamma_{a,s} \hat \beta_{a,s} - \sum_{i \in \aleph_{a,s}}P_{a,s,i,i} Y_i \acute \eps_{a,s,i} \right)\convP var(\phi_i(a)|S_i=s).
\end{align*}

\textbf{Step 2.2: Limit of $\hat \beta_{1,s}^\top \Gamma_s \hat \beta_{0,s}/n_s$.} By \eqref{eq:betahat}, we have
\begin{align*}
\hat \beta_{1,s}^\top \Gamma_s \hat \beta_{0,s} & = \beta_{1,s}^\top \Gamma_s \beta_{0,s}   + \beta_{1,s}^\top \Gamma_s \Gamma_{0,s}^{-1} \sum_{i \in \aleph_{0,s}} \breve X_i \eps_i(0) \\
& + \beta_{1,s}^\top \Gamma_s \Gamma_{0,s}^{-1} \left[ \sum_{i \in \aleph_{0,s}} \breve X_i \left( \alpha_{0,s} + \overline X_s^\top \beta_{0,s} - \hat \tau_0(s) + e_{i,s}(0) \right)\right] \\
& + \left[\sum_{i \in \aleph_{1,s}} \breve X_i \eps_i(1) \right]^\top \Gamma_{1,s}^{-1} \Gamma_s \beta_{0,s} + \left[\sum_{i \in \aleph_{1,s}} \breve X_i \eps_i(1) \right]^\top \Gamma_{1,s}^{-1} \Gamma_s \Gamma_{0,s}^{-1} \sum_{i \in \aleph_{0,s}} \breve X_i \eps_i(0) \\
& + \left[\sum_{i \in \aleph_{1,s}} \breve X_i \eps_i(1) \right]^\top \Gamma_{1,s}^{-1} \Gamma_s \Gamma_{0,s}^{-1} \left[ \sum_{i \in \aleph_{0,s}} \breve X_i \left( \alpha_{0,s} + \overline X_s^\top \beta_{0,s} - \hat \tau_0(s) + e_{i,s}(0) \right)\right]  \\
& + \left[\sum_{i \in \aleph_{1,s}} \breve X_i\left(\alpha_{1,s} + \overline X_s^\top \beta_{1,s} -\hat \tau_1(s) + e_{i,s}(1)\right)\right]^\top \Gamma_{1,s}^{-1} \Gamma_s \beta_{0,s} \\
& + \left[\sum_{i \in \aleph_{1,s}} \breve X_i\left(\alpha_{1,s} + \overline X_s^\top \beta_{1,s} -\hat \tau_1(s) + e_{i,s}(1)\right)\right]^\top \Gamma_{1,s}^{-1} \Gamma_s \Gamma_{0,s}^{-1} \sum_{i \in \aleph_{0,s}} \breve X_i \eps_i(0) \\
& + \left[\sum_{i \in \aleph_{1,s}} \breve X_i\left(\alpha_{1,s} + \overline X_s^\top \beta_{1,s} -\hat \tau_1(s) + e_{i,s}(1)\right)\right]^\top \Gamma_{1,s}^{-1} \Gamma_s \Gamma_{0,s}^{-1} \\
& \times \left[ \sum_{i \in \aleph_{0,s}} \breve X_i \left( \alpha_{0,s} + \overline X_s^\top \beta_{0,s} - \hat \tau_0(s) + e_{i,s}(0) \right)\right] \\
& \equiv \sum_{l \in [9]} I_l.
\end{align*}

Following the previous argument, we can show that 
\begin{align*}
\frac{I_1}{n} & = \frac{1}{n} \sum_{i \in \aleph_{s}}(\mathbb E(Y_i(1)|X_i,S_i=s) - \mathbb E(Y_i(1)|S_i=s))(\mathbb E(Y_i(0)|X_i,S_i=s) - \mathbb E(Y_i(0)|S_i=s)) +o_P(1) \\
& \convP cov(\phi_i(1),\phi_i(0)|S_i=s). 
\end{align*}

For $I_2$, we have
\begin{align*}
\mathbb E(I_2|A^{(n)}, S^{(n)},X^{(n)}) = 0
\end{align*}
and 
\begin{align*}
var(I_2|A^{(n)}, S^{(n)},X^{(n)}) \lesssim \beta_{1,s}^\top \Gamma_s \Gamma_{0,s}^{-1} \Gamma_s \beta_{1,s}  \max_{i \in \aleph_{0,s}} \mathbb E(\eps_i^2(0)| A^{(n)}, S^{(n)},X^{(n)}) = o_P(n^2),
\end{align*}
where we use the facts that \eqref{eq:xb} holds and 
\begin{align}
\beta_{1,s}^\top \Gamma_s \Gamma_{0,s}^{-1} \Gamma_s \beta_{1,s} \leq \left\Vert \Gamma_s^{1/2} \Gamma_{0,s}^{-1} \Gamma_s^{1/2}  \right\Vert_{op} ||\breve X_{\aleph_s} \beta_{1,s} ||_2^2 = o_P(n^2).
\label{eq:varI_2}
\end{align}
This implies $I_2/n = o_P(1)$. 

For $I_3$, by Cauchy's inequality, we have
\begin{align*}
|I_3| \leq \left\Vert \beta_{1,s}^\top \Gamma_s \Gamma_{0,s}^{-1} \breve X_{\aleph_{0,s}}^\top \right\Vert_{2} \times O_P(1) = o_P(n),
\end{align*}
where the first inequality holds by Assumption \ref{ass:linear}(iv) and the facts that 
\begin{align*}
\alpha_{0,s} + \overline X_s^\top \beta_{0,s} - \hat \tau_0(s) = O_P(n^{-1/2}), \quad \sum_{i \in \aleph_{0,s}} e_{i,s}^2(a) = o_P(1)
\end{align*}
and the equality holds by Assumption \ref{ass:variance}, \eqref{eq:varI_2}, and the fact that 
\begin{align*}
\left\Vert \beta_{1,s}^\top \Gamma_s \Gamma_{0,s}^{-1} \breve X_{\aleph_{0,s}}^\top \right\Vert_{2} = \left( \beta_{1,s}^\top \Gamma_s \Gamma_{0,s}^{-1} \Gamma_s \beta_{1,s}\right)^{1/2}  = o_P(n).
\end{align*}

We can show $I_4= o_P(n)$ following the same argument used to bound $I_2$.

For $I_5$, we note that the two index sets $\aleph_{0,s}$ and $\aleph_{1,s}$ do not overlap. Therefore, we have
\begin{align*}
\mathbb E(I_5|A^{(n)}, S^{(n)},X^{(n)}) = 0
\end{align*}
and 
\begin{align*}
var(I_5|A^{(n)}, S^{(n)},X^{(n)}) \leq \sum_{i \in \aleph_{1,s}}\sum_{j \in \aleph_{0,s}} B_{i,j}^2  \left[\max_{a = 0,1}\max_{i \in \aleph_{a,s}} \mathbb E(\eps_i^2(0)| A^{(n)}, S^{(n)},X^{(n)})\right]^2 = o_P(n^2),
\end{align*}
where we have $B_{i,j} = \breve X_i^\top \Gamma_{1,s}^{-1} \Gamma_s \Gamma_{0,s}^{-1} \breve X_j$ and 
\begin{align*}
\sum_{i \in \aleph_{1,s}}\sum_{j \in \aleph_{0,s}} B_{i,j}^2 & = \sum_{i \in \aleph_{1,s}} \text{trace}(\breve X_i^\top \Gamma_{1,s}^{-1} \Gamma_s \Gamma_{0,s}^{-1} \Gamma_s \Gamma_{1,s}^{-1} \breve X_i) \\
& = \text{trace}(\Gamma_s^{1/2}\Gamma_{0,s}^{-1}\Gamma_s \Gamma_{1,s}^{-1}\Gamma_s^{1/2} ) \\
& = \text{trace}(\Gamma_s^{1/2}\Gamma_{0,s}^{-1}\Gamma_s^{1/2} ) + \text{trace}(\Gamma_s^{1/2} \Gamma_{1,s}^{-1}\Gamma_s^{1/2} )= o_P(n^2), 
\end{align*}
where we use the fact that $\Gamma_s = \Gamma_{0,s}+\Gamma_{1,s}$. 

For $I_6$, we have
\begin{align*}
|I_6| & \leq \left\vert \left[\sum_{i \in \aleph_{1,s}} \breve X_i \eps_i(1) \right]^\top \Gamma_{1,s}^{-1} \Gamma_s \Gamma_{0,s}^{-1} \left[ \sum_{i \in \aleph_{0,s}} \breve X_i \right]  \right\vert |\alpha_{0,s} + \overline X_s^\top \beta_{0,s} - \hat \tau_0(s) | \\
& + \left\vert \left[\sum_{i \in \aleph_{1,s}} \breve X_i \eps_i(1) \right]^\top \Gamma_{1,s}^{-1} \Gamma_s \Gamma_{0,s}^{-1} \left[ \sum_{i \in \aleph_{0,s}} \breve X_i e_{i,s}(0)\right]  \right\vert\\
& \equiv |I_{6,1}| |\alpha_{0,s} + \overline X_s^\top \beta_{0,s} - \hat \tau_0(s) | + |I_{6,2}|.
\end{align*}
We note that 
\begin{align*}
\mathbb E(I_{6,1}|A^{(n)}, S^{(n)},X^{(n)}) = 0
\end{align*}
and 
\begin{align*}
&    var(I_{6,1}|A^{(n)}, S^{(n)},X^{(n)}) \\
& \leq 
\left[ \sum_{i \in \aleph_{0,s}} \breve X_i \right]^\top \Gamma_{0,s}^{-1}\Gamma_s \Gamma_{1,s}^{-1} \Gamma_s \Gamma_{0,s}^{-1}\left[ \sum_{i \in \aleph_{0,s}} \breve X_i \right] 
\left[\max_{i \in \aleph_{a,s}} \mathbb E(\eps_i^2(1)| A^{(n)}, S^{(n)},X^{(n)})\right] \\
& = o_P(n^3). 
\end{align*}
where the second last equality holds because 
\begin{align*}
\left[ \sum_{i \in \aleph_{0,s}} \breve X_i \right]^\top \Gamma_{0,s}^{-1}\Gamma_s \Gamma_{1,s}^{-1} \Gamma_s S_0^{-1}\left[ \sum_{i \in \aleph_{0,s}} \breve X_i \right] & \leq n_{0,s} \left\Vert \breve X_{\aleph_{0,s}} \Gamma_{0,s}^{-1}\Gamma_s \Gamma_{1,s}^{-1} \Gamma_s S_0^{-1} \breve X_{\aleph_{0,s}} \right\Vert_{op} \\
& \leq n_{0,s} \left\Vert \Gamma_s^{1/2} \Gamma_{1,s}^{-1}\Gamma_s^{1/2} \right\Vert_{op} \left\Vert \breve X_{\aleph_{0,s}} \Gamma_{0,s}^{-1}\Gamma_s^{1/2} \right\Vert_{op}^2 \\
& \leq n_{0,s} \left\Vert \Gamma_s^{1/2} \Gamma_{1,s}^{-1}\Gamma_s^{1/2} \right\Vert_{op} \left\Vert \Gamma_s^{1/2} \Gamma_{0,s}^{-1}\Gamma_s^{1/2} \right\Vert_{op}  = o_P(n^3).
\end{align*}
This implies $|I_{6,1}| |\alpha_{0,s} + \overline X_s^\top \beta_{0,s} - \hat \tau_0(s) | = o_P(n)$ because $|\alpha_{0,s} + \overline X_s^\top \beta_{0,s} - \hat \tau_0(s) | = O_P(n^{-1/2})$. Similarly, for $I_{6,2}$, we have 
\begin{align*}
\mathbb E(I_{6,2}|A^{(n)}, S^{(n)},X^{(n)}) = 0
\end{align*}
and 
\begin{align*}
&    var(I_{6,2}|A^{(n)}, S^{(n)},X^{(n)}) \\
& \leq 
\left[ \sum_{i \in \aleph_{0,s}} \breve X_i e_{i,s}(0)\right]^\top \Gamma_{0,s}^{-1}\Gamma_s \Gamma_{1,s}^{-1} \Gamma_s S_0^{-1}\left[ \sum_{i \in \aleph_{0,s}} \breve X_i e_{i,s}(0)\right] 
\left[\max_{i \in \aleph_{a,s}} \mathbb E(\eps_i^2(1)| A^{(n)}, S^{(n)},X^{(n)})\right] \\
& \leq \left\Vert \breve X_{\aleph_{0,s}} \Gamma_{0,s}^{-1}\Gamma_s \Gamma_{1,s}^{-1} \Gamma_s S_0^{-1} \breve X_{\aleph_{0,s}} \right\Vert_{op} \left\Vert e_{\aleph_{0,s}}(0)\right\Vert_2^2     \left[\max_{i \in \aleph_{a,s}} \mathbb E(\eps_i^2(1)| A^{(n)}, S^{(n)},X^{(n)})\right] = o_P(n^2),
\end{align*}
where we use the fact that $\left\Vert e_{\aleph_{0,s}}(0)\right\Vert_2^2 = O_P(1)$. This implies $I_6 = o_P(n)$. 

We can show $I_7 = o_P(n)$ and $I_8 = o_P(n)$ following the same argument used to bound $I_3$ and $I_6$, respectively. 


For $I_9$, we have
\begin{align*}
|I_9| & \leq \left\Vert \breve X_{\aleph_{1,s}} \Gamma_{1,s}^{-1}\Gamma_s \Gamma_{0,s}^{-1} \breve X_{\aleph_{0,s}} \right\Vert_{op} \left\Vert 1_{n_{1,s}}(\alpha_{1,s} + \overline X_s^\top - \hat \tau_1(s)) + e_{\aleph_{1,s}}(1) \right\Vert_2  \\
& \times \left\Vert 1_{n_{0,s}}(\alpha_{0,s} + \overline X_s^\top - \hat \tau_0(s)) + e_{\aleph_{0,s}}(0) \right\Vert_2 = o_P(n),
\end{align*}
where we use the fact that 
\begin{align*}
\left\Vert \breve X_{\aleph_{1,s}} \Gamma_{1,s}^{-1}\Gamma_s \Gamma_{0,s}^{-1} \breve X_{\aleph_{0,s}} \right\Vert_{op} & =  \left\Vert \breve X_{\aleph_{1,s}} \Gamma_{1,s}^{-1}\Gamma_s \Gamma_{0,s}^{-1} \Gamma_s \Gamma_{1,s}^{-1} \breve X_{\aleph_{1,s}}^\top \right\Vert_{op}^{1/2} \\
& \leq \left[\left\Vert \Gamma_s^{1/2} \Gamma_{0,s}^{-1} \Gamma_s^{1/2} \right\Vert_{op} \left\Vert \Gamma_s^{1/2} \Gamma_{1,s}^{-1} \breve X_{\aleph_{1,s}}  \right\Vert_{op}^2 \right]^{1/2} \\ 
& = \left[\left\Vert \Gamma_s^{1/2} \Gamma_{0,s}^{-1} \Gamma_s^{1/2} \right\Vert_{op} \left\Vert \Gamma_s^{1/2} \Gamma_{1,s}^{-1} \Gamma_s^{1/2} \right\Vert_{op} \right]^{1/2} = o_P(n)
\end{align*}
and 
\begin{align*}
\left\Vert 1_{n_{0,s}}(\alpha_{0,s} + \overline X_s^\top - \hat \tau_0(s)) + e_{\aleph_{0,s}}(0) \right\Vert_2 = O_P(1).
\end{align*}
This leads to the desired result in this step that $\hat \beta_{1,s}^\top \Gamma_s \hat \beta_{0,s}/n_s \convP cov(\phi_i(1),\phi_i(0)|S_i=s)$.

\textbf{Step 3: Limit of $\hat \Sigma_{\mathcal W}$.}
For $\hat \Sigma_{\mathcal W}$, by the proof of Theorem \ref{thm:main}, we have $\hat \tau_{1,s}- \hat \tau_{0,s} \convP \mathbb{E}(Y(1)-Y(0)|S=s)$ and $\hat \tau^{adj} \convP \tau$. By Assumption \ref{ass:assignment1}, we also have $\hat p_s \convP p_s$, which implies the desired result.

\section{Proof of Theorem \ref{thm:fixed_k}}\label{sec:proof_thm_fixed_K}
This proof depends on the results of Lemma \ref{lem:beta}. Let $S^{k_n}$ be the unit sphere in $\Re^{k_n}$, $\mathbb P_{n,a,s}f_i = \frac{1}{n_{a,s}}\sum_{i \in \aleph_{a,s}}f_i$ for any function of observations $f_i$, $\mathbb P_{a,s}f_i = \mathbb E(f_i|A_i=a,S_i=s) = \mathbb E(f_i|S_i=s)$, $||\cdot||_{\mathbb P_{n,a,s},2}$ be the $L_2$ norm w.r.t. the probability measure $\mathbb P_{n,a,s}$.

\noindent\textbf{Step 1: Asymptotic normality and the expression of $\Omega$.} Following Remark \ref{rem:aipw}, we have
\begin{align*}
\hat \tau^{adj} & = \frac{1}{n} \sum_{i \in [n]} \frac{A_i(Y_i - X_i^\top \hat \beta_{1,S_i})}{\hat \pi_{S_i}} - \frac{1}{n} \sum_{i \in [n]} \frac{(1-A_i)(Y_i -  X_i^\top \hat \beta_{0,S_i})}{1-\hat \pi_{S_i}} + \frac{1}{n}\sum_{i \in [n]}X_i^\top (\hat \beta_{1,S_i} - \hat \beta_{0,S_i}) \\
& = \frac{1}{n} \sum_{i \in [n]} \frac{A_i(Y_i(1) - X_i^\top \beta_{1,S_i}^*)}{\hat \pi_{S_i}} - \frac{1}{n} \sum_{i \in [n]} \frac{(1-A_i)(Y_i(0) -  X_i^\top \beta_{0,S_i}^*)}{1-\hat \pi_{S_i}} + \frac{1}{n}\sum_{i \in [n]}X_i^\top (\beta_{1,S_i}^* -  \beta_{0,S_i}^*) \\
& + \sum_{s \in \mathcal{S}} \hat p_s \left(\frac{1}{n_{1,s}} \sum_{i \in \aleph_{1,s}}X_i - \frac{1}{n_{s}} \sum_{i \in \aleph_{s}}X_i\right)^\top (\beta_{1,s}^* - \hat \beta_{1,s}) \\
& -\sum_{s \in \mathcal{S}} \hat p_s \left(\frac{1}{n_{0,s}} \sum_{i \in \aleph_{0,s}}X_i - \frac{1}{n_{s}} \sum_{i \in \aleph_{s}}X_i\right)^\top (\beta_{0,s}^* - \hat \beta_{0,s}) \\
& = \frac{1}{n} \sum_{i \in [n]} \frac{A_i(Y_i(1) - X_i^\top \beta_{1,S_i}^*)}{\hat \pi_{S_i}} - \frac{1}{n} \sum_{i \in [n]} \frac{(1-A_i)(Y_i(0) -  X_i^\top \beta_{0,S_i}^*)}{1-\hat \pi_{S_i}} \\
& + \frac{1}{n}\sum_{i \in [n]}X_i^\top (\beta_{1,S_i}^* -  \beta_{0,S_i}^*) + o_P(n^{-1/2}),
\end{align*}
where the last equality holds because  
\begin{align*}
&  \max_{a=0,1,s\in \mathcal{S}} \left\Vert\frac{1}{n_{a,s}} \sum_{i \in \aleph_{a,s}}X_i - \frac{1}{n_{s}} \sum_{i \in \aleph_{s}}X_i \right\Vert_2 = O_P((k_n/n)^{1/2}),\\
& \max_{a=0,1,s\in \mathcal{S}} ||\hat \beta_{a,s} - \beta_{a,s}^*||_2 = O_P\left(\sqrt{\frac{k_n \xi_n^2 \log(k_n)}{n} } \right), 
\end{align*}
by Lemma \ref{lem:beta}, 
and $k_n^2 \xi_n^2 \log(k_n) = o(n)$. Let $\delta_{Y,a,s,i} = Y_i(a) - \mathbb{E}(Y_i(a)|S_i=s)$ and $\delta_{X,s,i} = X_i - \mathbb{E}(X_i|S_i=s)$. Then, we have
\begin{align*}
\begin{pmatrix}
\sqrt{n}(\hat \tau^{adj} - \tau) \\
\sqrt{n}(\hat \tau^{unadj} - \tau) 
\end{pmatrix}    & = \sum_{s \in \mathcal{S}} \hat p_s \sqrt{n} \begin{pmatrix}
\frac{1}{n_{1,s}}\sum_{i \in \aleph_{1,s}}(\delta_{Y,1,s,i} - (1-\hat \pi_s)\delta_{X,s,i}^\top \beta_{1,s}^* - \hat \pi_s \delta_{X,s,i}^\top \beta_{0,s}^*)  \\
\frac{1}{n_{1,s}}\sum_{i \in \aleph_{1,s}}\delta_{Y,1,s,i}
\end{pmatrix} \\
& -      \sum_{s \in \mathcal{S}} \hat p_s \sqrt{n} \begin{pmatrix}
\frac{1}{n_{0,s}}\sum_{i \in \aleph_{0,s}}(\delta_{Y,0,s,i} - \hat \pi_s\delta_{X,s,i}^\top \beta_{0,s}^* - (1-\hat \pi_s) \delta_{X,s,i}^\top \beta_{1,s}^*)  \\
\frac{1}{n_{0,s}}\sum_{i \in \aleph_{0,s}}\delta_{Y,0,s,i} 
\end{pmatrix} \\
& + \sqrt{n} \sum_{s \in \mathcal{S}} \hat p_s (\mathbb{E}(Y_i(1)|S_i=s) - \mathbb{E}(Y_i(0)|S_i=s) - \tau)1_2 +  o_P(1) \\
& = \sum_{s \in \mathcal{S}} \hat p_s \sqrt{n} \begin{pmatrix}
\frac{1}{n_{1,s}}\sum_{i \in \aleph_{1,s}}(\delta_{Y,1,s,i} - (1- \pi_s)\delta_{X,s,i}^\top \beta_{1,s}^* -  \pi_s \delta_{X,s,i}^\top \beta_{0,s}^*)  \\
\frac{1}{n_{1,s}}\sum_{i \in \aleph_{1,s}}\delta_{Y,1,s,i}
\end{pmatrix} \\
& -      \sum_{s \in \mathcal{S}} \hat p_s \sqrt{n} \begin{pmatrix}
\frac{1}{n_{0,s}}\sum_{i \in \aleph_{0,s}}(\delta_{Y,0,s,i} - \pi_s\delta_{X,s,i}^\top \beta_{0,s}^* - (1- \pi_s) \delta_{X,s,i}^\top \beta_{1,s}^*)  \\
\frac{1}{n_{0,s}}\sum_{i \in \aleph_{0,s}}\delta_{Y,0,s,i} 
\end{pmatrix} \\
& + \frac{1}{\sqrt{n}}\sum_{i \in [n]}(\mathbb{E}(Y_i(1)|S_i) - \mathbb{E}(Y_i(0)|S_i) - \tau)1_2 +  o_P(1)\\
& \equiv \tilde U_{n,1} - \tilde U_{n,0} + W_n 1_2 + o_P(1),
\end{align*}
where we denote 
\begin{align*}
\tilde U_{n,1} & =  \sum_{s \in \mathcal{S}} \hat p_s \sqrt{n} \begin{pmatrix}
\frac{1}{n_{1,s}}\sum_{i \in \aleph_{1,s}}(\delta_{Y,1,s,i} - (1-\pi_s)\delta_{X,s,i}^\top \beta_{1,s}^* -  \pi_s \delta_{X,s,i}^\top \beta_{0,s}^*)  \\
\frac{1}{n_{1,s}}\sum_{i \in \aleph_{1,s}}\delta_{Y,1,s,i}
\end{pmatrix}, \\
\tilde U_{n,0} & =   \sum_{s \in \mathcal{S}} \hat p_s \sqrt{n} \begin{pmatrix}
\frac{1}{n_{0,s}}\sum_{i \in \aleph_{0,s}}(\delta_{Y,0,s,i} -  \pi_s\delta_{X,s,i}^\top \beta_{0,s}^* - (1- \pi_s) \delta_{X,s,i}^\top \beta_{1,s}^*)  \\
\frac{1}{n_{0,s}}\sum_{i \in \aleph_{0,s}}\delta_{Y,0,s,i} 
\end{pmatrix}, 
\end{align*}
$W_n$ is defined in \eqref{eq:Wn}, and the second equality is by the facts that
\begin{align*}
&    \hat \pi_s - \pi_s = \frac{\sum_{i \in [n]}(A_i - \pi_s)1\{S_i=s\}}{n_s} = o_P(1), \quad  \left\vert \frac{1}{n_{a,s}} \sum_{i \in \aleph_{a,s}}\delta_{X,s,i}(\beta_{1,s}^* - \beta_{0,s}^*)\right\vert = O_P(n^{-1/2}), \quad \text{and thus,}\\
& \max_{a=0,1,s\in \mathcal{S}}\left\vert (\hat \pi_s - \pi_s) \left(\frac{1}{n_{a,s}} \sum_{i \in \aleph_{a,s}}\delta_{X,s,i}^\top (\beta_{1,s}^* - \beta_{0,s}^*) \right)\right\vert = o_P(n^{-1/2}). 
\end{align*}

Following the same argument by \cite{BCS17},\footnote{Also see Lemma N.3 in \cite{JPTZ22} for a similar argument for the regression adjusted quantile treatment effect estimator.} we can show that 
\begin{align*}
\begin{pmatrix}
\Omega_{U,1}^{-1/2} \tilde U_1 \\
\Omega_{U,0}^{-1/2} \tilde U_0 \\
\sigma_{W}^{-1} W
\end{pmatrix}  \convD \N(0_5, I_5),
\end{align*}
where $0_5$ is a $5$-dimensional vector of zero, $I_5$ is a $5 \times 5$ identity matrix, $m_{a}(x,s) = \mathbb{E}(Y_i(a)|X_i=x,S_i=s)$, $\overline \beta_s^* = (1-\pi_s)\beta_{1,s} + \pi_s \beta_{0,s}^*$, 
\begin{align*}
& \Omega_{U,1} =\mathbb{E}\frac{Var(Y_i(1)|X_i,S_i)}{\pi_{S_i}}1_2 1_2^\top \\
&+ \sum_{s \in \mathcal{S}}\frac{p_s}{\pi_s}  \begin{pmatrix}
Var(m_{1}(X_i,s)- X_i^\top \overline \beta_{s}^*\mid S_i =s) & cov(m_{1}(X_i,s)- X_i^\top \overline \beta_{s}^*, m_{1}(X_i,s)\mid S_i =s)\\
cov(m_{1}(X_i,s)- X_i^\top \overline \beta_{s}^*, m_{1}(X_i,s)\mid S_i =s) & Var(m_{1}(X_i,s)),
\end{pmatrix} \\
& \Omega_{U,0} =  \mathbb{E}\frac{Var(Y_i(0)|X_i,S_i)}{1-\Pi_{S_i}}1_2 1_2^\top \\
&+ \sum_{s \in \mathcal{S}}\frac{p_s}{1-\pi_s}  \begin{pmatrix}
Var(m_{0}(X_i,s)- X_i^\top \overline \beta_s^* \mid S_i =s) & cov(m_{0}(X_i,s)- X_i^\top \overline \beta_{s}^*, m_{0}(X_i,s)\mid S_i =s)\\
cov(m_{0}(X_i,s)- X_i^\top \overline \beta_{s}^*, m_{0}(X_i,s)\mid S_i =s) & Var(m_{0}(X_i,s))
\end{pmatrix}, \\
& \sigma_W^2 = var(\mathbb{E}(Y_i(1)|S_i) - \mathbb{E}(Y_i(0)|S_i)).
\end{align*}

We further note that 
\begin{align*}
\Omega & =    \Omega_{U,1} +   \Omega_{U,0} +   Var(\mathbb{E}(Y_i(1)|S_i) - \mathbb{E}(Y_i(0)|S_i))1_2 1_2^\top  \\
& = \left\{ \mathbb{E}\left[\frac{Var(Y_i(1)|X_i,S_i)}{\pi_{S_i}} + \frac{Var(Y_i(0)|X_i,S_i)}{1-\pi_{S_i}}\right] + \mathbb{E}(m_{1}(X_i,S_i) - m_{0}(X_i,S_i)-\tau)^2 \right\}1_2 1_2^\top \\
& + \sum_{s \in \mathcal{S}}\frac{p_s}{\pi_s (1-\pi_s)}\begin{pmatrix}
V_s & V_s \\
V_s & V_s'
\end{pmatrix}
\end{align*}
where 
\begin{align*}
& V_s = Var((1-\pi_s) m_1(X_i,s) + \pi_s m_0(X_i,s) - X_i^\top \overline{\beta}_s^* \mid S_i=s) \quad \text{and}\\
& V_s' = Var((1-\pi_s) m_1(X_i,s) + \pi_s m_0(X_i,s) \mid S_i=s).
\end{align*}
Then, by the definition of $\beta_{a,s}^*$, we see that $V_s' \geq V_s$ for $s \in \mathcal{S}$ and thus, $\Omega_{1,1} = \Omega_{1,2} \leq \Omega_{2,2}$. 

Last, we have
\begin{align*}
\Omega^{-1/2}  \begin{pmatrix}
\sqrt{n}(\hat \tau^{adj} - \tau) \\
\sqrt{n}(\hat \tau^{unadj} - \tau) 
\end{pmatrix}  =   \Omega^{-1/2} \begin{pmatrix}
\Omega_{U,1}^{1/2} & \Omega_{U,0}^{1/2} &    \sigma_{W} 1_2
\end{pmatrix}  \begin{pmatrix}
\Omega_{U,1}^{-1/2} \tilde U_1 \\
\Omega_{U,0}^{-1/2} \tilde U_0 \\
\sigma_{W}^{-1} W
\end{pmatrix} \convD \N(0_2,I_2)
\end{align*}

\noindent\textbf{Step 2: Limit of $\hat \Sigma_{\mathcal U}$. } 
We have 
\begin{align*}
\gamma_{a,s,n} = 1 - \left(\frac{1}{n_{a,s}}\sum_{i \in \aleph_{a,s}}\breve X_i^\top \right) \left(\frac{1}{n_{a,s}}\sum_{i \in \aleph_{a,s}}\breve X_i\breve X_i^\top \right)^{-1}\left(\frac{1}{n_{a,s}}\sum_{i \in \aleph_{a,s}}\breve X_i \right) \convP 1, 
\end{align*}
where we use the fact that 
\begin{align*}
0 & \leq \left(\frac{1}{n_{a,s}}\sum_{i \in \aleph_{a,s}}\breve X_i^\top \right) \left(\frac{1}{n_{a,s}}\sum_{i \in \aleph_{a,s}}\breve X_i\breve X_i^\top \right)^{-1}\left(\frac{1}{n_{a,s}}\sum_{i \in \aleph_{a,s}}\breve X_i \right) \\
& \leq \left[\lambda_{\min}\left(\frac{1}{n_{a,s}}\sum_{i \in \aleph_{a,s}}\breve X_i\breve X_i^\top \right)\right]^{-1} \left\Vert \frac{1}{n_{a,s}}\sum_{i \in \aleph_{a,s}}\breve X_i \right\Vert_2^2 \\
& = O_P(k_n/n) = o_P(1). 
\end{align*}
In addition, we have $\sum_{j \in \aleph_{a,s}}M_{a,s,i,j} = 1 - R_{a,s,i}$ where 
$$R_{a,s,i} = \breve X_i^\top \left(\frac{1}{n_{a,s}}\sum_{i \in \aleph_{a,s}}\breve X_i\breve X_i^\top \right)^{-1}\left(\frac{1}{n_{a,s}}\sum_{i \in \aleph_{a,s}}\breve X_i \right) $$ 
so that 
\begin{align*}
\max_{i \in \aleph_{a,s}}|R_{a,s,i}| \leq \max_{i \in [n]} ||\breve X_i||_2 \lambda_{\min}\left(\frac{1}{n_{a,s}}\sum_{i \in \aleph_{a,s}}\breve X_i\breve X_i^\top \right)^{-1} \left\Vert \frac{1}{n_{a,s}}\sum_{i \in \aleph_{a,s}}\breve X_i \right\Vert_2 = O_P(\sqrt{k_n^2 \xi_n/n}) = o_P(1). 
\end{align*}
and
\begin{align*}
\left| \frac{1}{n_{a,s}}\sum_{i \in \aleph_{a,s}}\left[\left(\sum_{j \in \aleph_{a,s}}M_{a,s,i,j}\right)^2-1\right] Y_i \acute{\eps}_{a,s,i}\right| \leq 2\max_{i \in \aleph_{a,s}}\left(\left|R_{a,s,i}\right| + R_{a,s,i}^2\right) \frac{1}{n_{a,s}} \sum_{i \in \aleph_{a,s}}\left|Y_i \acute{\eps}_{a,s,i}\right| = o_P(1),
\end{align*}
where we use the facts that 
\begin{align*}
\min_{i \in \aleph_{a,s}} M_{a,s,i,i} \geq 1- \max_{i \in \aleph_{a,s}}P_{a,s,i,i} & \geq 1- \frac{1}{n_{a,s}}\lambda_{\max}\left( \frac{1}{n_{a,s}} \sum_{i \in \aleph_{a,s}} \breve X_i \breve X_i^\top \right) \max_{i \in \aleph_{a,s}}||\breve X_i||_2^2 = 1-o_P(1)
\end{align*}
and
\begin{align*}
\frac{1}{n_{a,s}}\sum_{i \in \aleph_{a,s}}\left|Y_i \acute{\eps}_{a,s,i}\right| & \leq \frac{1}{n_{a,s} \min_{i \in \aleph_{a,s}} M_{a,s,i,i}} \sum_{i \in \aleph_{a,s}} |Y_i||Y_i - \hat \tau_{a,s} - \breve X_i^\top \hat \beta_{a,s}| \\
& \leq \frac{1}{n_{a,s} \min_{i \in \aleph_{a,s}} M_{a,s,i,i}} \sum_{i \in \aleph_{a,s}} \left[Y_i^2 +  |Y_i| |\hat \tau_{a,s}| + |Y_i|||\breve X_i||_2 ||\hat \beta_{a,s} - \beta_{a,s}^*||_2 +  |Y_i||\breve X_i^\top \beta_{a,s}^*| \right] \\
& = O_P(1).
\end{align*}
To see the last equality of the last display, we note that $ \frac{1}{n_{a,s} } \sum_{i \in \aleph_{a,s}} Y_i^2 = O_P(1)$, $\hat \tau_{a,s} = O_P(1)$ as shown in \eqref{eq:tau_as}, 
\begin{align*}
\frac{1}{n_{a,s} } \sum_{i \in \aleph_{a,s}} |Y_i|||\breve X_i||_2 ||\hat \beta_{a,s} - \beta_{a,s}^*||_2 & \leq  ||Y_i||_{\mathbb P_{n,a,s},2} ||\breve X_i||_{\mathbb P_{n,a,s},2} ||\hat \beta_{a,s} - \beta_{a,s}^*||_2 \\
& = O_P\left(1 \times \sqrt{k_n} \times  \sqrt{\frac{k_n \xi_n^2 \log (k_n)}{n}} \right) = o_P(1),
\end{align*}
\begin{align*}
\frac{1}{n_{a,s} } \sum_{i \in \aleph_{a,s}}|Y_i||\breve X_i^\top \beta_{a,s}^*| & \leq  ||Y_i||_{\mathbb P_{n,a,s},2} ||\breve X_i^\top \beta_{a,s}^* ||_{\mathbb P_{n,a,s},2} \\
& = O_P\left(1 \times ||\beta_{a,s}^*||_2 \right) = O_P(1),
\end{align*}
where we use the fact that 
\begin{align*}
||\beta_{a,s}^*||_2 \leq \lambda_{\min}^{-1}(Var(X_i|S_i=s)) ||Cov(X_i,Y_i|S_i=s)||_2 \leq C ||Cov(X_i,Y_i|S_i=s)||_2
\end{align*}
and by Assumption \ref{ass:reg_fixed} and \citet[Theorem 1.1]{tripathi1999}
\begin{align*}
\infty > C & \geq   Var(Y_i|S_i=s) \\
& \geq Cov(X_i^\top,Y_i|S_i=s)(Var(X_i|S_i=s))^{-1}Cov(X_i,Y_i|S_i=s) \\
& \geq \lambda_{\max}^{-1}(Var(X_i|S_i=s)) ||Cov(X_i,Y_i|S_i=s)||_2^2 \geq c||Cov(X_i,Y_i|S_i=s)||_2^2,
\end{align*}
implying $||Cov(X_i,Y_i|S_i=s)||_2^2$, and thus, $ ||\beta_{a,s}^*||_2$ are O(1).

Next, recall $\hat \tau_{a,s}$ and $\hat \beta_{a,s}$ are the intercept and slope of the OLS regression of $Y_i$ on $(1,\breve X_i)$ using observations $i \in \aleph_{a,s}$. Therefore, by Lemma \ref{lem:beta}, we have 
\begin{align*}
\frac{1}{n_{a,s}}\sum_{i \in \aleph_{a,s}}\breve X_i^\top (\hat \beta_{a,s} - \beta_{a,s}^*) = O_P(\sqrt{k_n/n})||\hat \beta_{a,s} - \beta_{a,s}^*||_2 = o_P(1), 
\end{align*}
$$\frac{1}{n_{a,s}}\sum_{i \in \aleph_{a,s}}\breve X_i^\top \beta_{a,s}^* = o_P(1),$$
\begin{align}\label{eq:tau_as}
\hat \tau_{a,s} & = \frac{1}{n_{a,s}}\sum_{i \in \aleph_{a,s}}(Y_i - \breve X_i^\top \hat \beta_{a,s}) \notag \\
& = \frac{1}{n_{a,s}}\sum_{i \in \aleph_{a,s}}Y_i - \frac{1}{n_{a,s}}\sum_{i \in \aleph_{a,s}}\breve X_i^\top \beta_{a,s}^* + \frac{1}{n_{a,s}}\sum_{i \in \aleph_{a,s}}\breve X_i^\top (\hat \beta_{a,s} - \beta_{a,s}^*) \notag \\
& \convP \mathbb{E}(Y_i(1)|S_i=s), 
\end{align}
\begin{align*}
\left| \frac{1}{n_{a,s}}\sum_{i \in \aleph_{a,s}}Y_i (\acute{\eps}_{a,s,i} - \hat \eps_{a,s,i}) \right|\leq \left(\max_{i \in \aleph_{a,s}}P_{a,s,i,i}\right) \left(    \frac{1}{n_{a,s}}\sum_{i \in \aleph_{a,s}}\left|Y_i \acute{\eps}_{a,s,i}\right| \right)= o_P(1),
\end{align*}
and 
\begin{align*}
\frac{1}{n_{a,s}}\sum_{i \in \aleph_{a,s}}Y_i \hat{\eps}_{a,s,i} & = \frac{1}{n_{a,s}}\sum_{i \in \aleph_{a,s}}Y_i (Y_i - \hat \tau_{a,s} - \breve X_i^\top \hat \beta_{a,s}) \\
& = \frac{1}{n_{a,s}}\sum_{i \in \aleph_{a,s}}Y_i^2  - \left(\frac{1}{n_{a,s}}\sum_{i \in \aleph_{a,s}}Y_i\right) \hat \tau_{a,s} \\
&- \left[\frac{1}{n_{a,s}}\sum_{i \in \aleph_{a,s}}Y_i X_i - \left(\frac{1}{n_{a,s}}\sum_{i \in \aleph_{a,s}}Y_i\right) \left(\frac{1}{n_{s}}\sum_{i \in \aleph_{s}}X_i\right)\right]^\top \beta_{a,s}^* \\
& - \left[\frac{1}{n_{a,s}}\sum_{i \in \aleph_{a,s}}Y_i X_i - \left(\frac{1}{n_{a,s}}\sum_{i \in \aleph_{a,s}}Y_i\right) \left(\frac{1}{n_{s}}\sum_{i \in \aleph_{s}}X_i\right)\right]^\top (\hat \beta_{a,s}-\beta_{a,s}^*) \\
& =  Var(Y_i(a)|S_i=s) + o_P(1) - cov(Y_i(a),X_i^\top \beta_{a,s}^*|S_i=s)  - O_P(k_n^{1/2})||\hat \beta_{a,s}-\beta_{a,s}^*||_2 \\
& = Var(Y_i(a)|S_i=s) - cov(Y_i(a),X_i^\top \beta_{a,s}^*|S_i=s) +o_P(1) \\
& = Var(Y_i(a) - X_i^\top \beta_{a,s}^*|S_i=s) + o_P(1).
\end{align*}
Therefore, we have $\hat \omega^2_{a,s} = Var(Y_i(a) - X_i^\top \beta_{a,s}^*|S_i=s) + o_P(1)$. By a similar argument, we have
\begin{align*}
\hat \varpi_{a,s} = Var(Y_i(a) - X_i^\top \beta_{a,s}^*|S_i=s) +o_P(1). 
\end{align*}

This implies 
\begin{align*}
\sum_{s \in \mathcal{S}} \sum_{a=0,1} \frac{n_s^2}{n n_{a,s}} \hat \omega_{a,s}^2 =     \sum_{s\in \mathcal{S}}p_s \left[\frac{Var(Y_i(1) - X_i^\top \beta^*_{1,s}|S_i=s)}{\pi_{s}} +\frac{Var(Y_i(0) - X_i^\top \beta^*_{0,s}|S_i=s)}{1-\pi_{s}}\right] + o_P(1)
\end{align*}
and 
\begin{align*}
\sum_{s \in \mathcal{S}} \sum_{a=0,1} \frac{n_s^2}{n n_{a,s}} \hat \varpi_{a,s} =        \sum_{s\in \mathcal{S}}p_s \left[\frac{Var(Y_i(1) - X_i^\top \beta^*_{1,s}|S_i=s)}{\pi_{s}} +\frac{Var(Y_i(0) - X_i^\top \beta^*_{0,s}|S_i=s)}{1-\pi_{s}}\right] + o_P(1).
\end{align*}

\textbf{Step 3: Limit of $\hat \Sigma_{\mathcal V}^{adj}$.} We first note that 
\begin{align*}
\left\vert \frac{1}{n_{a,s}}\sum_{i \in \aleph_{a,s}}P_{a,s,i,i}Y_i \acute{\eps}_{a,s,i} \right\vert \leq \left( \max_{i \in \aleph_{a,s}}P_{a,s,i,i}\right) \left( \frac{1}{n_{a,s}}\sum_{i \in \aleph_{a,s}} |Y_i \acute{\eps}_{a,s,i}|\right) = o_P(1).
\end{align*}
In addition, we have 
\begin{align*}
& \left|\left(    \frac{1}{n_{a,s}} \hat \beta_{a,s}^\top \Gamma_{a,s} \hat \beta_{a,s} \right)^{1/2} - \left(    \frac{1}{n_{a,s}} \sum_{i \in \aleph_{a,s}}(\breve X_i^\top \beta_{a,s}^*)^2 \right)^{1/2}\right| \\
& \leq \left(    \frac{1}{n_{a,s}} \sum_{i \in \aleph_{a,s}}(\breve X_i^\top (\hat \beta_{a,s} - \beta_{a,s}^*))^2 \right)^{1/2} = o_P(1), 
\end{align*}
and 
\begin{align*}
\frac{1}{n_{a,s}} \sum_{i \in \aleph_{a,s}}(\breve X_i^\top \beta_{a,s}^*)^2 = \frac{1}{n_{a,s}} \sum_{i \in \aleph_{a,s}}(\Phi_i - \frac{1}{n_a}\sum_{i \in \aleph_s}\Phi_i)^2 = Var(\Phi_i|S_i=s) + o_P(1),   
\end{align*}
where $\Phi_i = X_i^\top \beta_{a,s}^*$. 
This implies 
\begin{align*}
\frac{1}{n_{a,s}} \hat \beta_{a,s}^\top \Gamma_{a,s} \hat \beta_{a,s} = Var(X_i^\top \beta_{a,s}^*|S_i=s) + o_P(1). 
\end{align*}

Similarly, we have 
\begin{align*}
\frac{1}{n_{s}} \hat \beta_{1,s}^\top \Gamma_{s} \hat \beta_{0,s} & = \frac{1}{4}\left[ (\hat \beta_{1,s}+\hat \beta_{0,s})^\top \Gamma_{s} (\hat \beta_{1,s}+\hat \beta_{0,s}) - (\hat \beta_{1,s}-\hat \beta_{0,s})^\top \Gamma_{s} (\hat \beta_{1,s}-\hat \beta_{0,s})\right] \\
& = \frac{1}{4}\left[ Var(X_i^\top (\beta_{1,s}^* + \beta_{0,s}^*)|S_i=s)-Var(X_i^\top (\beta_{1,s}^* - \beta_{0,s}^*)|S_i=s) \right] + o_P(1).
\end{align*}
Plug the above results into the expression of $\hat  \Sigma_{\mathcal V}^{adj}$ and obtain
\begin{align*}
\hat  \Sigma_{\mathcal V}^{adj} =       \mathbb E Var(X_i^\top (\beta_{1,s}^* - \beta_{0,s}^*)|S_i) + o_P(1).  
\end{align*}

\textbf{Step 4: $\hat \Sigma^{-1} \Omega \convP I_2.$} We have already shown in \textbf{Step 1} that $\hat \tau_{a,s} \convP \mathbb E (Y_i(a)|S_i=s)$. This implies $\hat \Sigma_{\mathcal W} \convP var(\mathbb{E}(Y_i(1)|S_i) - \mathbb{E}(Y_i(0)|S_i)).$ Therefore, combining results in \textbf{Steps 2 and 3}, we have
\begin{align*}
\hat \Sigma = & \left\{\sum_{s\in \mathcal{S}}p_s \left[\frac{Var(Y_i(1) - X_i^\top \beta^*_{1,s}|S_i=s)}{\pi_{s}} +\frac{Var(Y_i(0) - X_i^\top \beta^*_{0,s}|S_i=s)}{1-\pi_{s}} + \right]\right\}1_21_2^\top\\
& + \begin{pmatrix}
\mathbb E Var(X_i^\top (\beta_{1,S_i}^* - \beta_{0,S_i}^*)|S_i) & \mathbb E Var(X_i^\top (\beta_{1,S_i}^* - \beta_{0,S_i}^*)|S_i)\\
\mathbb E Var(X_i^\top (\beta_{1,S_i}^* - \beta_{0,S_i}^*)|S_i) & \mathbb E \left[\frac{Var(X_i^\top \beta_{1,S_i}^*|S_i)}{\pi_{S_i}} + \frac{Var(X_i^\top \beta_{0,S_i}^*|S_i)}{1-\pi_{S_i}} \right] 
\end{pmatrix}\\
&     +var(\mathbb{E}(Y_i(1)|S_i) - \mathbb{E}(Y_i(0)|S_i))1_21_2^\top +o_P(1) \\
& = \sum_{s \in \mathcal{S}}p_s \left\{\frac{\mathbb E[ Var(Y_i(1)|X_i,S_i)|S_i=s] }{\pi_s}+\frac{\mathbb E[ Var(Y_i(0)|X_i,S_i)|S_i=s] }{1-\pi_s} \right\}1_21_2^\top     \\
& + \sum_{s \in \mathcal{S}}\frac{p_s}{(1-\pi_s)\pi_s} \left\{Var( (1-\pi_s)(m_1(X_i,S_i) - X_i^\top \beta_{1,s}^*)|S_i=s) \right\}1_21_2^\top \\
& + \sum_{s \in \mathcal{S}}\frac{p_s}{(1-\pi_s)\pi_s} \left\{Var( \pi_s(m_0(X_i,S_i) - X_i^\top \beta_{0,s}^*)|S_i=s) \right\}1_21_2^\top \\
&+ \sum_{s \in \mathcal{S}} p_s\left\{Var(m_1(X_i,S_i)|S_i=s) + Var(m_0(X_i,S_i)|S_i=s)-2cov(X_i^\top \beta_{1,s}^*,X_i^\top \beta_{0,s}^*\mid S_i=s) \right\}1_21_2^\top \\
& + \begin{pmatrix}
0 & 0\\
0 & \mathbb E \left[\frac{Var(X_i^\top \beta_{1,s}^*|S_i)}{\pi_{S_i}} + \frac{Var(X_i^\top \beta_{0,s}^*|S_i)}{1-\pi_{S_i}} - Var(X_i^\top (\beta_{1,s}^* - \beta_{0,s}^*|S_i))\right] 
\end{pmatrix}\\
&     +var(\mathbb{E}(Y_i(1)|S_i) - \mathbb{E}(Y_i(0)|S_i))1_21_2^\top +o_P(1)\\
& =  \sum_{s \in \mathcal{S}}p_s \left\{\frac{\mathbb E[ Var(Y_i(1)|X_i,S_i)|S_i=s] }{\pi_s}+\frac{\mathbb E[ Var(Y_i(0)|X_i,S_i)|S_i=s] }{1-\pi_s} \right\}1_21_2^\top     \\
& + \mathbb{E}(m_1(X_i,S_i) - m_1(X_i,S_i) - \tau)^2 1_2 1_2^\top \\
& +  \sum_{s \in \mathcal{S}}\frac{p_s}{(1-\pi_s)\pi_s} \left\{Var( (1-\pi_s)m_1(X_i,S_i) + \pi_s m_0(X_i,S_i) - X_i^\top \overline \beta_{s}^*)|S_i=s) \right\}1_21_2^\top \\
& + \begin{pmatrix}
0 & 0\\
0 & \mathbb E \left[\frac{Var(X_i^\top \beta_{1,s}^*|S_i)}{\pi_{S_i}} + \frac{Var(X_i^\top \beta_{0,s}^*|S_i)}{1-\pi_{S_i}} - Var(X_i^\top (\beta_{1,s}^* - \beta_{0,s}^*|S_i))\right] 
\end{pmatrix} + o_P(1)\\
& = \Omega + o_P(1),  
\end{align*}
where the first equality holds by the facts that
\begin{align*}
&    \frac{Var(Y_i(1)-X_i^\top \beta_{1,s}^*|S_i=s)}{\pi_s} \\
& = \frac{\mathbb E[ Var(Y_i(1)|X_i,S_i)|S_i=s] }{\pi_s} + \frac{Var(m_1(X_i,S_i) - X_i^\top \beta_{1,s}^*|S_i=s)}{\pi_s}  \\
& = \frac{\mathbb E[ Var(Y_i(1)|X_i,S_i)|S_i=s] }{\pi_s} + \frac{Var( (1-\pi_s)(m_1(X_i,S_i) - X_i^\top \beta_{1,s}^*)|S_i=s)}{(1-\pi_s)\pi_s}\\
& + Var(m_1(X_i,S_i)|S_i=s) - Var(X_i^\top \beta_{1,s}^*|S_i=s),   
\end{align*}
\begin{align*}
&    \frac{Var(Y_i(0)-X_i^\top \beta_{0,s}^*|S_i=s)}{1-\pi_s} \\
& = \frac{\mathbb E[ Var(Y_i(0)|X_i,S_i)|S_i=s] }{1-\pi_s} + \frac{Var( \pi_s(m_0(X_i,S_i) - X_i^\top \beta_{0,s}^*)|S_i=s)}{(1-\pi_s)\pi_s}\\
& + Var(m_0(X_i,S_i)|S_i=s) - Var(X_i^\top \beta_{0,s}^*|S_i=s),
\end{align*}
and
\begin{align*}
&   Var(X_i^\top(\beta_{1,s}^* - \beta_{0,s}^*)|S_i=s) - Var(X_i^\top(\beta_{1,s}^* )|S_i=s) - Var(X_i^\top(\beta_{0,s}^* )|S_i=s)\\ &=   
-2cov(X_i^\top \beta_{1,s}^*,X_i^\top \beta_{0,s}^*|S_i=s) \\
& = 2cov(m_1(X_i,S_i) - X_i^\top \beta_{1,s}^*,m_0(X_i,S_i) - X_i^\top \beta_{0,s}^*|S_i=s) \\
& - 2cov(m_1(X_i,S_i),m_0(X_i,S_i)|S_i=s), 
\end{align*}
the second equality holds by the fact that 
\begin{align*}
& \sum_{s \in \mathcal{S}}\frac{p_s}{(1-\pi_s)\pi_s} \left\{Var( (1-\pi_s)m_1(X_i,S_i) + \pi_s m_0(X_i,S_i) - X_i^\top \overline \beta_{s}^*)|S_i=s) \right\}\\
& = \sum_{s \in \mathcal{S}}\frac{p_s}{(1-\pi_s)\pi_s} \left\{Var( (1-\pi_s)(m_1(X_i,S_i) - X_i^\top \beta_{1,s}^*)|S_i=s) \right\} \\
& + \sum_{s \in \mathcal{S}}\frac{p_s}{(1-\pi_s)\pi_s} \left\{Var( \pi_s(m_0(X_i,S_i) - X_i^\top \beta_{0,s}^*)|S_i=s) \right\}\\
& + 2cov(m_1(X_i,S_i) - X_i^\top \beta_{1,s}^*,m_0(X_i,S_i) - X_i^\top \beta_{0,s}^*|S_i=s), 
\end{align*}
and the last equality holds by the fact that 
\begin{align*}
& \sum_{s \in \mathcal{S}}\frac{p_s}{(1-\pi_s)\pi_s} \left\{Var( (1-\pi_s)m_1(X_i,S_i) + \pi_s m_0(X_i,S_i)|S_i=s) \right\}\\
& = \mathbb E \left[\frac{Var(X_i^\top \beta_{1,s}^*|S_i)}{\pi_{S_i}} + \frac{Var(X_i^\top \beta_{0,s}^*|S_i)}{1-\pi_{S_i}} - Var(X_i^\top (\beta_{1,s}^* - \beta_{0,s}^*|S_i))\right] \\
& + \sum_{s \in \mathcal{S}}\frac{p_s}{(1-\pi_s)\pi_s} \left\{Var( (1-\pi_s)m_1(X_i,S_i) + \pi_s m_0(X_i,S_i) - X_i^\top \overline \beta_{s}^*)|S_i=s) \right\}. 
\end{align*}
Given $\liminf_{n \rightarrow \infty}\lambda_{\min}(\Omega) >0$, this concludes the proof of $\hat \Sigma^{-1} \Omega \convP I_2$. 

\textbf{Step 5: Asymptotic normality of $\hat \tau^*$.}
We see that $\Omega_{1,1} = \Omega_{1,2}$ and 
\begin{align*}
\liminf_{n \rightarrow \infty}    (\Omega_{2,2} - \Omega_{1,1}) = \liminf_{n \rightarrow \infty} \sum_{s \in \mathcal{S}}\frac{p_s}{\pi_s(1-\pi_s)}(V_s'-V_s) = \liminf_{n \rightarrow \infty} \sum_{s \in \mathcal{S}}\frac{p_s}{\pi_s(1-\pi_s)}Var(X_i^\top \overline{\beta}_s^*|S_i=s) > 0. 
\end{align*}
This implies $\hat w \convP 1$, and thus, 
\begin{align*}
|\sqrt{n}(\hat \tau^* - \hat \tau^{adj})| \leq|1-\hat w| \left[|\sqrt{n}(\hat \tau^{adj} - \tau)| +|\sqrt{n}(\hat \tau^{unadj} - \tau)| \right]   = o_P(1).  
\end{align*}
Therefore, we have
\begin{align*}
\sqrt{n} \left[(\hat w, 1-\hat w)\hat \Sigma (\hat w, 1-\hat w)^\top\right]^{-1/2} ( \hat \tau^* - \tau) \convD \N(0,1)
\end{align*}
and $(\hat w, 1-\hat w)\hat \Sigma (\hat w, 1-\hat w)^\top = \Omega_{1,1} + o_P(1)< \Omega_{2,2}$.

\section{Technical Lemmas}\label{sec:lemma}
\begin{lem}
	Recall the definitions of $\mathcal U_n$, $\mathcal V_n$, and $\mathcal W_n$ in \eqref{eq:Un}, \eqref{eq:Vn}, and \eqref{eq:Wn}, respectively. Suppose Assumptions \ref{ass:assignment1}--\ref{ass:omega} holds.  
	Then, we have
	\begin{align*}
	\mathcal U_{n} \convD \mathcal U \stackrel{d}{=}  \N\left(\begin{pmatrix}
	0 \\
	0
	\end{pmatrix}, \begin{pmatrix}
	\mathbb{E}\left(\frac{\omega_{1,S_i,\infty}^2}{\pi_{S_i}}+\frac{\omega_{0,S_i,\infty}^2}{1-\pi_{S_i}}\right) & 	\mathbb{E}\left(\frac{\varpi_{1,S_i,\infty}}{\pi_{S_i}}+\frac{\varpi_{0,S_i,\infty}}{1-\pi_{S_i}}\right)\\
	\mathbb{E}\left(\frac{\varpi_{1,S_i,\infty}}{\pi_{S_i}}+\frac{\varpi_{0,S_i,\infty}}{1-\pi_{S_i}}\right) &  \mathbb{E}\left(\frac{\mathbb{E}(\eps_i^2(1)|S_i)}{\pi_{S_i}}+\frac{\mathbb{E}(\eps_i^2(0)|S_i)}{1-\pi_{S_i}}\right)
	\end{pmatrix} \right),
	\end{align*}
	\begin{align*}
	\mathcal V_{n} \convD \mathcal V \stackrel{d}{=} \N\left(\begin{pmatrix}
	0 \\
	0
	\end{pmatrix},\begin{pmatrix}
	\mathbb{E}	var(\phi_i(1) - \phi_i(0)|S_i) & \mathbb{E}	var(\phi_i(1) - \phi_i(0)|S_i)\\
	\mathbb{E}	var(\phi_i(1) - \phi_i(0)|S_i) & \mathbb{E}\left[\frac{var(\phi_i(1)|S_i)}{\pi_{S_i}} + \frac{var(\phi_i(0)|S_i)}{(1-\pi_{S_i})}\right]
	\end{pmatrix}  \right),
	\end{align*}
	and 
	\begin{align*}
	\mathcal W_n \convD \mathcal W \stackrel{d}{=} \N(0,var(\mathbb{E}(Y_i(1)|S_i) - \mathbb{E}(Y_i(0)|S_i))).
	\end{align*}
	In addition, $(\mathcal U,\mathcal V,\mathcal W)$ are independent.  
	\label{lem:clt}
\end{lem}
\begin{proof}
	Following \cite{BCS17}, we define $\{(X_i^s,\eps_i^s(1),\eps_i^s(0)): 1\leq i \leq n\}$ as a sequence of i.i.d. random variables with marginal distributions equal to the distribution of $(X_i,\eps_i(1),\eps_i(0))|S_i = s$ and $N_s = \sum_{i =1}^n1\{S_i <s\}$. We order units by strata and then by $A_i = 1$ first and $A_i = 0$ second within each stratum. This means for the $s$th stratum, units indexed from $N_s+1$ to $N_s + n_{1,s}$ are treated and units indexed from $N_s + n_{1,s}+1$ to $N_s + n_s$ are untreated. Then, we have $(Y_i(a)\mid S_i=s) \stackrel{d}{=} Y_i^s(a)$ where 
	\begin{align*}
	Y_i^s(a) = \mu_a(X_i^s,s) + \eps_i^s(a) \quad \text{and} \quad \mu_a(x,s) = \mathbb{E}(Y_i(a)|X_i=x,S_i=s). 
	\end{align*}
	We further define $\tilde M_{a,s,i,j}$ as the $(i,j)$th entry of the $n_{a,s} \times n_{a,s}$ matrix,  
	$\tilde M_{a,s} = I_{n_{a,s}}-\Xi_{a,s} (\Xi_{a,s}^\top  \Xi_{a,s})^{-1} \Xi_{a,s}^\top $, $\Xi_{1,s} = ((\breve X_{N_s+1}^s)^\top,\cdots,( \breve X_{N_s+n_{1,s}}^s)^\top)^\top$ is an $n_{1,s} \times k_n$ matrix, $\Xi_{0,s} = ((\breve X_{N_s+1+n_{1,s}}^s)^\top,\cdots,(\breve X_{N_s+n_{s}}^s)^\top)^\top$ is an $n_{0,s} \times k_n$ matrix,
	\begin{align*}
	\breve X_i^s = 	X_i^s - \frac{1}{n_{s}}\sum_{j=N_s+1}^{N_s+n_{s}}X_j^s, \quad \text{if}\quad N_s+1 \leq i \leq N_s+n_{s},
	\end{align*}
	$\tilde \gamma_{a,s,n} = 1_{n_{a,s}}^\top \tilde M_{a,s} 1_{n_{a,s}}$, and $\phi_i^s(a) = \mathbb{E}(Y_i^s(a)|X_i^s)- \mathbb{E}(Y_i^s(a)) \stackrel{d}{=} \mathbb E (Y_i(a)|X_i,S_i=s) - \mathbb E (Y_i(a)|S_i=s)$.

	Conditional on $(A^{(n)},S^{(n)})$, we have
	\begin{align*}
	(\mathcal U_n,\mathcal V_n) |(A^{(n)},S^{(n)}) \stackrel{d}{=} (\tilde{ \mathcal U}_n, \tilde{ \mathcal  V}_n)|(A^{(n)},S^{(n)}),
	\end{align*}
	where
	\begin{align*}
	\tilde {\mathcal	U}_{n} = \sum_{s \in \mathcal{S}} \hat p_s \sqrt{n} \begin{pmatrix}
	\tilde \gamma_{1,s,n}^{-1} n_{1,s}^{-1}\left[\sum_{i = N_s+1}^{N_s+n_{1,s}} \left(\sum_{j = N_s+1}^{N_s+n_{1,s}} \tilde M_{1,s,i,j}\right)\eps_{i}^s(1)\right] \\
	- 	\tilde \gamma_{0,s,n}^{-1} n_{0,s}^{-1}\left[\sum_{i = N_s+n_{1,s}+1}^{N_s+n_{s}}\left(\sum_{j = N_s+n_{1,s}+1}^{N_s+n_{s}} \tilde M_{0,s,i,j}\right)\eps_{i}^s(0)\right] \\
	n_{1,s}^{-1}\left[\sum_{i = N_s+1}^{N_s+n_{1,s}} \eps_{i}^s(1)\right] - n_{0,s}^{-1}\left[\sum_{i = N_s+n_{1,s}+1}^{N_s+n_{s}} \eps_{i}^s(0)\right]
	\end{pmatrix} \quad \text{and}
	\end{align*}
	\begin{align*}
	\tilde {\mathcal V}_{n} = \sum_{s \in \mathcal{S}} \hat p_s \sqrt{n} \begin{pmatrix}
	\frac{1}{n_{s}}\sum_{i = N_s+1}^{N_s+n_{s}} \left[\phi_i^s(1) - \phi_i^s(0)\right] \\
	\frac{1}{n_{1,s}}\sum_{i =N_s+1}^{N_s+n_{1,s}}\phi_i^s(1) -
	\frac{1}{n_{0,s}}\sum_{i=N_s+n_{1,s}+1}^{N_s+n_{s}}\phi_i^s(0) 
	\end{pmatrix}.
	\end{align*}
	In addition, because $\mathcal W_n$ only depends on $S^{(n)}$, it implies the joint distribution of $(\mathcal U_n,\mathcal V_n,\mathcal W_n)$ and $(\tilde {\mathcal{U}}_n, \tilde {\mathcal{V}}_n,\mathcal W_n)$ are the same, i.e., 
	\begin{align}
	(\mathcal U_n,\mathcal V_n,\mathcal W_n) \stackrel{d}{=} 	(\tilde {\mathcal{U}}_n, \tilde {\mathcal{V}}_n,\mathcal W_n).
	\label{eq:tildeUV}
	\end{align}
	
	We aim to show that, for any $t_1 \in \Re^2$, $t_2 \in \Re^2$, and $t_3 \in \Re$,  
	\begin{align*}
	\left|\mathbb{P}\left(\mathcal U_n \leq t_1,\mathcal V_n \leq t_2,\mathcal W_n \leq t_3 \right) - \mathbb{P}\left(\mathcal U \leq t_1\right)\mathbb{P}\left(\mathcal V \leq t_2\right)\mathbb{P}\left(\mathcal W \leq t_3 \right)\right| = o(1).
	\end{align*}	
	
	By the above definition, we have
	\begin{align*}
	& \left|\mathbb{P}\left(\mathcal U_n \leq t_1,\mathcal V_n \leq t_2,\mathcal W_n \leq t_3 \right) - \mathbb{P}\left(\mathcal U \leq t_1\right)\mathbb{P}\left(\mathcal V \leq t_2\right)\mathbb{P}\left(\mathcal W \leq t_3 \right)\right| \\
	& = \left|\mathbb{P}\left(\tilde {\mathcal{U}}_n \leq t_1,\tilde {\mathcal{V}}_n \leq t_2,\mathcal W_n \leq t_3 \right) - \mathbb{P}\left(\mathcal U \leq t_1\right)\mathbb{P}\left(\mathcal V \leq t_2\right)\mathbb{P}\left(\mathcal W \leq t_3 \right)\right| \\
	& = 	\left|\mathbb{E}\mathbb{P}\left(\tilde {\mathcal{U}}_n \leq t_1, \tilde {\mathcal{V}}_n \leq t_2\bigg|A^{(n)},S^{(n)}\right) 1\{\mathcal W_n \leq t_3\} - \mathbb{P}\left(\mathcal U \leq t_1\right)\mathbb{P}\left(\mathcal V \leq t_2\right)\mathbb{P}\left(\mathcal W \leq t_3 \right)\right|\\
	& \leq 	\mathbb{E}\left|\mathbb{P}\left(\tilde {\mathcal{U}}_n \leq t_1, \tilde {\mathcal{V}}_n \leq t_2\bigg|A^{(n)},S^{(n)} \right)-\mathbb{P}\left(\mathcal U \leq t_1\right)\mathbb{P}\left(\mathcal V \leq t_2\right)\right| + |\mathbb{P}(\mathcal W_n \leq t_3) - \mathbb{P}(\mathcal W\leq t_3)| \\
	& \leq \mathbb{E}\left|\mathbb{P}\left(\tilde {\mathcal{U}}_n \leq t_1, \tilde {\mathcal{V}}_n \leq t_2\bigg|A^{(n)},S^{(n)}\right)-\mathbb{P}\left(\mathcal U \leq t_1\right)\mathbb{P}\left(\mathcal V \leq t_2\right)\right| + o(1),
	\end{align*}	
	where first equality is by \eqref{eq:tildeUV} and the second inequality is because by the Lindeberg central limit theorem, $\mathcal W_n \convD \mathcal W$. Therefore, it suffices to show 
	\begin{align*}
	\mathbb{P}\left(\tilde {\mathcal{U}}_n \leq t_1, \tilde {\mathcal{V}}_n \leq t_2\bigg|A^{(n)},S^{(n)}\right)-\mathbb{P}\left(\mathcal U \leq t_1\right)\mathbb{P}\left(\mathcal V \leq t_2\right) \convP 0.
	\end{align*}
	
	Let $\Xi^{(n)} = \{X_i^s\}_{i \in [n],s\in \mathcal{S}}$. Then, $\tilde {\mathcal{V}}_n$ belongs to the sigma field generated by $(A^{(n)},S^{(n)},\Xi^{(n)})$. We have
	\begin{align}
	& \left\vert \mathbb{P}\left(\tilde {\mathcal{U}}_n \leq t_1, \tilde {\mathcal{V}}_n \leq t_2\bigg|A^{(n)},S^{(n)}\right)-\mathbb{P}\left(\mathcal U \leq t_1\right)\mathbb{P}\left(\mathcal V \leq t_2\right) \right\vert \notag \\
	& \leq \mathbb{E}\left[\left\vert \mathbb{P}\left(\tilde {\mathcal{U}}_n \leq t_1\bigg|A^{(n)},S^{(n)},\Xi^{(n)}\right)-\mathbb{P}\left(\mathcal U \leq t_1\right)\right\vert \bigg|A^{(n)},S^{(n)}  \right] + \left|\mathbb{P}\left(\tilde {\mathcal{V}}_n \leq t_2\bigg|A^{(n)},S^{(n)}\right) - \mathbb{P}\left(\mathcal V \leq t_2\right)\right|.
	\label{eq:sufficient}
	\end{align}	
	
	We show the two terms on the RHS of \eqref{eq:sufficient} vanish in probability in the following two steps. 
	
	\textbf{Step 1: The first term on the RHS of \eqref{eq:sufficient}}. It suffices to show that 
	\begin{align*}
	\mathbb{P}\left(\tilde {\mathcal{U}}_n \leq t_1\bigg|A^{(n)},S^{(n)},\Xi^{(n)}\right)-\mathbb{P}\left(\mathcal U \leq t_1\right)  = o_P(1).
	\end{align*}
	
	Let
	\begin{align*}
	Z_i = \begin{cases}
	\left((n_{s}/n_{1,s}) \tilde \gamma_{1,s,n}^{-1}\theta_i \eps_{i}^s(1), (n_{s}/n_{1,s}) \eps_{i}^s(1)\right)^\top & \text{if} \quad N_s+1 \leq i \leq N_s +n_{1,s} \\
	\left((n_{s}/n_{0,s}) \tilde \gamma_{0,s,n}^{-1} \theta_i \eps_{i}^s(0), (n_{s}/n_{0,s}) \eps_{i}^s(0)\right)^\top & \text{if} \quad N_s+ n_{1,s}+1 \leq i \leq N_s +n_{s}
	\end{cases},
	\end{align*}
	where 
	\begin{align*}
	\theta_i = \begin{cases}
	\sum_{j = N_s+1}^{N_s+n_{1,s}} \tilde M_{1,s,i,j} & \text{if} \quad N_s+1 \leq i \leq N_s +n_{1,s} \\
	\sum_{j = N_s+n_{1,s}+1}^{N_s+n_{s}} \tilde M_{0,s,i,j} & \text{if} \quad N_s+ n_{1,s}+1 \leq i \leq N_s +n_{s}
	\end{cases}.
	\end{align*}
	Then, we have
	\begin{align*}
	\tilde {\mathcal{U}}_n = \frac{1}{\sqrt{n}}\sum_{i\in [n]} Z_i.
	\end{align*}
	
	By construction,
	\begin{align*}
	\max_{N_s+1 \leq i \leq N_s+n_{1,s}}\left|\sum_{j = N_s+1}^{N_s+n_{1,s}} \tilde M_{1,s,i,j}\right| \quad \text{and} \quad \max_{ i \in \aleph_{1,s}} \left|\sum_{j \in \aleph_{1,s}}M_{1,s,i,j}\right|
	\end{align*}
	share the same distribution conditional on $(A^{(n)},S^{(n)})$, and thus, unconditionally. Same for
	\begin{align*}
	\max_{N_s+ n_{1,s}+1 \leq i \leq N_s+n_{s}}\left|\sum_{j = N_s+ n_{1,s}+1 }^{N_s+n_{s}} \tilde M_{0,s,i,j}\right| \quad \text{and} \quad \max_{ i \in \aleph_{0,s}} \left|\sum_{j \in \aleph_{0,s}}M_{0,s,i,j}\right|.	
	\end{align*}
	Therefore, Assumption \ref{ass:linear}(v) implies 
	\begin{align}
	\max_{i \in [n]}	|\theta_i| = o_P(n^{1/2})
	\label{eq:theta}
	\end{align} 
	
	In addition, we note that $\left(\hat p_s, n_{1,s},n_{0,s}, n_{s}, \tilde \gamma_{1,s,n}, \tilde \gamma_{0,s,n}, \theta_i\right)$ belong to the sigma field generated by $(A^{(n)},S^{(n)},\Xi^{(n)})$ and $\{Z_i\}_{i \in [n]}$ are independent and mean-zero conditional on $(A^{(n)},S^{(n)},\Xi^{(n)})$. Define 
	\begin{align*}
	H_n & = \frac{1}{n}\sum_{i \in [n]} \mathbb{E}\left(Z_iZ_i^\top\bigg|A^{(n)},S^{(n)},\Xi^{(n)}\right) \\
	& = \frac{1}{n}\sum_{s \in \mathcal{S}} \biggl[\sum_{i = N_s+1}^{N_s+n_{1,s}} \begin{pmatrix}
	(n_{s}/n_{1,s})^2 \tilde \gamma_{1,s,n}^{-2} \theta_i^2 \eta_i^2 & (n_{s}/n_{1,s})^2 \tilde \gamma_{1,s,n}^{-1} \theta_i \eta_i^2 \\
	(n_{s}/n_{1,s})^2 \tilde \gamma_{1,s,n}^{-1} \theta_i \eta_i^2 & (n_{s}/n_{1,s})^2  \eta_i^2
	\end{pmatrix} \\
	& + \sum_{i = N_s+n_{1,s}+1}^{N_s+n_{s}} \begin{pmatrix}
	(n_{s}/n_{0,s})^2 \tilde \gamma_{0,s,n}^{-2} \theta_i^2 \eta_i^2 & (n_{s}/n_{0,s})^2 \tilde \gamma_{0,s,n}^{-1} \theta_i \eta_i^2 \\
	(n_{s}/n_{0,s})^2 \tilde \gamma_{0,s,n}^{-1} \theta_i \eta_i^2 & (n_{s}/n_{0,s})^2  \eta_i^2
	\end{pmatrix} \biggr],
	\end{align*}
	where 
	\begin{align*}
	\eta_i^2 = \begin{cases}
	\mathbb{E}\left[(\eps_{i}^s(1))^2\bigg|A^{(n)},S^{(n)},\Xi^{(n)}\right] & \text{if} \quad N_s+1 \leq i \leq N_s +n_{1,s} \\
	\mathbb{E}\left[(\eps_{i}^s(0))^2\bigg|A^{(n)},S^{(n)},\Xi^{(n)}\right] & \text{if} \quad N_s+ n_{1,s}+1 \leq i \leq N_s +n_{s}
	\end{cases}.
	\end{align*}
	
	By construction, we have 
	\begin{align*}
	\eta_i^2 \stackrel{d}{=} \begin{cases}
	\mathbb{E}(\eps_i^2(1)|X_i,S_i=s) & \text{if} \quad N_s+1 \leq i \leq N_s +n_{1,s} \\
	\mathbb{E}(\eps_i^2(0)|X_i,S_i=s)  & \text{if} \quad N_s+ n_{1,s}+1 \leq i \leq N_s +n_{s}
	\end{cases}.
	\end{align*}
	
	Further define $Z_i = (Z_{i,1},Z_{i,2})^\top$ and
	\begin{align*}
	L_n = \max_{\ell=1,2} \frac{1}{n^{3/2}}\sum_{i \in [n]} \mathbb{E}\left(|Z_{i,\ell}^3| \mid A^{(n)},S^{(n)},\Xi^{(n)}\right).    
	\end{align*}
	Then, we have
	\begin{align*}
	L_n & \leq \frac{\overline \sigma_n^3}{n^{3/2}} \sum_{s \in \mathcal{S}}\biggl[ (n/n_{1,s})^3\sum_{i = N_s+1}^{N_s+n_{1,s}} \tilde \gamma_{1,s,n}^{-3} |\theta_i|^3 + (n/n_{0,s})^3\sum_{i = N_s+n_{1,s}+1}^{N_s+n_{s}} \gamma_{0,s,n}^{-3} |\theta_i|^3 \biggr] \\
	& \leq \frac{\overline{\sigma}^3_n n^3}{\min_{a=0,1, s\in \mathcal{S}}n_{a,s}^3 \min_{a=0,1, s\in \mathcal{S}} \tilde \gamma_{a,s,n}^2} \frac{\max_{ i \in [n]} |\theta_i|}{\sqrt{n}} \\
	&  \convP 0,
	\end{align*}
	where $\tilde \gamma_{a,s,n} \stackrel{d}{=} \gamma_{a,s,n}$, 
	$$\overline \sigma_n^3 = \max_{i \in [n],a=0,1,s \in \mathcal{S}}\mathbb{E}\left[|\eps_{i}^s(a)|^3 \mid A^{(n)},S^{(n)},\Xi^{(n)} \right] \stackrel{d}{=} \max_{ i \in [n],a=0,1,s\in \mathcal{S}} \mathbb{E}(|\eps_i^3(a)|\mid X_i,S_i=s),$$ the second inequality holds by the facts that 
	\begin{align*}
	\sum_{i = N_s+1}^{N_s+n_{1,s}}  |\theta_i|^3 \leq \max_{ i \in [n]}|\theta_i| 	\sum_{i = N_s+1}^{N_s+n_{1,s}}  |\theta_i|^2 = \max_{ i \in [n]}|\theta_i| \tilde \gamma_{1,s,n}
	\end{align*}
	and 
	\begin{align*}
	\sum_{i = N_s+n_{1,s}+1}^{N_s+n_{s}}  |\theta_i|^3 \leq \max_{ i \in [n]}|\theta_i| 	\sum_{i = N_s+n_{1,s}+1}^{N_s+n_{s}}  |\theta_i|^2 = \max_{ i \in [n]}|\theta_i| \tilde \gamma_{0,s,n},
	\end{align*}
	and the last convergence is by \eqref{eq:theta}. 
	
	By the Yurinskii's coupling (\citet[Theorem 10]{P02}), there exists a version of $\breve {\mathcal{U}}_n$ and a universal constant $C_0$ such that 
	\begin{align*}
	\breve {\mathcal{U}}_n \mid A^{(n)},S^{(n)},\Xi^{(n)} \stackrel{d}{=} \N(0,H_n)    
	\end{align*}
	and
	\begin{align}\label{eq:coupling}
	\mathbb{P}\left(\left\Vert \breve {\mathcal{U}}_n - \tilde {\mathcal{U}}_n \right\Vert_2 \geq 3\delta_n \bigg|A^{(n)},S^{(n)},\Xi^{(n)}\right) \leq C_0 2L_n \delta_n^{-3} \left(1 + \frac{|\log(1/(2L_n\delta_n^{-3})) |}{2} \right).
	\end{align}
	Because $L_n \convP 0$, we can choose $\delta_n = o(1)$ such that $L_n \delta_n^{-3} = o(1)$, which implies 
	$$C_0 2L_n \delta_n^{-3} \left(1 + \frac{|\log(1/(2L_n\delta_n^{-3})) |}{2} \right) \convP 0,$$ and thus, 
	\begin{align*}
	\left\Vert \breve {\mathcal{U}}_n - \tilde {\mathcal{U}}_n \right\Vert_2 = o_P(1).
	\end{align*}
	
	Furthermore, note that 
	\begin{align*}
	H_n 
	& = \sum_{s \in \mathcal{S}} \frac{n_s^2}{n n_{1,s}} \biggl[\frac{1}{n_{1,s}}\sum_{i = N_s+1}^{N_s+n_{1,s}} \begin{pmatrix}
	\tilde \gamma_{1,s,n}^{-2} \theta_i^2 \eta_i^2 & \tilde \gamma_{1,s,n}^{-1} \theta_i \eta_i^2 \\
	\tilde \gamma_{1,s,n}^{-1} \theta_i \eta_i^2 &   \eta_i^2
	\end{pmatrix}\biggr] \\
	& + \sum_{s \in \mathcal{S}} \frac{n_s^2}{n n_{0,s}} \biggl[\frac{1}{n_{0,s}}\sum_{i = N_s+n_{1,s}+1}^{N_s+n_{s}} \begin{pmatrix}
	\tilde \gamma_{0,s,n}^{-2} \theta_i^2 \eta_i^2 &  \tilde \gamma_{0,s,n}^{-1} \theta_i \eta_i^2 \\
	\tilde \gamma_{0,s,n}^{-1} \theta_i \eta_i^2 &   \eta_i^2
	\end{pmatrix} \biggr] \\ 
	& \stackrel{d}{=} \sum_{s \in \mathcal{S}} \frac{n_s^2}{n n_{1,s}} \biggl[\begin{pmatrix}
	\gamma_{1,s,n}^{-2} \sigma_{1,s,n}^2 &  \gamma_{1,s,n}^{-1} \rho_{1,s,n}\\
	\gamma_{1,s,n}^{-1} \rho_{1,s,n} &  \tilde \gamma_{1,s,n}^{-1} \frac{1}{n_{1,s}}\sum_{i \in \aleph_{1,s}} \left[\mathbb{E}(\eps_{i}^2(1)|X_i,S_i=s) \right]
	\end{pmatrix}\biggr] \\
	&+ \sum_{s \in \mathcal{S}} \frac{n_s^2}{n n_{0,s}} \biggl[\begin{pmatrix}
	\gamma_{0,s,n}^{-2} \sigma_{0,s,n}^2 &  \tilde \gamma_{0,s,n}^{-1} \rho_{0,s,n} \\
	\gamma_{0,s,n}^{-1} \rho_{0,s,n} &  \frac{1}{n_{0,s}}\sum_{i \in \aleph_{0,s}} \left[\mathbb{E}(\eps_{i}^2(0)|X_i,S_i=s) \right]
	\end{pmatrix} \biggr] \\
	& \convP \sum_{s \in \mathcal{S}}\left[\frac{p_s}{\pi_s} \begin{pmatrix}
	\omega_{1,s,n}^2 & \varpi_{1,s,n} \\
	\varpi_{1,s,n} &  \mathbb{E}(\eps_i^2(1)|S_i=s)
	\end{pmatrix} + \frac{p_s}{1-\pi_s} \begin{pmatrix}
	\omega_{0,s,n}^2 & \varpi_{0,s,n} \\
	\varpi_{0,s,n} &  \mathbb{E}(\eps_i^2(0)|S_i=s)
	\end{pmatrix}
	\right] \\
	& = H,
	\end{align*}
	where the second last line holds by Assumption \ref{ass:omega} and $H$ is nonrandom. 
	
	We can write  $\breve{\mathcal{U}}_n = H^{1/2}_n \mathcal{G}_2$, where $\mathcal{G}_2 $ is a two-dimensional standard Gaussian vector such that $\mathcal G_2 \indep A^{(n)},S^{(n)},\Xi^{(n)}$. Further denote $\mathcal U = H^{1/2} \mathcal{G}_2$ and note that
	\begin{align*}
	||H^{1/2}_n- H^{1/2}||_{op} \leq ||H_n- H||_{op}^{1/2} = o_P(1),
	\end{align*}
	where  the inequality is by \citet[Theorem X.1.1]{bhatia2013} with $f(u) = u^{1/2}$. Let $F_n(\delta') = \{||H_n- H||_{op}^{1/2} \leq (\delta')^2 \}$ for any $(\delta')^2>0$. Then, $\mathbb{P}(F_n(\delta')) \rightarrow 1$ and $F_n(\delta')$ belongs to the sigma field generated by $(A^{(n)},S^{(n)},\Xi^{(n)})$. On $F_n(\delta')$, we have
	\begin{align*}
	&  \mathbb{P}\left(\breve {\mathcal{U}}_n \leq t_1\bigg|A^{(n)},S^{(n)},\Xi^{(n)}\right)-\mathbb{P}\left(U \leq t_1\right)  \\
	& =  \mathbb{P}\left( H_n^{1/2} \mathcal{G}_2 \leq t_1\bigg|A^{(n)},S^{(n)},\Xi^{(n)}\right)-\mathbb{P}\left(H^{1/2} \mathcal{G}_2 \leq t_1 \right) \\
	& \leq \mathbb{P}\left( H^{1/2}\mathcal{G}_2  \leq t_1+ 1_2||H^{1/2}_n - H^{1/2}||_{op} ||\mathcal{G}_2||_2 \bigg|A^{(n)},S^{(n)},\Xi^{(n)} \right) - \mathbb{P}( H^{1/2} \mathcal{G}_2 \leq t_1) \\
	& \leq \mathbb{P}\left( H^{1/2}\mathcal{G}_2  \leq t_1+ 1_2\delta' \right) - \mathbb{P}( H^{1/2} \mathcal{G}_2 \leq t_1) + \mathbb{P}(||\mathcal{G}_2||_2 \geq 1/\delta').
	\end{align*}
	By Assumption \ref{ass:linear}(ii), we see that $H_{1,1}$ and $H_{2,2}$ are bounded above from zero. Therefore, by \citet[Lemma A.1]{CCK14}, we have
	\begin{align*}
	\mathbb{P}\left( H^{1/2}\mathcal{G}_2  \leq t_1+ 1_2\delta' \right) - \mathbb{P}( H^{1/2} \mathcal{G}_2 \leq t_1) \leq C \delta',
	\end{align*}
	which implies 
	\begin{align*}
	\mathbb{P}\left(\breve {\mathcal{U}}_n \leq t_1\bigg|A^{(n)},S^{(n)},\Xi^{(n)}\right)-\mathbb{P}\left(\mathcal U \leq t_1\right) \leq C \delta' +  \mathbb{P}(||\mathcal{G}_2||_2 \geq 1/\delta').
	\end{align*}
	
	By a similar argument, we have 
	\begin{align*}
	&  \mathbb{P}\left(\breve {\mathcal{U}}_n > t_1\bigg|A^{(n)},S^{(n)},\Xi^{(n)}\right)-\mathbb{P}\left(\mathcal U > t_1\right)  \\
	& =  \mathbb{P}\left( H_n^{1/2} \mathcal{G}_2 > t_1\bigg|A^{(n)},S^{(n)},\Xi^{(n)}\right)-\mathbb{P}\left(H^{1/2} \mathcal{G}_2 > t_1 \right) \\
	& \leq \mathbb{P}\left( H^{1/2}\mathcal{G}_2  > t_1- 1_2||H^{1/2}_n - H^{1/2}||_{op} ||\mathcal{G}_2||_2 \bigg|A^{(n)},S^{(n)},\Xi^{(n)} \right) - \mathbb{P}( H^{1/2} \mathcal{G}_2 > t_1) \\
	& \leq \mathbb{P}\left( H^{1/2}\mathcal{G}_2  > t_1 - 1_2\delta' \right) - \mathbb{P}( H^{1/2} \mathcal{G}_2 > t_1) + \mathbb{P}(||\mathcal{G}_2||_2 \geq 1/\delta') \\
	& \leq C \delta' +\mathbb{P}(||\mathcal{G}_2||_2 \geq 1/\delta').
	\end{align*}
	Combining these two bounds, we have
	\begin{align*}
	\left|\mathbb{P}\left(\mathcal U \leq t_1\right)	- \mathbb{P}\left(\breve {\mathcal{U}}_n \leq t_1\bigg|A^{(n)},S^{(n)},\Xi^{(n)}\right) \right| \leq C \delta' +\mathbb{P}(||\mathcal{G}_2||_2 \geq 1/\delta') + 1\{F_n^c(\delta')\}.
	\end{align*}
	By letting $n \rightarrow \infty$ ( so that $\delta_n \rightarrow 0$) followed by $\delta' \downarrow 0$, we have
	\begin{align*}
	& \left| \mathbb{P}\left(\mathcal U \leq t_1\right)	- \mathbb{P}\left(\tilde {\mathcal{U}}_n \leq t_1\bigg|A^{(n)},S^{(n)},\Xi^{(n)}\right)\right| \\
	&\leq \left| \mathbb{P}\left(\mathcal U \leq t_1\right)	- \mathbb{P}\left(\breve {\mathcal{U}}_n \leq t_1\bigg|A^{(n)},S^{(n)},\Xi^{(n)}\right)\right| \\
	& + \left| \mathbb{P}\left(\breve {\mathcal{U}}_n \leq t_1\bigg|A^{(n)},S^{(n)},\Xi^{(n)}\right)	- \mathbb{P}\left(\tilde {\mathcal{U}}_n \leq t_1\bigg|A^{(n)},S^{(n)},\Xi^{(n)}\right)\right| \\
	& \leq C\delta' + \mathbb P(||\mathcal G_2||_2 \geq 1/\delta') + 1\{F_n^c(\delta')\} + \mathbb{P}\left(t_1 - 3\delta_n 1_2 \leq \breve {\mathcal{U}}_n \leq t_1 + 3\delta_n  1_2 \bigg|A^{(n)},S^{(n)},\Xi^{(n)}\right) \\
	& + \mathbb{P}\left(\left\Vert \breve {\mathcal{U}}_n - \tilde {\mathcal{U}}_n \right\Vert_2 \geq 3\delta_n \bigg|A^{(n)},S^{(n)},\Xi^{(n)}\right) \\
	& \leq C\left[\delta' + \mathbb P(||\mathcal G_2||_2 \geq 1/\delta')  + \mathbb{P}\left(t_1 - 3\delta_n 1_2 \leq \mathcal{U} \leq t_1 + 3\delta_n  1_2 \right) + 1\{F_n^c(\delta')\}\right] \\
	& +  C_0 2L_n \delta_n^{-3} \left(1 + \frac{|\log(1/(2L_n\delta_n^{-3})) |}{2} \right) \\
	& \leq C\left[\delta' + \mathbb P(||\mathcal G_2||_2 \geq 1/\delta')  + \delta_n + 1\{F_n^c(\delta')\}\right] +  C_0 2L_n \delta_n^{-3} \left(1 + \frac{|\log(1/(2L_n\delta_n^{-3})) |}{2} \right) \\
	& \convP 0,
	\end{align*}
	where the third inequality is by \eqref{eq:coupling} and fact that
	\begin{align*}
	& \left|\mathbb{P}\left(t_1 - 3\delta_n 1_2 \leq \breve {\mathcal{U}}_n \leq t_1 + 3\delta_n  1_2\bigg|A^{(n)},S^{(n)},\Xi^{(n)}\right) - \mathbb{P}\left(t_1 - 3\delta_n 1_2 \leq \mathcal{U} \leq t_1 + 3\delta_n 1_2 \right) \right| \\
	& \leq \left|\mathbb{P}\left(\breve {\mathcal{U}}_n \leq t_1 + 3\delta_n 1_2 \bigg|A^{(n)},S^{(n)},\Xi^{(n)}\right) - \mathbb{P}\left( \mathcal{U} \leq t_1 + 3\delta_n 1_2\bigg|A^{(n)},S^{(n)},\Xi^{(n)}\right) \right| \\
	& + \left|\mathbb{P}\left(\breve {\mathcal{U}}_n \leq t_1 - 3\delta_n 1_2 \bigg|A^{(n)},S^{(n)},\Xi^{(n)}\right) - \mathbb{P}\left( \mathcal{U} \leq t_1 - 3\delta_n 1_2 \right) \right| \\
	&  \leq C\delta' + \mathbb P(||\mathcal G_2||_2 \geq 1/\delta') + 2\cdot1\{F_n^c(\delta')\},
	\end{align*}
	and the fourth inequality is by \citet[Lemma A.1]{CCK14}. 
	
	\textbf{Step 2: The second term on the RHS of \eqref{eq:sufficient}}. 	We first define
	\begin{align*}
	\tilde 	V^*_{n} = \sum_{s \in \mathcal{S}} p_s \sqrt{n} \begin{pmatrix}
	\frac{1}{n p_s}\sum_{i = \lfloor n\mathbb{P}(S<s)\rfloor+1}^{\lfloor n\mathbb{P}(S\leq s)\rfloor} \left[\phi_i^s(1)- \mathbb{E}(\phi_i^s(1)) - \left(\phi_i^s(0)- \mathbb{E}(\phi_i^s(0))\right)\right] \\
	\frac{1}{n \pi_s p_s}\sum_{i =\lfloor n\mathbb{P}(S<s)\rfloor+1}^{\lfloor n(\mathbb{P}(S<s)+\pi_sp_s)\rfloor}(\phi_i^s(1)- \mathbb{E}(\phi_i^s(1))) \\
	-
	\frac{1}{n (1-\pi_s)p_s}\sum_{i=\lfloor n(\mathbb{P}(S<s)+\pi_sp_s)\rfloor+1}^{\lfloor n(\mathbb{P}(S\leq s))\rfloor}(\phi_i^s(0)- \mathbb{E}(\phi_i^s(0))) 
	\end{pmatrix},
	\end{align*}
	
	we note that	
	\begin{align*}
	N_s/n \convP \mathbb{P}(S<s), \quad n_{1,s}/n \convP \pi_sp_s, \quad \text{and} \quad n_{0,s}/n \convP (1-\pi_s)p_s.
	\end{align*}
	
	Because the partial sum process is stochastically equicontinuous, we have
	$\tilde {\mathcal{V}}_n = \tilde {\mathcal{V}}_n^* + o_P(1)$. By construction, we have $\tilde {\mathcal{V}}_n^* \indep (A^{(n)},X^{(n)})$, and by the Lindeberg CLT, $\tilde {\mathcal{V}}_n^* \convD V$. In particular, we see that the limit distribution of $\tilde {\mathcal{V}}_n^*$
	\begin{align*}
	& \N \left(\begin{pmatrix}
	0 \\
	0
	\end{pmatrix}, \sum_{s \in \mathcal{S}}p^2(s) \begin{pmatrix}
	\frac{var(\phi_i^s(1) - \phi_i^s(0))}{p_s} & \frac{var(\phi_i^s(1) - \phi_i^s(0))}{p_s}\\
	\frac{var(\phi_i^s(1) - \phi_i^s(0))}{p_s} & \frac{var(\phi_i^s(1))}{\pi_sp_s} + \frac{var(\phi_i^s(0))}{(1-\pi_s)p_s}
	\end{pmatrix}  \right) \\
	& \stackrel{d}{=} \N \left(\begin{pmatrix}
	0 \\
	0
	\end{pmatrix},  \begin{pmatrix}
	\mathbb{E}	var(\phi_i(1) - \phi_i(0)|S_i) & \mathbb{E}	var(\phi_i(1) - \phi_i(0)|S_i)\\
	\mathbb{E}	var(\phi_i(1) - \phi_i(0)|S_i) & \mathbb{E}\left[\frac{var(\phi_i(1)|S_i)}{\pi_{S_i}} + \frac{var(\phi_i(0)|S_i)}{(1-\pi_{S_i})}\right]
	\end{pmatrix}  \right).
	\end{align*}
	
	Therefore, we have
	\begin{align*}
	& \mathbb{P}\left(\tilde {\mathcal{V}}_n \leq t_2\bigg|A^{(n)},S^{(n)}\right) \\
	& \leq \mathbb{P}\left(\tilde {\mathcal{V}}_n^* \leq t_2+\delta \bigg|A^{(n)},S^{(n)}\right) + \mathbb{P}\left(|\tilde {\mathcal{V}}_n - \tilde {\mathcal{V}}_n^*| \geq \delta \bigg| A^{(n)},S^{(n)} \right) \\
	& = \mathbb{P}\left(\tilde {\mathcal{V}}_n^* \leq t_2+\delta \right) + \mathbb{P}\left(|\tilde {\mathcal{V}}_n - \tilde {\mathcal{V}}_n^*| \geq \delta \bigg|A^{(n)},S^{(n)} \right)\\
	& = \mathbb{P}\left(\tilde {\mathcal{V}}_n^* \leq t_2+\delta \right)  + o_P(1) \\
	& = \mathbb{P}\left(\mathcal V \leq t_2+\delta \right) + o_P(1),
	\end{align*}
	where the first equality is by the fact that $\tilde {\mathcal{V}}_n^*$ is independent of $(A^{(n)},S^{(n)})$ and the second equality holds because by Markov's inequality, for any $\delta'>0$,
	\begin{align*}
	\mathbb{P}\left(\mathbb{P}\left(|\tilde {\mathcal{V}}_n - \tilde {\mathcal{V}}_n^*| \geq \delta \bigg|A^{(n)},S^{(n)} \right) \geq \delta' \right) \leq \frac{\mathbb{P}\left(|\tilde {\mathcal{V}}_n - \tilde {\mathcal{V}}_n^*| \geq \delta \right)}{\delta'} \rightarrow 0.
	\end{align*}
	
	Similarly, we can show that 
	\begin{align*}
	& \mathbb{P}\left(\tilde {\mathcal{V}}_n > t_2\bigg|A^{(n)},S^{(n)}\right) \leq \mathbb{P}\left(\mathcal V > t_2- \delta \right) + o_P(1),
	\end{align*}
	or equivalently, 
	\begin{align*}
	\mathbb{P}\left(\tilde {\mathcal{V}}_n \leq t_2\bigg|A^{(n)},S^{(n)}\right) \geq \mathbb{P}\left(\mathcal V \leq t_2- \delta \right) + o_P(1).
	\end{align*}
	By letting $n \rightarrow \infty$ followed by $\delta \downarrow 0$, we have
	\begin{align*}
	\left|\mathbb{P}\left(\tilde {\mathcal{V}}_n \leq t_2\bigg|A^{(n)},S^{(n)}\right) - \mathbb{P}\left(\mathcal V \leq t_2\right)\right| \convP 0. 
	\end{align*}
	This concludes the proof.

\end{proof}

\begin{lem}
	Suppose Assumptions \ref{ass:assignment1}--\ref{ass:variance} hold. Then, we have 
	$$\eps_{\aleph_{a,s}}^\top(a) P_{a,s}\eps_{\aleph_{a,s}}(a) - \sum_{i \in \aleph_{a,s}}P_{a,s,i,i} Y_i \acute \eps_{a,s,i} = o_P(n).$$
	\label{lem:I1}
\end{lem}
\begin{proof}
	Let $H_{a,s,i,j} = M_{a,s,i,j} - (\sum_{j \in \aleph_{a,s}}M_{a,s,i,j}) \gamma_{a,s,\infty}^{-1} (\sum_{i \in \aleph_{a,s}}M_{a,s,i,j})$ and $H_{a,s}$ be the $n_{a,s} \times n_{a,s}$ matrix with its $(i,j)$th entry being $H_{a,s,i,j}$.  We can see that $\sum_{j \in \aleph_{a,s}}H_{a,s,i,j}H_{a,s,j,k} = H_{a,s,i,k}$, which means $H_{a,s}$ is idempotent. In addition, we have $\hat \eps_{a,s,i} = \sum_{j \in \aleph_{a,s}}H_{a,s,i,j}(e_{j,s}(a) + \eps_j(a))$ and $\acute \eps_{a,s,i} = \hat  \eps_{a,s,i}/M_{a,s,i,i}$. Further denote $\tilde H_{a,s,i,j} = H_{a,s,i,j}/M_{a,s,i,i}$ and $\mu_i(a) = \mathbb{E}(Y_i(a)|X_i,S_i)$.  This implies 
	\begin{align*}
	I_1 & = \sum_{i \in \aleph_{a,s}} \sum_{j \in \aleph_{a,s}, j \neq i} \eps_{i}(a)P_{a,s,i,j}\eps_{j}(a) + \sum_{i \in \aleph_{a,s}}P_{a,s,i,i}(\eps_{i}^2(a) - Y_i \Acute{\eps}_{a,s,i}) \\
	& =  \left[\sum_{i \in \aleph_{a,s}}P_{a,s,i,i}\eps_i^2(a) (1 - \tilde H_{a,s,i,i})\right] + \left[\sum_{i \in \aleph_{a,s}} \sum_{j \in \aleph_{a,s}, j \neq i} \eps_{i}(a)P_{a,s,i,j}\eps_{i}(a)\right] \\
	&- \left[\sum_{i \in \aleph_{a,s}} \sum_{j \in \aleph_{a,s}, j \neq i}P_{a,s,i,i} \tilde H_{a,s,i,j}\eps_i(a)\eps_j(a)\right]  - \left[\sum_{i \in \aleph_{a,s}} \sum_{j \in \aleph_{a,s}}P_{a,s,i,i} \tilde H_{a,s,i,j}\mu_i(a)e_{j,s}(a) \right] \\
	& - \left[ \sum_{i \in \aleph_{a,s}} \sum_{j \in \aleph_{a,s}}P_{a,s,i,i} \tilde H_{a,s,i,j}\mu_i(a)\eps_j(a) \right]-\left[\sum_{i \in \aleph_{a,s}} \sum_{j \in \aleph_{a,s}}P_{a,s,i,i} \tilde H_{a,s,i,j}\eps_i(a)e_{j,s}(a) \right] \\
	& \equiv I_{1,1} + I_{1,2} - I_{1,3} - I_{1,4} - I_{1,5} - I_{1,6}.
	\end{align*}
	Conditional on $(A^{(n)},S^{(n)},X^{(n)})$, $\aleph_{a,s}$, $P_{a,s}$, $\tilde H_{a,s}$, $\phi_i(a)$, and $e_{i,s}(a)$ are all deterministic and $\{\eps_i(a)\}_{i \in \aleph_{a,s}}$ are independent across $i$, conditionally mean zero, and $\eps_i(a)|(A^{(n)},S^{(n)},X^{(n)}) \stackrel{d}{=} \eps_i(a)|X_i,S_i=s$. For $I_{1,1}$, we have
	\begin{align*}
	|  I_{1,1}| & = \sum_{i \in \aleph_{a,s}}P_{a,s,i,i}\eps_i^2(a) \frac{(\sum_{j \in \aleph_{a,s}}M_{a,s,i,j})^2}{1_{\aleph_{a,s}}^\top M_{a,s}1_{\aleph_{a,s}} M_{a,s,i,i}} \\
	& \leq \left(\sum_{i \in \aleph_{a,s}}P_{a,s,i,i}\eps_i^2(a)\right) \left( \frac{\max_{i \in \aleph_{a,s}}(\sum_{j \in \aleph_{a,s}}M_{a,s,i,j})^2}{n}\right) \left(\frac{n}{n_{a,s} \gamma_{a,s,n} \min_{i \in \aleph_{a,s}} M_{a,s,i,i}}\right) = o_P(n)
	\end{align*}
	where we use Assumptions \ref{ass:linear}, \ref{ass:omega}, \ref{ass:variance}, and the facts that $\min_{i \in \aleph_{a,s}}M_{a,s,i,i} \geq \delta>0$ and 
	\begin{align*}
	& \left(\sum_{i \in \aleph_{a,s}}P_{a,s,i,i}\eps_i^2(a)\right) = O_P(n).
	\end{align*}
	
	For $I_{1,2}$, we have
	\begin{align*}
	\mathbb E (I_{1,2}|A^{(n)},S^{(n)},X^{(n)}) = 0 
	\end{align*}
	and 
	\begin{align*}
	var(I_{1,2}|A^{(n)},S^{(n)},X^{(n)}) & =  \sum_{i \in \aleph_{a,s}} \sum_{j \in \aleph_{a,s}, j \neq i} \mathbb E(\eps^2_{i}(a)|A^{(n)},S^{(n)},X^{(n)}) P_{a,s,i,j}^2\mathbb E(\eps^2_{j}(a)|A^{(n)},S^{(n)},X^{(n)}) \\
	& \leq \sum_{i \in \aleph_{a,s}} \sum_{j \in \aleph_{a,s}} P_{a,s,i,j}^2 \max_{i \in [n]}\mathbb{E}(\eps^2_i(a)|X_i,S_i) = O_P(n),
	\end{align*}
	which implies 
	\begin{align*}
	I_{1,2} = O_P(n^{1/2}) = o_P(n).
	\end{align*}
	
	Similarly, we have  
	\begin{align*}
	\mathbb E (I_{1,3}|A^{(n)},S^{(n)},X^{(n)}) = 0 
	\end{align*}
	and 
	\begin{align*}
	var(I_{1,3}|A^{(n)},S^{(n)},X^{(n)}) & \lesssim \sum_{i \in \aleph_{a,s}} \sum_{j \in \aleph_{a,s}} \tilde H_{a,s,i,j}^2P_{a,s,i,i}^2 \max_{i \in [n]}\mathbb{E}(\eps^2_i(a)|X_i,S_i) \\
	& \lesssim \sum_{i \in \aleph_{a,s}} \sum_{j \in \aleph_{a,s}} H_{a,s,i,j}^2 (\min_{i \in \aleph_{a,s}}M_{a,s,i,i})^{-2}\max_{i \in [n]}\mathbb{E}(\eps^2_i(a)|X_i,S_i)\\
	& = O_P(n),
	\end{align*}
	where the second inequality holds by the fact that $P_{a,s,i,i}^2 \leq 1$ and the last equality holds because $H_{a,s}$ is idempotent. This implies $I_{1,3} = o_P(n)$. 
	
	For $I_{1,4}$, we have
	\begin{align*}
	I_{1,4}^2 & \leq \left(\sum_{i \in \aleph_{a,s}} P_{a,s,i,i}^2\mu_i^2(a) M_{a,s,i,i}^{-2}\right) \left( \sum_{i \in \aleph_{a,s}} \left(\sum_{j \in \aleph_{a,s}}H_{a,s,i,j}e_{j,s}(a)\right)^2\right) \\
	& \leq \left(\sum_{i \in \aleph_{a,s}} P_{a,s,i,i}^2 M_{a,s,i,i}^{-2}\mu_i^2(a) \right)  \left( \sum_{j' \in \aleph_{a,s}} \sum_{j \in \aleph_{a,s}}H_{a,s,j',j}e_{j',s}(a)e_{j,s}(a)\right) \\
	& \leq \left(\max_{i \in \aleph_{a,s}}P_{a,s,i,i}^2 M_{a,s,i,i}^{-2} \right)\left(\sum_{i \in \aleph_{a,s}} \mu_i^2(a)\right)  \left(\sum_{j \in \aleph_{a,s}}e_{j,s}^2(a)\right) =o_P(n),
	\end{align*}
	where the second inequality is because $H_{a,s}$ is idempotent , the last inequality is by $||H_{a,s}||_{op} \leq 1$, and the last equality is by the facts that
	\begin{align*}
	& \left(\sum_{i \in \aleph_{a,s}} \mu_i^2(a)\right) = O_P(n), \quad \max_{i \in \aleph_{a,s}}P_{a,s,i,i}^2 M_{a,s,i,i}^{-2} = O_P(1), \quad   \left(\sum_{j \in \aleph_{a,s}}e_{j,s}^2(a)\right) = o_P(1).
	\end{align*}
	This implies $I_{1,4} = o_P(n^{1/2}) = o_P(n)$. 
	
	For $I_{1,5}$, we have $ \mathbb E (I_{1,5}|A^{(n)},S^{(n)},X^{(n)}) = 0$ and 
	\begin{align*}
	& var(I_{1,5}|A^{(n)},S^{(n)},X^{(n)}) \\
	& =   \sum_{j \in \aleph_{a,s}}  \left[ \sum_{i \in \aleph_{a,s}} P_{a,s,i,i} \tilde H_{a,s,i,j} \mu_i(a)\right]^2  \mathbb E(\eps^2_{j}(a)|A^{(n)},S^{(n)},X^{(n)}) \\
	& \leq \sum_{j \in \aleph_{a,s}}\sum_{i \in \aleph_{a,s}} \frac{P_{a,s,i,i} \mu_i(a)}{M_{a,s,i,i}} H_{a,s,i,j} \frac{P_{a,s,j,j} \mu_j(a)}{M_{a,s,j,j}} \left[\max_{j \in \aleph_{a,s}}\mathbb E(\eps^2_{j}(a)|A^{(n)},S^{(n)},X^{(n)}) \right] \\
	& \leq \sum_{i \in \aleph_{a,s}} \frac{P_{a,s,i,i}^2 \mu_i^2(a)}{M_{a,s,i,i}^2} \left[\max_{j \in \aleph_{a,s}}\mathbb E(\eps^2_{j}(a)|A^{(n)},S^{(n)},X^{(n)}) \right] \\
	& \leq \left(\sum_{i \in \aleph_{a,s}} \mu^2_i(a)\right)\left(\max_{i \in \aleph_{a,s}}\frac{P_{a,s,i,i}^2 }{M_{a,s,i,i}^2}\right) \left(\max_{j \in \aleph_{a,s}}\mathbb E(\eps^2_{j}(a)|A^{(n)},S^{(n)},X^{(n)}) \right) \\
	& = O_P(n),
	\end{align*}
	where the second inequality is by $||H_{a,s}||_{op} \leq 1$. This implies $I_{1,5} = O_P(n^{1/2}) = o_P(n)$. In the same manner, we can show that $I_{1,6} = o_P(n)$, which concludes the proof of this lemma. 
\end{proof}

\begin{lem}\label{lem:beta}
	Suppose Assumptions in Theorem \ref{thm:fixed_k} hold. Then, we have
	\begin{align*}
	||\hat \beta_{a,s} -  \beta_{a,s}^*||_2 =  O_P\left(\sqrt{\frac{k_n \xi_n^2 \log(k_n)}{n} } \right).
	\end{align*}
\end{lem}
\begin{proof}
	Denote $\overline X_{a,s} = \frac{1}{n_{a,s}}\sum_{i \in \aleph_{a,s}} X_i$ and $\overline Y_{a,s} = \frac{1}{n_{a,s}}\sum_{i \in \aleph_{a,s}} Y_i$
	\begin{align*}
	& \hat \beta_{a,s} =   \left(\frac{1}{n_{a,s}}\sum_{i \in \aleph_{a,s}}(X_i-\overline X_{a,s})(X_i-\overline X_{a,s})^\top \right)^{-1}\left(\frac{1}{n_{a,s}}\sum_{i \in \aleph_{a,s}}(X_i-\overline X_{a,s})(Y_i-\overline Y_{a,s}) \right), \\
	& \beta_{a,s}^* =    Var(X_i|S_i=s)^{-1}Cov(X_i,Y_i|S_i=s).
	\end{align*}
	We first define 
	\begin{align*}
	\tilde \beta_{a,s} & =   \left(\frac{1}{n_{a,s}}\sum_{i \in \aleph_{a,s}}(X_i-\mathbb E(X_i|S_i=s))(X_i-\mathbb E(X_i|S_i=s))^\top \right)^{-1} \\
	& \times \left(\frac{1}{n_{a,s}}\sum_{i \in \aleph_{a,s}}(X_i-\mathbb E(X_i|S_i=s))(Y_i-\mathbb E(Y_i|S_i=s)) \right).
	\end{align*}
	
	\textbf{We first bound $||\hat \beta_{a,s} - \tilde \beta_{a,s}||_2$.} Let $S^{k_n}$ be the unit sphere in $\Re^{k_n}$, $\mathbb P_{n,a,s}f_i = \frac{1}{n_{a,s}}\sum_{i \in \aleph_{a,s}}f_i$ for any function of observations $f_i$, $\mathbb P_{a,s}f_i = \mathbb E(f_i|A_i=a,S_i=s) = \mathbb E(f_i|S_i=s)$, $||\cdot||_{\mathbb P_{n,a,s},2}$ be the $L_2$ norm w.r.t. the probability measure $\mathbb P_{n,a,s}$. Note that 
	\begin{align}\label{eq:beta1}
	& \left\Vert \frac{1}{n_{a,s}}\sum_{i \in \aleph_{a,s}}\left[ (X_i-\overline X_{a,s})(X_i-\overline X_{a,s})^\top  - (X_i-\mathbb E(X_i|S_i=s))(X_i-\mathbb E(X_i|S_i=s))^\top \right]\right\Vert_{op} \notag \\
	& = \sup_{\lambda \in S^{k_n}} \left| \left\Vert (X_i-\overline X_{a,s})^\top \lambda\right\Vert_{\mathbb P_{n,a,s},2}^2-  \left\Vert (X_i-\mathbb E(X_i|S_i=s))^\top \lambda\right\Vert_{\mathbb P_{n,a,s},2}^2\right| \notag \\
	& \leq \sup_{\lambda \in S^{k_n}} \left| \left\Vert (X_i-\overline X_{a,s})^\top \lambda\right\Vert_{\mathbb P_{n,a,s},2}-  \left\Vert (X_i-\mathbb E(X_i|S_i=s))^\top \lambda\right\Vert_{\mathbb P_{n,a,s},2}\right| \notag \\
	& \times \sup_{\lambda \in S^{k_n}}\left|\left\Vert (X_i-\overline X_{a,s})^\top \lambda\right\Vert_{\mathbb P_{n,a,s},2} + \left\Vert (X_i-\mathbb E(X_i|S_i=s))^\top \lambda\right\Vert_{\mathbb P_{n,a,s},2}\right| \notag \\
	&  \leq  \left\Vert \overline X_{a,s} - \mathbb E(X_i|S_i=s))\right\Vert_2 \sup_{\lambda \in S^{k_n}}\left|\left\Vert (X_i-\overline X_{a,s})^\top \lambda\right\Vert_{\mathbb P_{n,a,s},2} + \left\Vert (X_i-\mathbb E(X_i|S_i=s))^\top \lambda\right\Vert_{\mathbb P_{n,a,s},2}\right| \notag \\
	& \leq O_P(\sqrt{k_n/n}) \left( 1 + O_P(\sqrt{k_n/n})\right) = O_P(\sqrt{k_n/n}),
	\end{align}
	where we use the fact that 
	\begin{align*}
	\sup_{\lambda \in S^{k_n}}     \left\Vert (X_i-\mathbb E(X_i|S_i=s))^\top \lambda\right\Vert_{\mathbb P_{n,a,s},2}^2 & = \left\Vert \frac{1}{n_{a,s}}\sum_{i \in \aleph_{a,s}}\left[  (X_i-\mathbb E(X_i|S_i=s))(X_i-\mathbb E(X_i|S_i=s))^\top \right]\right\Vert_{op} \\
	& = O_P(1). 
	\end{align*}
	
	Similarly, we can show that 
	\begin{align}\label{eq:beta2}
	\left\Vert \frac{1}{n_{a,s}}\sum_{i \in \aleph_{a,s}}\left[(X_i-\mathbb E(X_i|S_i=s))(Y_i-\mathbb E(Y_i|S_i=s))-(X_i-\overline X_{a,s})(Y_i-\overline Y_{a,s})\right] \right\Vert_2 =  O_P(\sqrt{k_n/n}), 
	\end{align}
	implying that 
	\begin{align*}
	||\hat \beta_{a,s} - \tilde \beta_{a,s}||_2 & \leq  \left\Vert \frac{1}{n_{a,s}}\sum_{i \in \aleph_{a,s}}\left[ (X_i-\overline X_{a,s})(X_i-\overline X_{a,s})^\top  - (X_i-\mathbb E(X_i|S_i=s))(X_i-\mathbb E(X_i|S_i=s))^\top \right]\right\Vert_{op} \\
	& \times \left\Vert \frac{1}{n_{a,s}}\sum_{i \in \aleph_{a,s}}(X_i-\mathbb E(X_i|S_i=s))(Y_i-\mathbb E(Y_i|S_i=s))\right\Vert_2 \\
	& + \left[\lambda_{\min}\left( \frac{1}{n_{a,s}}\sum_{i \in \aleph_{a,s}}(X_i-\overline X_{a,s})(X_i-\overline X_{a,s})^\top \right)\right]^{-1}  \\
	& \times \left\Vert \frac{1}{n_{a,s}}\sum_{i \in \aleph_{a,s}}\left[(X_i-\mathbb E(X_i|S_i=s))(Y_i-\mathbb E(Y_i|S_i=s))-(X_i-\overline X_{a,s})(Y_i-\overline Y_{a,s})\right] \right\Vert_2\\
	& = O_P(\sqrt{k_n/n}), 
	\end{align*}
	where the last equality holds by \eqref{eq:beta1}, \eqref{eq:beta2}, \eqref{eq:beta4} below and the fact that  
	\begin{align*}
	\left\Vert \frac{1}{n_{a,s}}\sum_{i \in \aleph_{a,s}}(X_i-\mathbb E(X_i|S_i=s))(Y_i-\mathbb E(Y_i|S_i=s))\right\Vert_2 = \left\Vert Cov(X_i,Y_i|S_i=s)\right\Vert_2 + o_P(1) = O_P(1),
	\end{align*}

	\textbf{Next, we bound $||\tilde \beta_{a,s} -  \beta_{a,s}^*||_2$.} We note that 
	\begin{align*}
	&     \left\Vert \frac{1}{n_{a,s}}\sum_{i \in \aleph_{a,s}} (X_i-\mathbb E(X_i|S_i=s))(X_i-\mathbb E(X_i|S_i=s))^\top - Var(X_i|S_i=s)\right\Vert_{op} \\
	& = \sup_{\lambda \in S^{k_n}}     \left\vert (\mathbb P_{n,a,s} - \mathbb P_{a,s})((X_i-\mathbb E(X_i|S_i=s))^\top \lambda)^2 \right\vert.
	\end{align*}
	Further denote $M = \max_{i \in \aleph_{a,s}}\sup_{\lambda \in S^{k_n}}|(X_i-\mathbb E(X_i|S_i=s))^\top \lambda| =\max_{i \in \aleph_{a,s}}||X_i-\mathbb E(X_i|S_i=s)||_2$. 
	
	Conditioning on $(S_i,A_i)_{i \in [n]}$, we can treat $\{X_i\}_{i \in \aleph_{a,s}}$ as i.i.d. with distribution $\mathbb P_{a,s}$ which is just the conditional distribution of $ X_i$ given $S_i=s$. 
	

	In addition, by \citet[Lemma 6.2]{belloni2015}, we have
	\begin{align}\label{eq:beta3}
	& \left\Vert \frac{1}{n_{a,s}}\sum_{i \in \aleph_{a,s}} (X_i-\mathbb E(X_i|S_i=s))(X_i-\mathbb E(X_i|S_i=s))^\top - Var(X_i|S_i=s)\right\Vert_{op} \notag \\
	& = O_P\left(\sqrt{\frac{k_n \xi_n^2 \log(k_n)}{n} } + \frac{k_n \xi_n^2 \log(k_n)}{n} \right) 
	\end{align}
	and 
	\begin{align}\label{eq:beta4}
	& \left\Vert \frac{1}{n_{a,s}}\sum_{i \in \aleph_{a,s}} (X_i-\mathbb E(X_i|S_i=s))(Y_i-\mathbb E(Y_i|S_i=s)) - Cov(X_i,Y_i|S_i=s)\right\Vert_{2} \notag \\
	& = O_P\left(\sqrt{\frac{k_n \xi_n^2 \log(k_n)}{n} } + \frac{k_n \xi_n^2 \log(k_n)}{n} \right),
	\end{align}

	which implies 
	\begin{align*}
	||\tilde \beta_{a,s} -  \beta_{a,s}^*||_2 =  O_P\left(\sqrt{\frac{k_n \xi_n^2 \log(k_n)}{n} } \right),
	\end{align*}
	and thus, 
	\begin{align*}
	||\hat \beta_{a,s} -  \beta_{a,s}^*||_2 =  O_P\left(\sqrt{\frac{k_n \xi_n^2 \log(k_n)}{n} } \right).
	\end{align*}
\end{proof}

\section{Additional Simulation Results}\label{sec:add_sim}

\subsection{Additional Figures for Rejection Probability}
Figures \ref{fig:srs400}-\ref{fig:wei800} show how the rejection rates of $\hat{\tau}^{adj}$, $\hat{\tau}^{\ast}$, YYS and LTM vary with the number of regressors under the null and alternative hypotheses under the SRS, BCD, and WEI randomization schemes respectively.

\begin{figure}[H]
	\centering
	\subfigure[$H_{0}:\mu_1-\mu_0=0$]{
		\includegraphics[width=\textwidth]{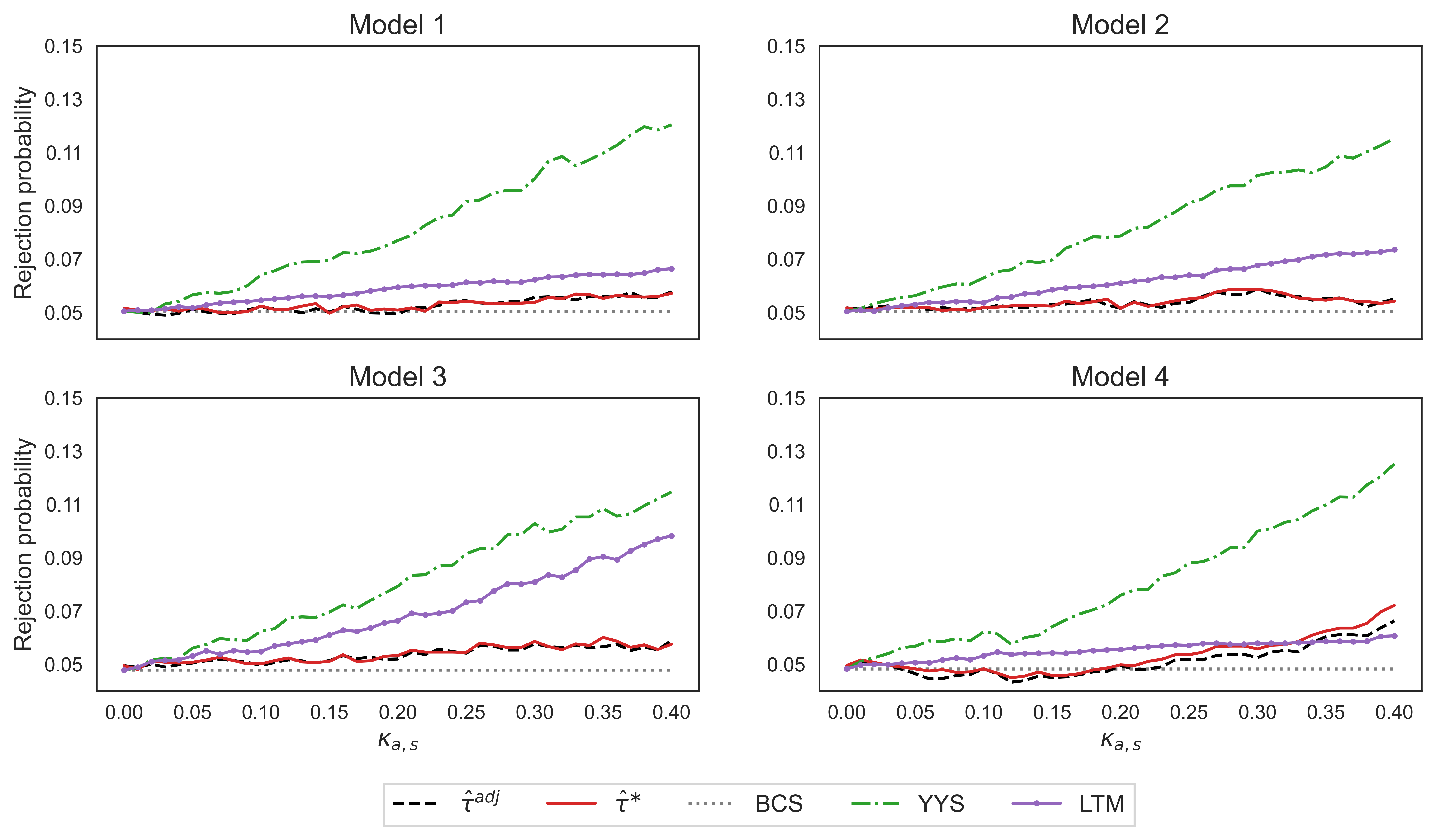}
	}
	
	\subfigure[$H_{1}:\mu_1-\mu_0=0.2$]{
		\includegraphics[width=\textwidth]{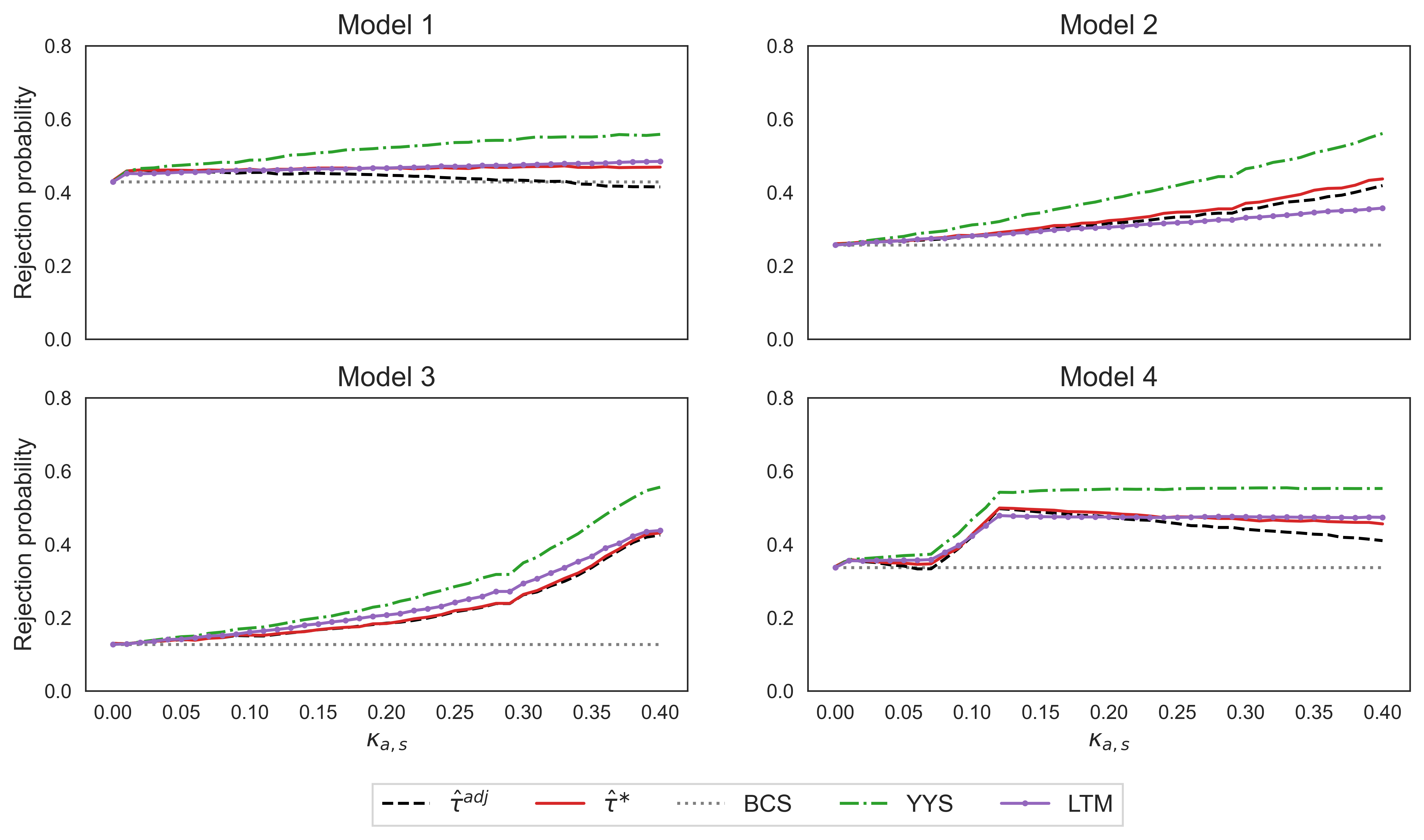}
	}
	\caption{Rejection probabilities under SRS when $n=400$}\label{fig:srs400}
\end{figure}

\begin{figure}[H]
	\centering
	\subfigure[$H_{0}:\mu_1-\mu_0=0$]{
		\includegraphics[width=\textwidth]{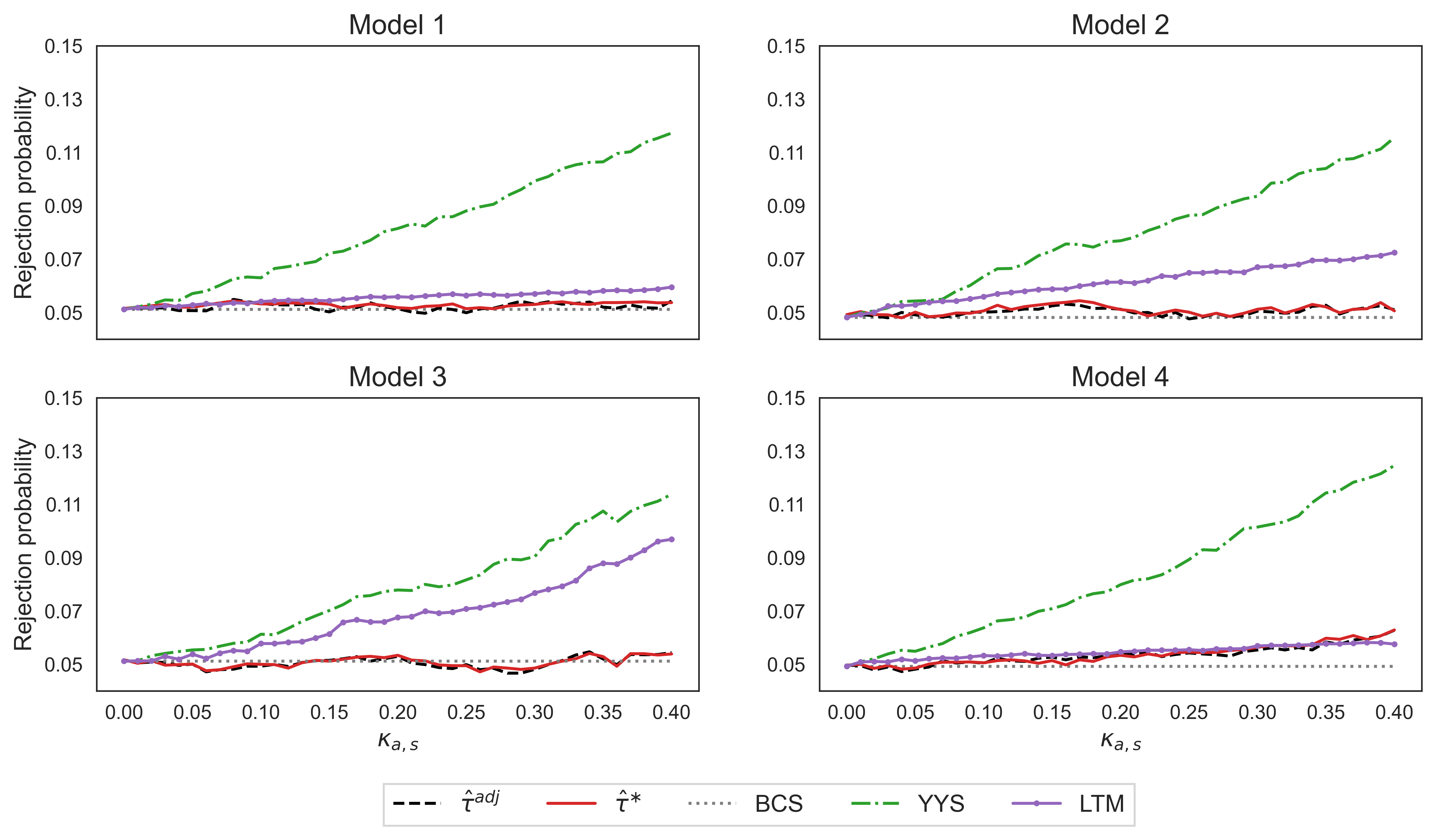}
	}
	
	\subfigure[$H_{1}:\mu_1-\mu_0=0.2$]{
		\includegraphics[width=\textwidth]{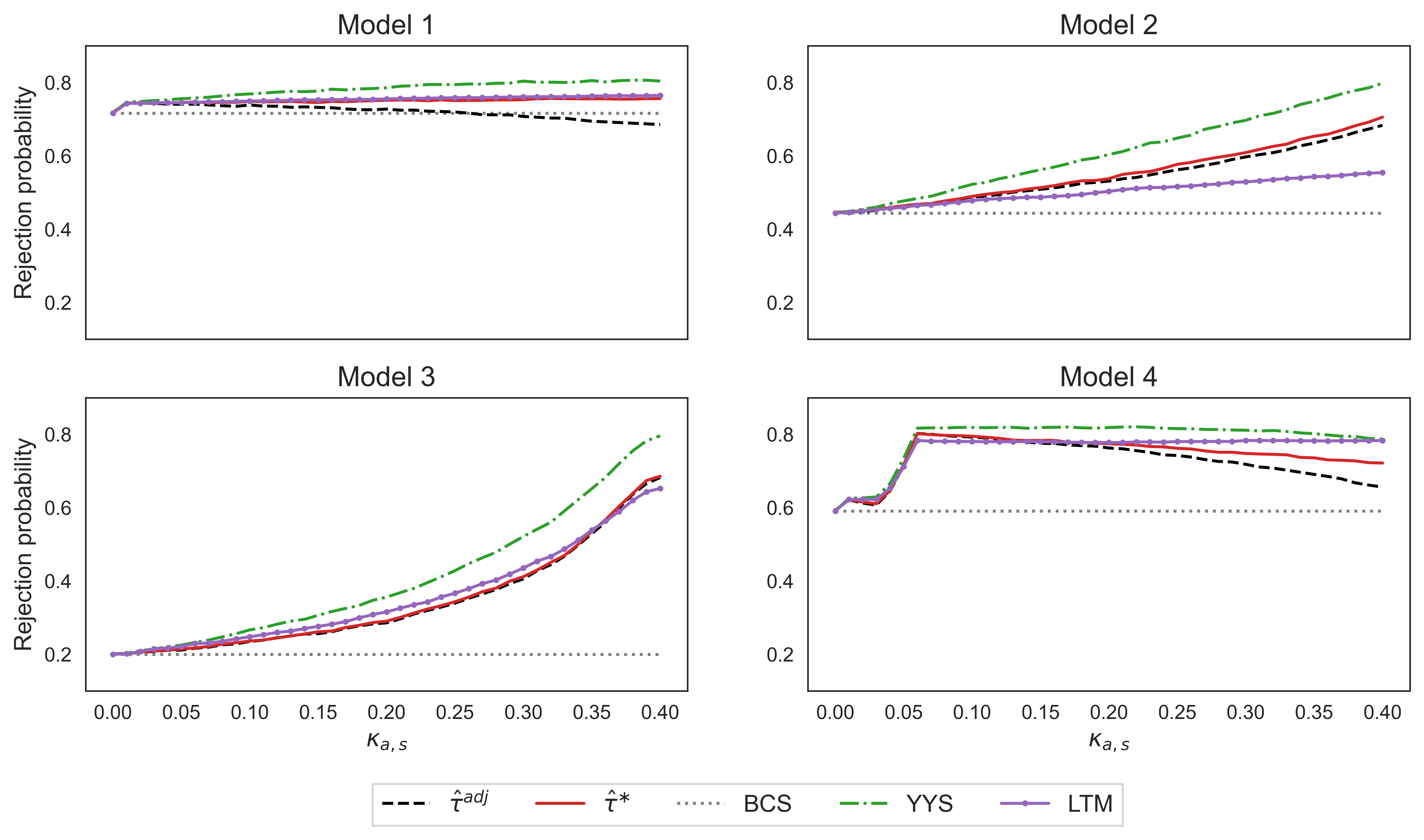}
	}
	\caption{Rejection probabilities under SRS when $n=800$}\label{fig:srs800}
\end{figure}

\begin{figure}[H]
	\centering
	\subfigure[$H_{0}:\mu_1-\mu_0=0$]{
		\includegraphics[width=\textwidth]{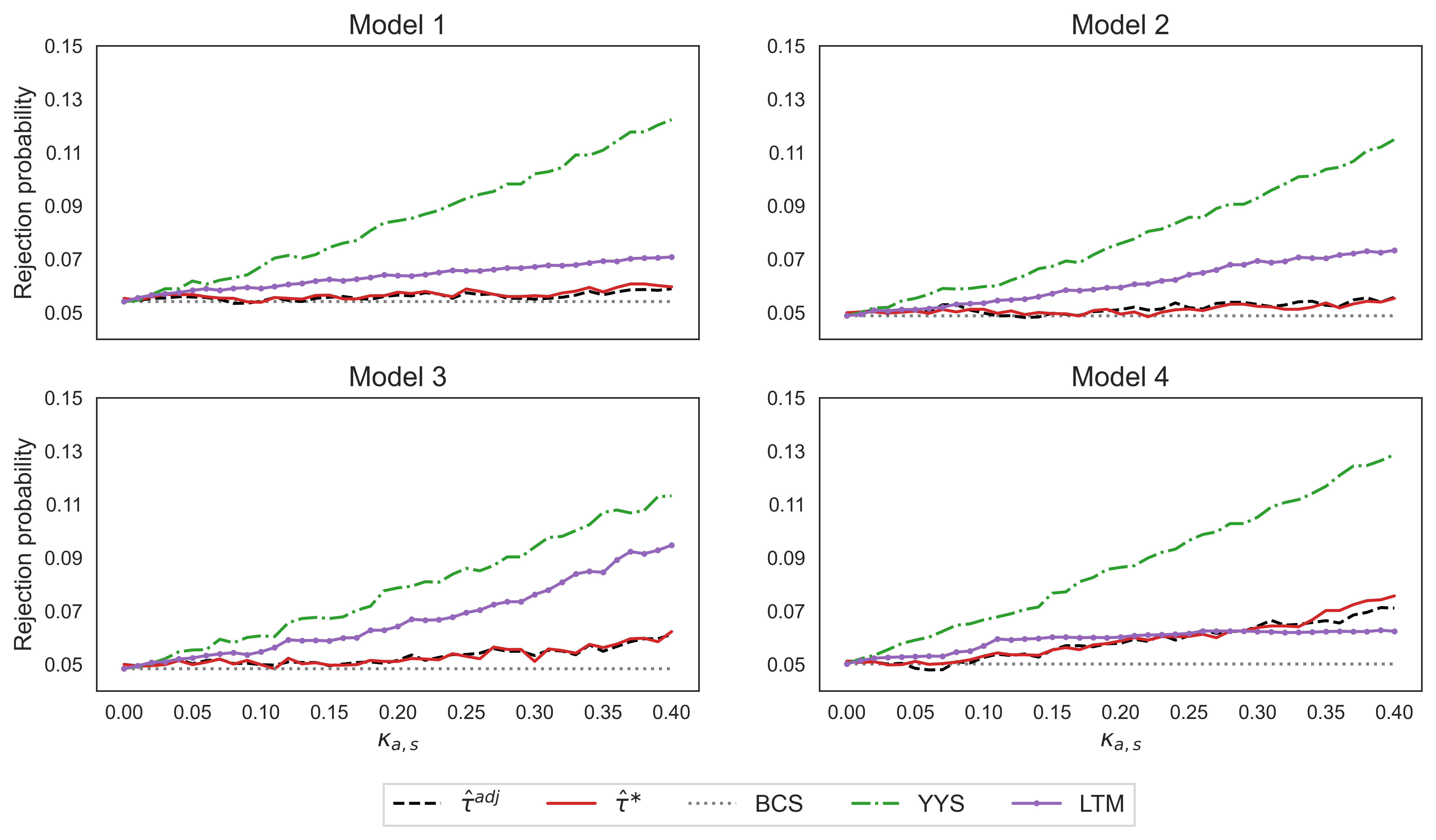}
	}
	
	\subfigure[$H_{1}:\mu_1-\mu_0=0.2$]{
		\includegraphics[width=\textwidth]{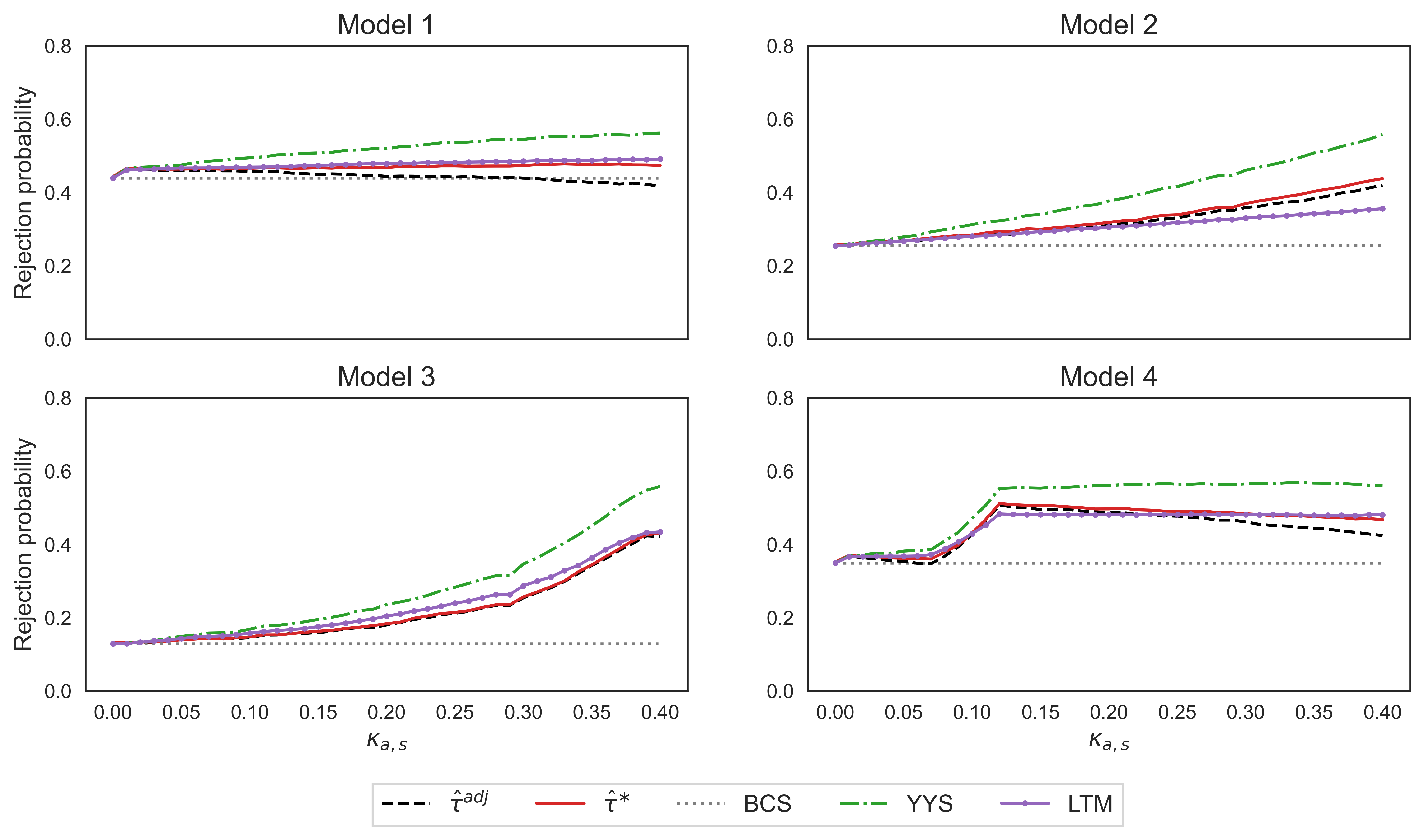}
	}
	\caption{Rejection probabilities under BCD when $n=400$}\label{fig:bcd400}
\end{figure}

\begin{figure}[H]
	\centering
	\subfigure[$H_{0}:\mu_1-\mu_0=0$]{
		\includegraphics[width=\textwidth]{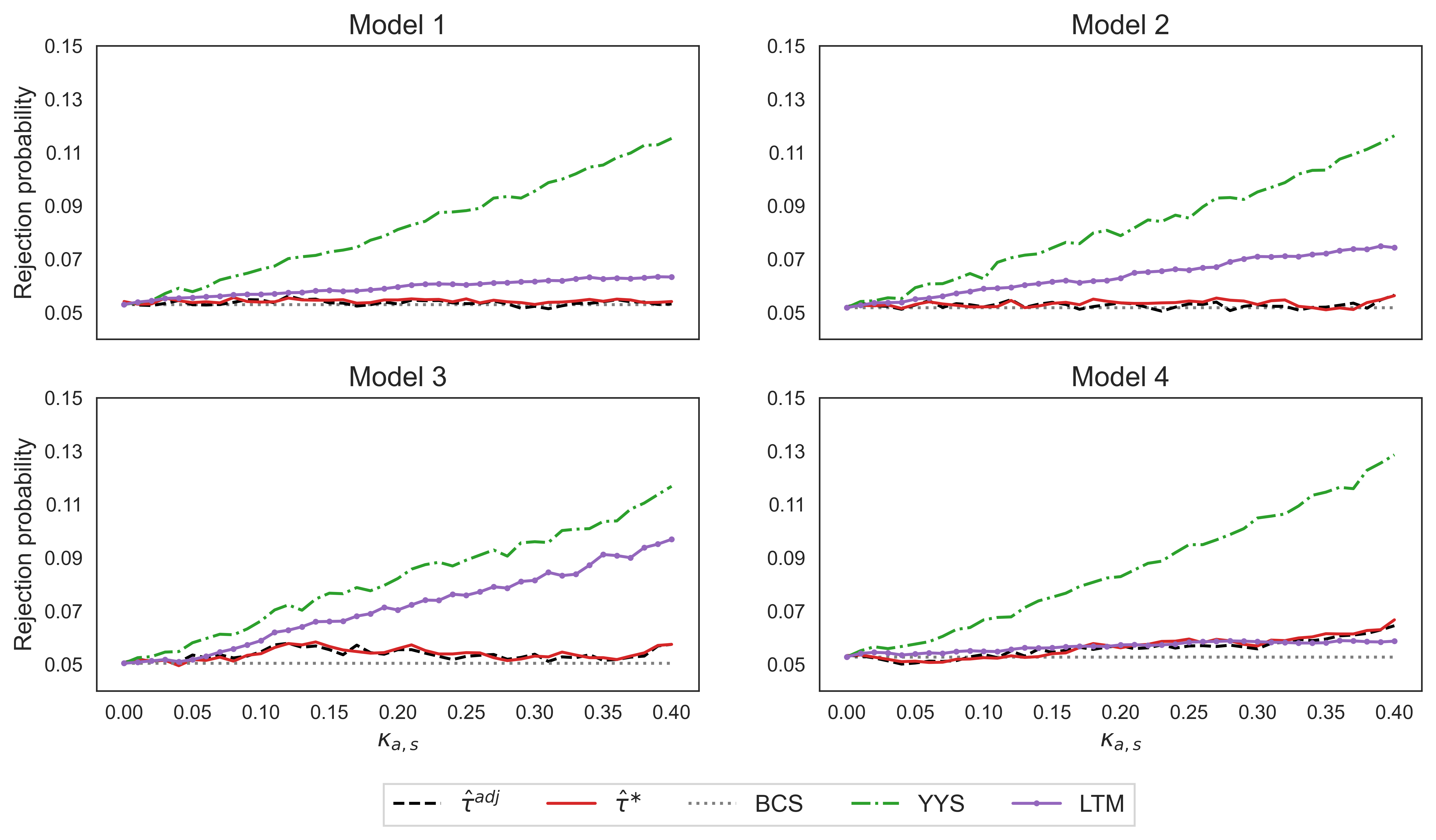}
	}
	
	\subfigure[$H_{1}:\mu_1-\mu_0=0.2$]{
		\includegraphics[width=\textwidth]{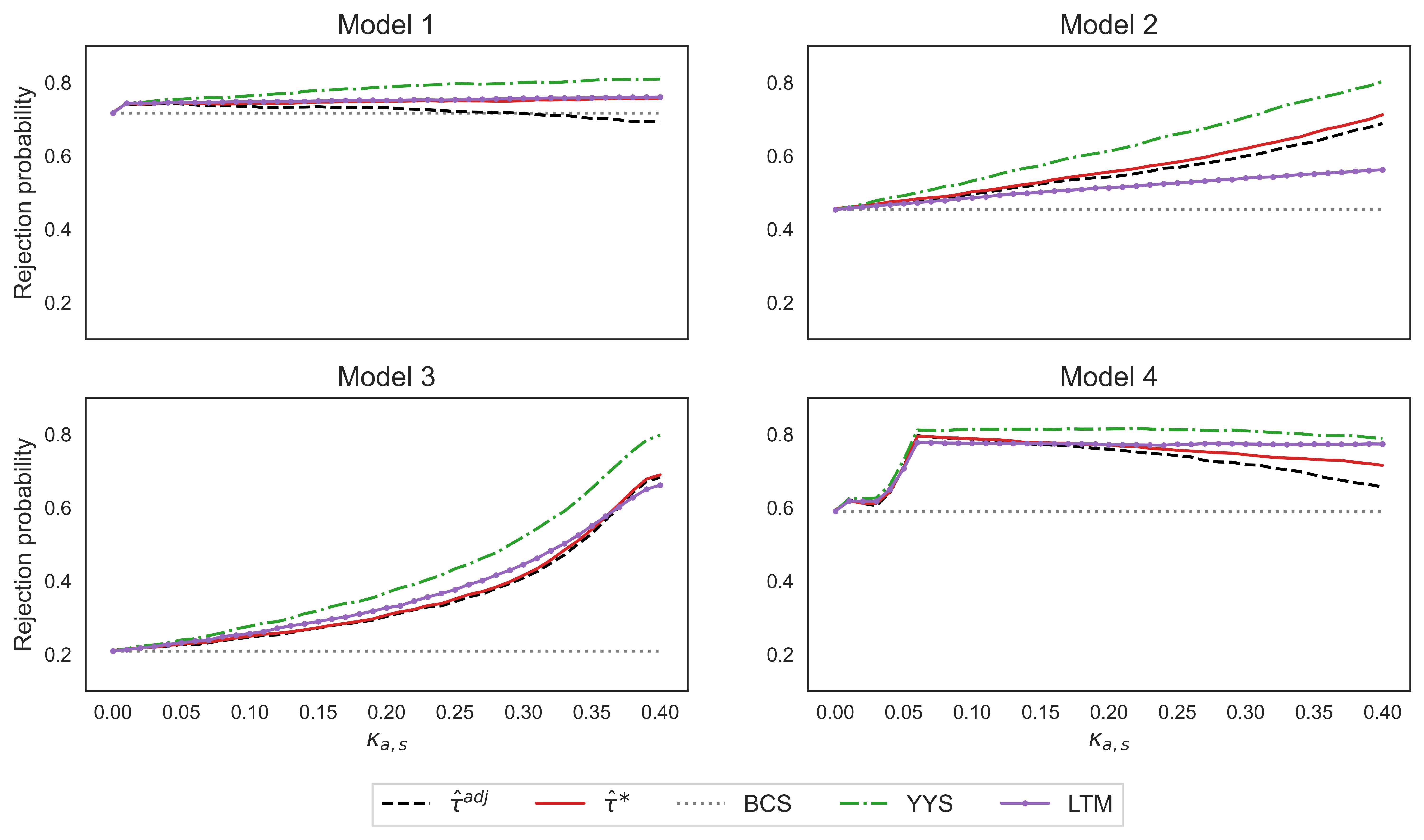}
	}
	\caption{Rejection probabilities under BCD when $n=800$}\label{fig:bcd800}
\end{figure}

\begin{figure}[H]
	\centering
	\subfigure[$H_{0}:\mu_1-\mu_0=0$]{
		\includegraphics[width=\textwidth]{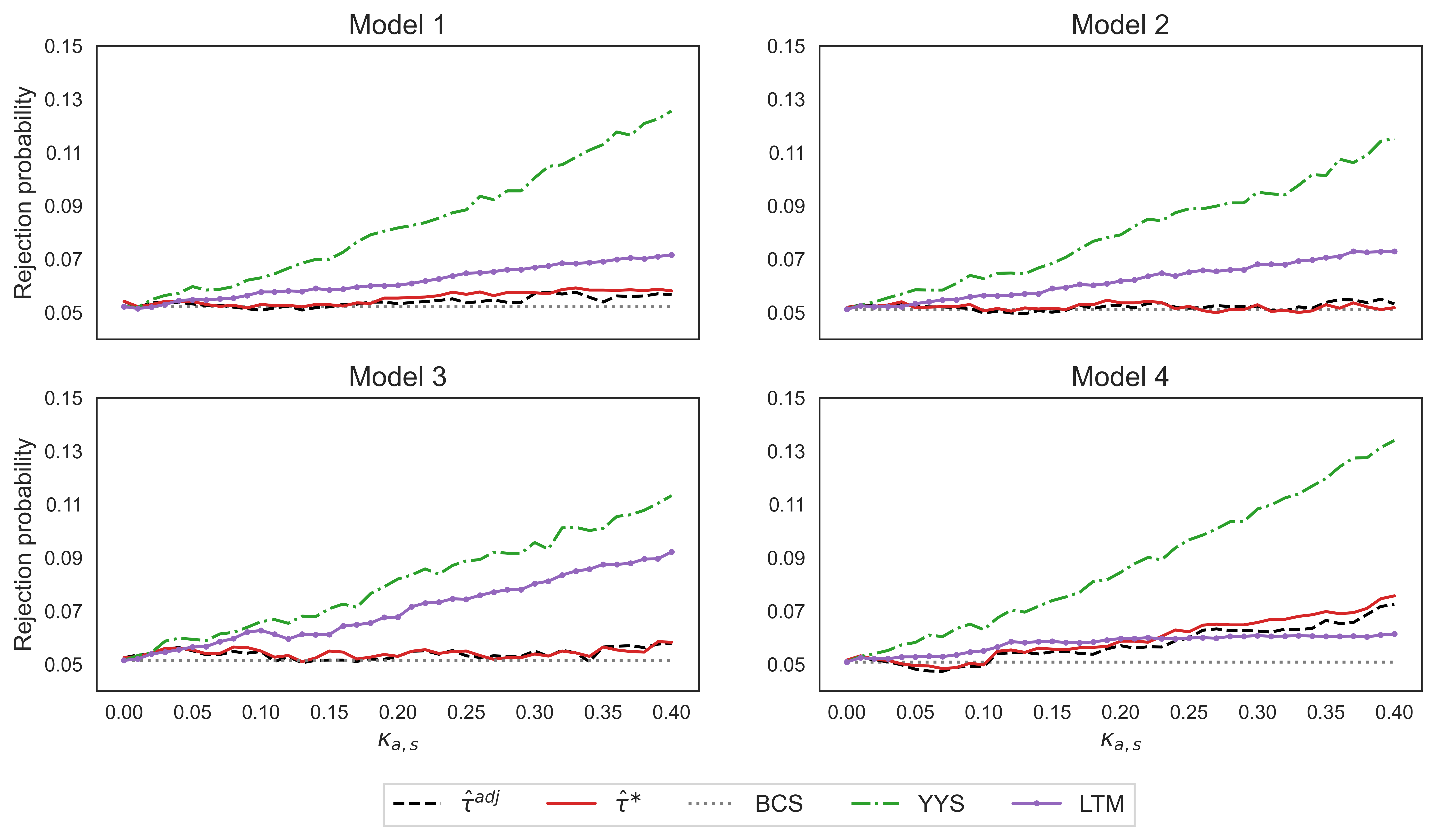}
	}
	
	\subfigure[$H_{1}:\mu_1-\mu_0=0.2$]{
		\includegraphics[width=\textwidth]{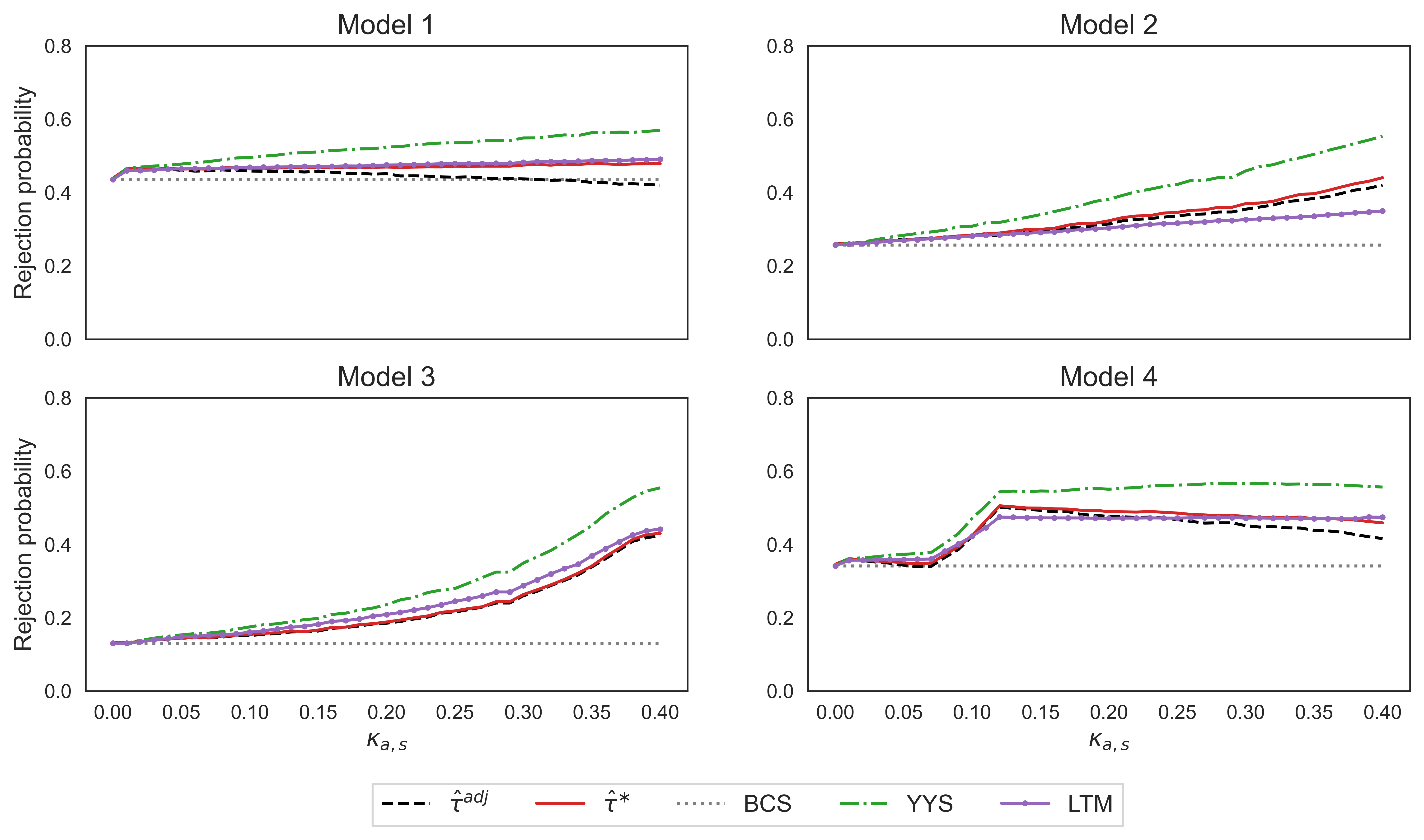}
	}
	\caption{Rejection probabilities under WEI when $n=400$}\label{fig:wei400}
\end{figure}

\begin{figure}[H]
	\centering
	\subfigure[$H_{0}:\mu_1-\mu_0=0$]{
		\includegraphics[width=\textwidth]{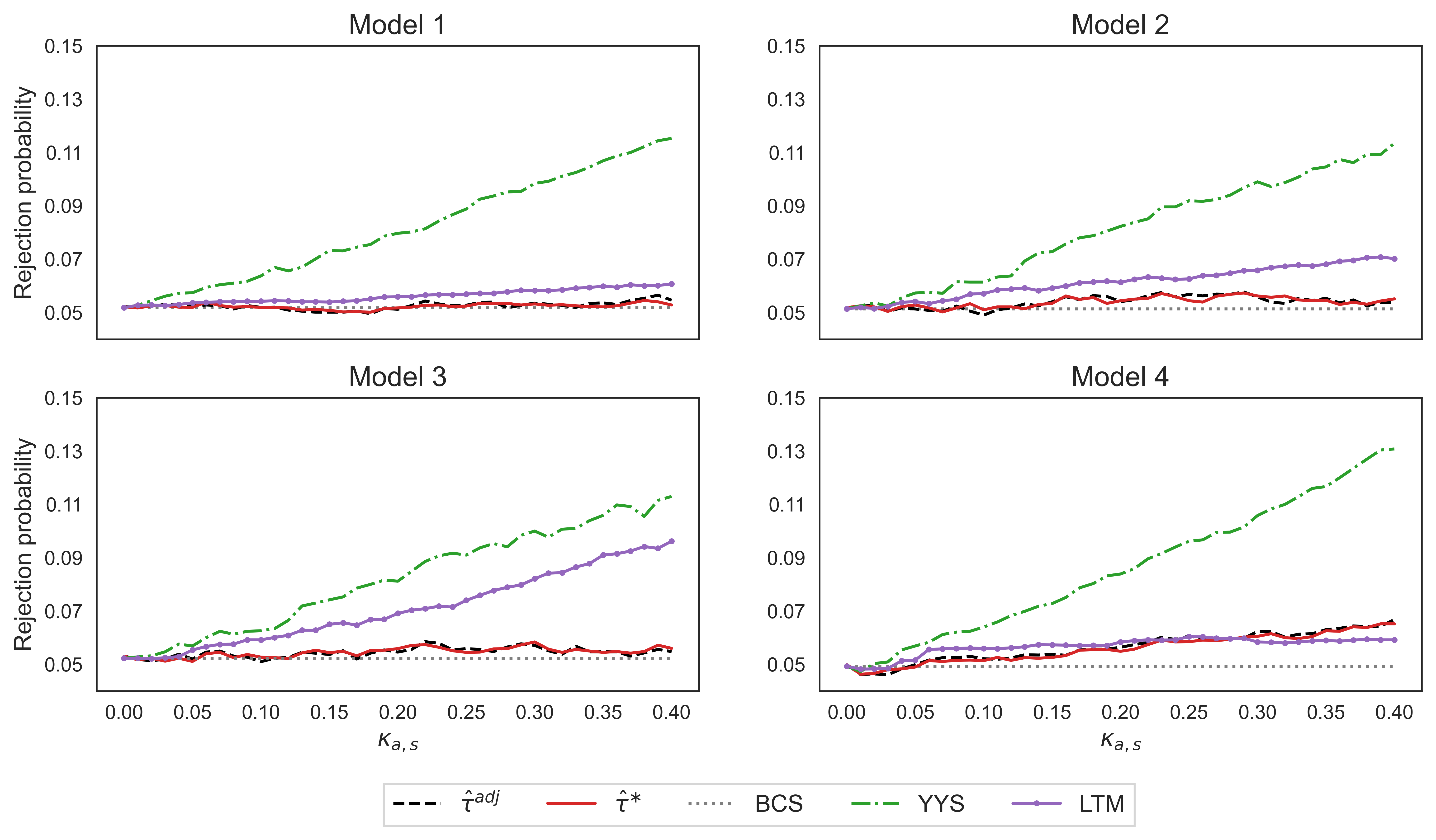}
	}
	
	\subfigure[$H_{1}:\mu_1-\mu_0=0.2$]{
		\includegraphics[width=\textwidth]{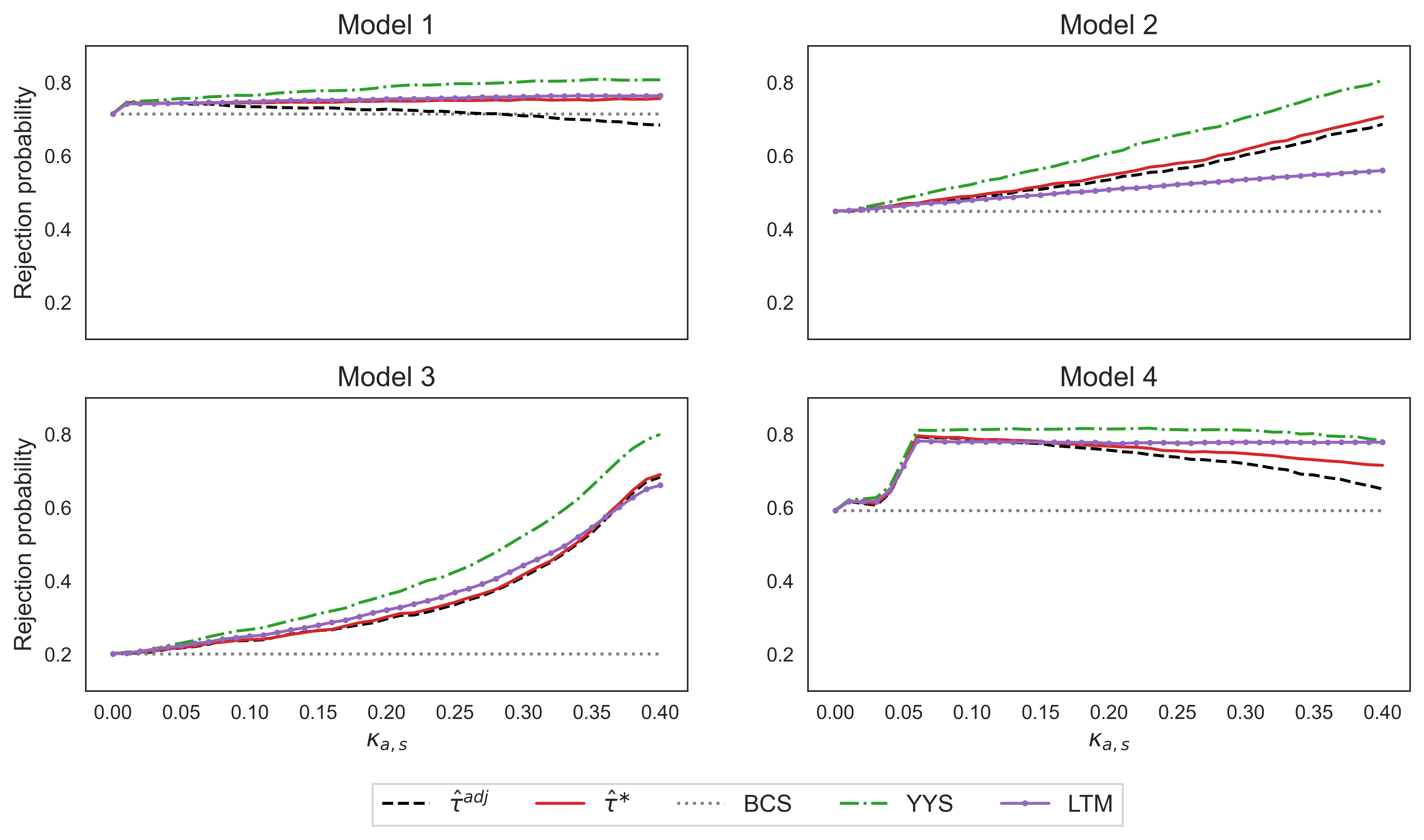}
	}
	\caption{Rejection probabilities under WEI when $n=800$}\label{fig:wei800}
\end{figure}

\subsection{Additional Results for Bias}
Table \ref{Table_Bias_H0} presents the empirical bias of ATE estimators across all cases. The number of regressors is set to \( 0.4\frac{n}{2|\mathcal{S}|} \), corresponding to \( \kappa_{a, s} = 0.4 \). Figure \ref{fig:sbr400bias_H0} illustrates how the empirical bias of \(\hat{\tau}^{adj}\), \(\hat{\tau}^{\ast}\), YYS, and LTM changes with the number of regressors under both the null and alternative hypotheses and the SBR randomization scheme.

We find that the empirical bias terms are asymptotically negligible in all cases compared to the standard deviations. This finding supports our theoretical results that all treatment effect estimators are asymptotically unbiased. Similar results were observed for other randomization schemes, which are omitted here.

\newcolumntype{L}{>{\raggedright\arraybackslash}X} \newcolumntype{C}{>{\centering\arraybackslash}X}
\begin{table}[tbp]\caption{Bias under $H_0:\mu_{1}-\mu_{0}=0$}  \label{Table_Bias_H0}
	\centering
	\begin{tabularx}{1\textwidth}{CCCCCCCCCC} 
		\toprule
		\multicolumn{1}{l}{} & \multicolumn{1}{c}{} &  \multicolumn{4}{c}{n=400, $k_n=40$} &\multicolumn{4}{c}{n=800, $k_n=80$} \\ 
		
		\cmidrule(lr){3-6} \cmidrule(lr){7-10}
		
		Model                & Method               & SRS   & BCD   & WEI   & SBR   & \multicolumn{1}{c}{SRS} & \multicolumn{1}{c}{BCD} & \multicolumn{1}{c}{WEI} & \multicolumn{1}{c}{SBR} \\ 
		\midrule
		1 & $\hat{\tau}^{adj}$  & -0.001 & 0.000  & 0.000  & 0.000  & -0.001 & 0.001  & 0.001  & 0.001  \\
		& $\hat{\tau}^{\ast}$ & -0.001 & 0.000  & 0.000  & 0.000  & 0.000  & 0.001  & 0.001  & 0.001  \\
		& BCS                 & -0.001 & 0.000  & 0.000  & 0.000  & 0.000  & 0.000  & 0.000  & 0.000  \\
		& YYS                 & -0.001 & 0.000  & 0.000  & 0.000  & -0.001 & 0.001  & 0.001  & 0.001  \\
		& LTM                 & -0.001 & 0.000  & 0.000  & 0.000  & 0.000  & 0.000  & 0.000  & 0.000  \\ \hline
		2 & $\hat{\tau}^{adj}$  & 0.000  & 0.000  & 0.000  & 0.000  & -0.001 & 0.000  & 0.000  & 0.000  \\
		& $\hat{\tau}^{\ast}$ & 0.000  & 0.000  & 0.000  & 0.000  & -0.002 & 0.000  & 0.000  & 0.000  \\
		& BCS                 & 0.000  & -0.001 & -0.001 & -0.001 & -0.003 & 0.000  & 0.000  & 0.000  \\
		& YYS                 & 0.000  & 0.000  & 0.000  & 0.000  & -0.001 & 0.000  & 0.000  & 0.000  \\
		& LTM                 & 0.000  & -0.001 & -0.001 & -0.001 & -0.002 & 0.000  & 0.000  & 0.000  \\ \hline
		3 & $\hat{\tau}^{adj}$  & 0.000  & 0.000  & 0.000  & 0.000  & -0.001 & 0.000  & 0.000  & 0.000  \\
		& $\hat{\tau}^{\ast}$ & 0.000  & -0.001 & -0.001 & -0.001 & -0.001 & 0.000  & 0.000  & 0.000  \\
		& BCS                 & 0.001  & -0.002 & -0.002 & -0.002 & -0.004 & 0.000  & 0.000  & 0.000  \\
		& YYS                 & 0.000  & 0.000  & 0.000  & 0.000  & -0.001 & 0.000  & 0.000  & 0.000  \\
		& LTM                 & 0.000  & -0.001 & -0.001 & -0.001 & -0.002 & -0.001 & -0.001 & -0.001 \\ \hline
		4 & $\hat{\tau}^{adj}$  & 0.000  & 0.002  & 0.002  & 0.002  & 0.001  & 0.001  & 0.001  & 0.001  \\
		& $\hat{\tau}^{\ast}$ & 0.000  & 0.002  & 0.002  & 0.002  & 0.001  & 0.000  & 0.000  & 0.000  \\
		& BCS                 & 0.001  & 0.002  & 0.002  & 0.002  & 0.002  & 0.001  & 0.001  & 0.001  \\
		& YYS                 & 0.000  & 0.002  & 0.002  & 0.002  & 0.001  & 0.001  & 0.001  & 0.001  \\
		& LTM                 & 0.001  & 0.002  & 0.002  & 0.002  & 0.002  & 0.001  & 0.001  & 0.001   \\ 
		
		\bottomrule
	\end{tabularx}
\end{table}

\begin{figure}[H]
	\centering
	\includegraphics[width=\textwidth]{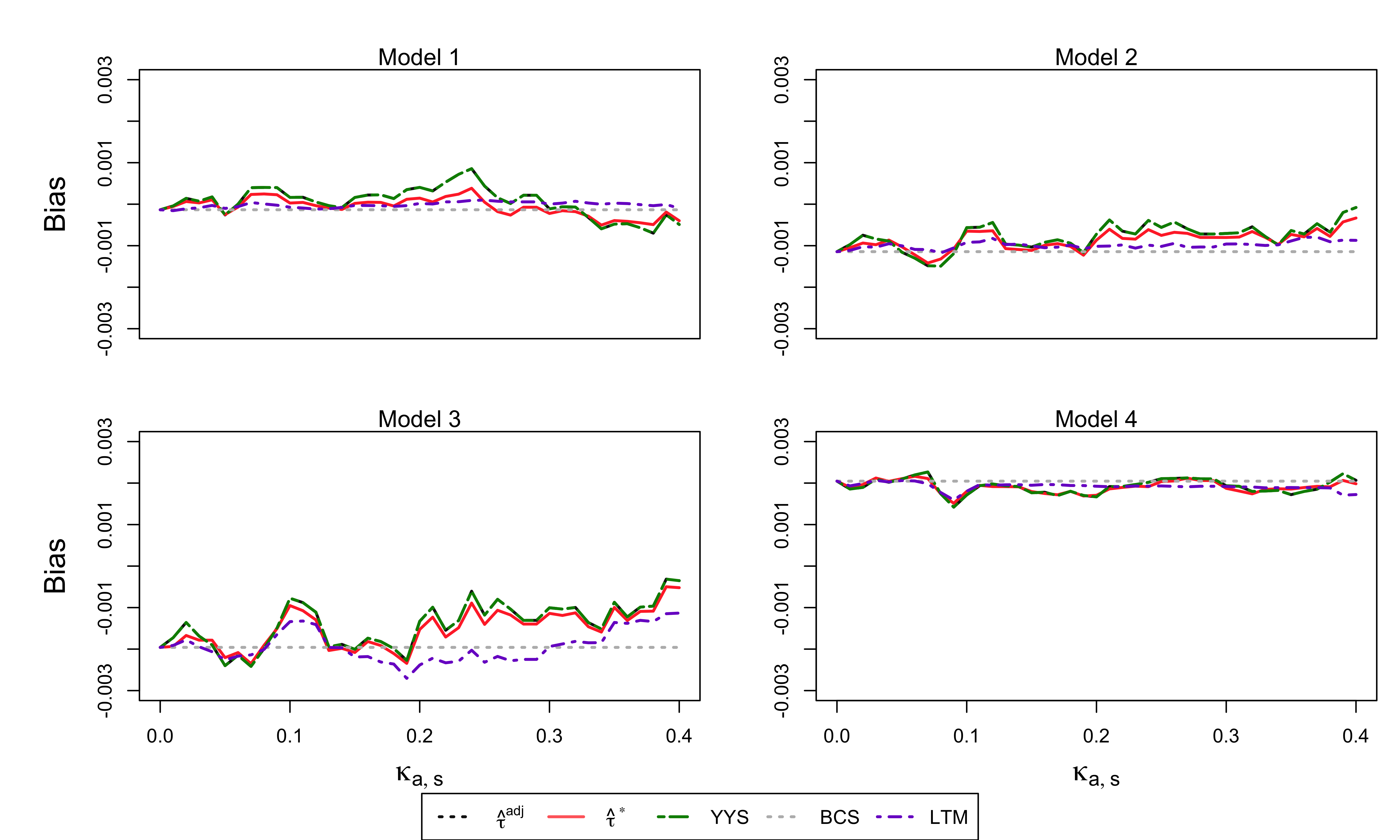}
	
	\caption{Bias under SBR  and $H_0$ when $n=400$}\label{fig:sbr400bias_H0}
\end{figure}

\subsection{Additional Results for Standard Deviation}
Table \ref{Table_Std_H0} presents the standard deviations of ATE estimators for \( n = 400 \) and \( n = 800 \) under \( H_{0} : \mu_1 = \mu_0 = 0 \). The number of regressors is set to \( 0.4\frac{n}{2|\mathcal{S}|} \), corresponding to \( \kappa_{a, s} = 0.4 \). Values in parentheses indicate the ratio of standard deviations to the average of estimated standard errors. These results confirm that the standard error estimator proposed in this paper for \( \hat{\tau}^{adj} \) and \( \hat{\tau}^{\ast} \) is consistent. In contrast, YYS significantly underestimates standard errors across all four models, while LTM also underestimates standard errors for models 1-3, though to a lesser extent.

\begin{figure}[H]
	\includegraphics[width=\textwidth]{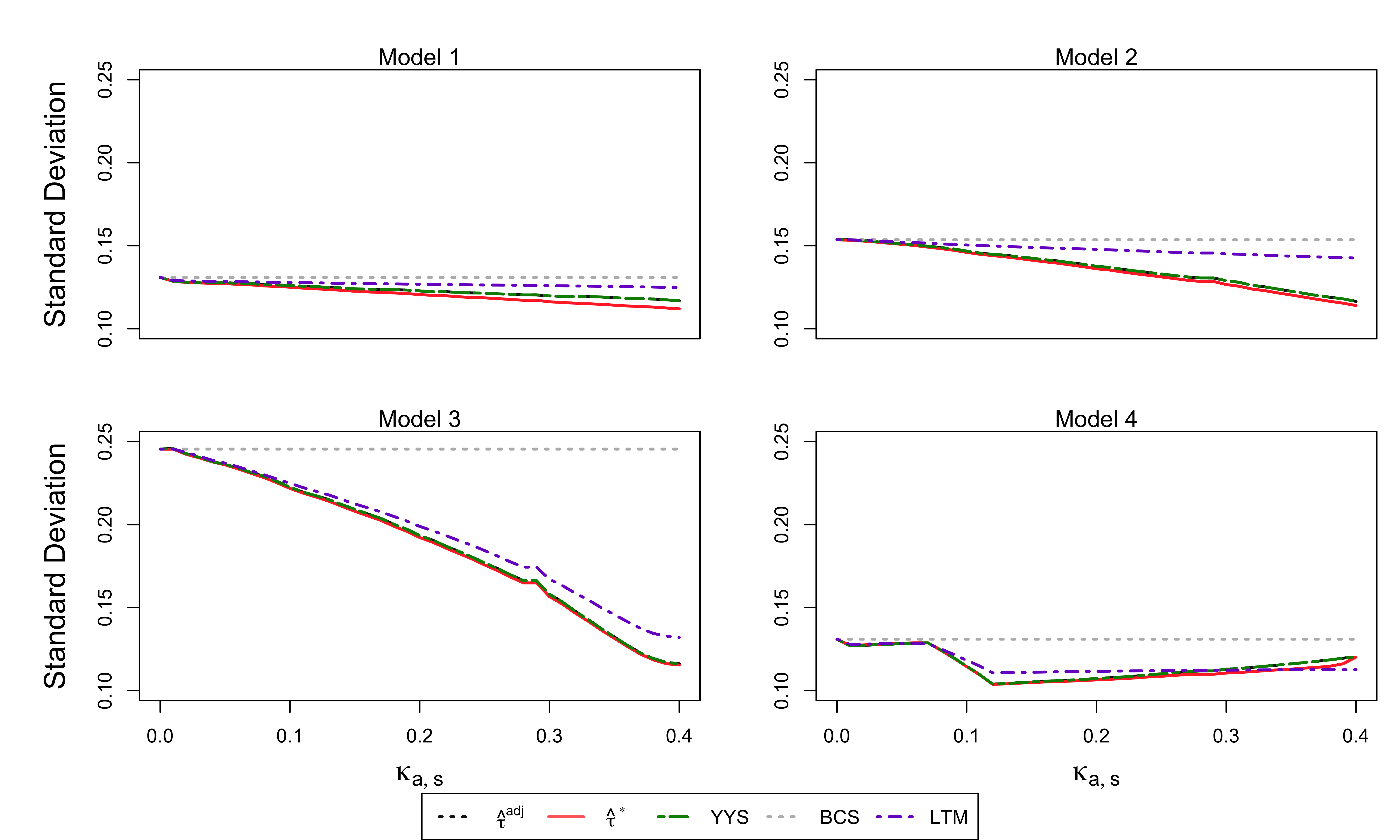}
	\caption{Standard deviation under SBR and $H_0$ when $n=400$}\label{fig:sbr400Std}
\end{figure}

Figure \ref{fig:sbr400Std} illustrates how the standard deviations of \(\hat{\tau}^{adj}\), \(\hat{\tau}^{\ast}\), YYS, and LTM change with the number of regressors under the null hypothesis and the SBR randomization scheme. For \( n = 800 \), we set the number of regressors to \( k_n = 0, 2, \ldots, 80 \), and for \( n = 400 \), to \( k_n = 0, 1, \ldots, 40 \), corresponding to \( \kappa_{a, s} = 0, 0.02, \ldots, 0.4 \) on the \( x \)-axis.

{\small 
	\newcolumntype{L}{>{\raggedright\arraybackslash}X} \newcolumntype{C}{>{\centering\arraybackslash}X}
	\begin{table}[H]\caption{Standard Deviations under $H_0:\mu_{1}-\mu_{0}=0$}  \label{Table_Std_H0}
		\centering
		\begin{tabularx}{1\textwidth}{CCCCCCCCCC} 
			\toprule
			\multicolumn{1}{l}{} & \multicolumn{1}{c}{} &  \multicolumn{4}{c}{n=400, $k_n=40$} &\multicolumn{4}{c}{n=800, $k_n=80$} \\ 
			\cmidrule(lr){3-6} \cmidrule(lr){7-10}
			Model                & Method               & SRS   & BCD   & WEI   & SBR   & \multicolumn{1}{c}{SRS} & \multicolumn{1}{c}{BCD} & \multicolumn{1}{c}{WEI} & \multicolumn{1}{c}{SBR} \\ 
			\midrule
			1 & $\hat{\tau}^{adj}$  & 0.118   & 0.117   & 0.118   & 0.117   & 0.082   & 0.082   & 0.083   & 0.084   \\
			&                     & (1.021) & (1.018) & (1.021) & (1.012) & (1.010) & (1.012) & (1.022) & (1.034) \\
			& $\hat{\tau}^{\ast}$ & 0.113   & 0.113   & 0.112   & 0.112   & 0.079   & 0.079   & 0.080   & 0.080   \\
			&                     & (1.027) & (1.030) & (1.018) & (1.018) & (1.017) & (1.019) & (1.030) & (1.030) \\
			& BCS                 & 0.131   & 0.133   & 0.131   & 0.131   & 0.093   & 0.094   & 0.094   & 0.093   \\
			&                     & (0.992) & (1.010) & (0.992) & (0.992) & (1.003) & (1.015) & (1.013) & (1.003) \\
			& YYS                 & 0.118   & 0.117   & 0.118   & 0.117   & 0.082   & 0.082   & 0.083   & 0.084   \\
			&                     & (1.264) & (1.254) & (1.264) & (1.253) & (1.253) & (1.255) & (1.268) & (1.283) \\
			& LTM                 & 0.125   & 0.127   & 0.125   & 0.125   & 0.089   & 0.089   & 0.090   & 0.089   \\
			&                     & (1.066) & (1.085) & (1.066) & (1.066) & (1.071) & (1.073) & (1.083) & (1.071) \\ \hline
			2 & $\hat{\tau}^{adj}$  & 0.117   & 0.116   & 0.115   & 0.116   & 0.082   & 0.082   & 0.082   & 0.082   \\
			&                     & (1.012) & (1.009) & (0.995) & (1.004) & (1.000) & (1.003) & (1.000) & (1.000) \\
			& $\hat{\tau}^{\ast}$ & 0.114   & 0.113   & 0.112   & 0.114   & 0.080   & 0.079   & 0.080   & 0.080   \\
			&                     & (1.015) & (1.010) & (0.997) & (1.015) & (1.002) & (0.993) & (1.002) & (1.002) \\
			& BCS                 & 0.152   & 0.151   & 0.153   & 0.154   & 0.108   & 0.108   & 0.108   & 0.108   \\
			&                     & (0.994) & (0.990) & (1.000) & (1.007) & (0.998) & (0.999) & (0.998) & (0.998) \\
			& YYS                 & 0.117   & 0.116   & 0.115   & 0.116   & 0.082   & 0.082   & 0.082   & 0.082   \\
			&                     & (1.242) & (1.233) & (1.221) & (1.232) & (1.234) & (1.235) & (1.234) & (1.234) \\
			& LTM                 & 0.142   & 0.140   & 0.142   & 0.143   & 0.101   & 0.101   & 0.102   & 0.101   \\
			&                     & (1.097) & (1.084) & (1.097) & (1.105) & (1.079) & (1.081) & (1.090) & (1.079) \\ \hline
			3 & $\hat{\tau}^{adj}$  & 0.117   & 0.116   & 0.114   & 0.116   & 0.082   & 0.082   & 0.082   & 0.083   \\
			&                     & (1.031) & (1.026) & (1.004) & (1.022) & (1.010) & (1.012) & (1.010) & (1.022) \\
			& $\hat{\tau}^{\ast}$ & 0.116   & 0.115   & 0.113   & 0.115   & 0.082   & 0.081   & 0.081   & 0.082   \\
			&                     & (1.032) & (1.026) & (1.005) & (1.023) & (1.018) & (1.008) & (1.006) & (1.018) \\
			& BCS                 & 0.244   & 0.242   & 0.246   & 0.245   & 0.175   & 0.176   & 0.176   & 0.176   \\
			&                     & (0.997) & (0.992) & (1.005) & (1.001) & (0.997) & (1.004) & (1.003) & (1.003) \\
			& YYS                 & 0.117   & 0.116   & 0.114   & 0.116   & 0.082   & 0.082   & 0.082   & 0.083   \\
			&                     & (1.252) & (1.243) & (1.220) & (1.241) & (1.239) & (1.241) & (1.239) & (1.254) \\
			& LTM                 & 0.133   & 0.131   & 0.131   & 0.132   & 0.096   & 0.095   & 0.096   & 0.095   \\
			&                     & (1.183) & (1.169) & (1.165) & (1.174) & (1.183) & (1.174) & (1.183) & (1.171) \\ \hline
			4 & $\hat{\tau}^{adj}$  & 0.120   & 0.121   & 0.122   & 0.120   & 0.086   & 0.087   & 0.087   & 0.087   \\
			&                     & (1.016) & (1.032) & (1.033) & (1.016) & (1.019) & (1.035) & (1.031) & (1.031) \\
			& $\hat{\tau}^{\ast}$ & 0.115   & 0.115   & 0.117   & 0.120   & 0.081   & 0.082   & 0.082   & 0.082   \\
			&                     & (1.056) & (1.056) & (1.075) & (1.102) & (1.039) & (1.053) & (1.051) & (1.051) \\
			& BCS                 & 0.129   & 0.129   & 0.131   & 0.131   & 0.091   & 0.093   & 0.092   & 0.092   \\
			&                     & (0.986) & (0.988) & (1.002) & (1.002) & (0.992) & (1.015) & (1.003) & (1.003) \\
			& YYS                 & 0.120   & 0.121   & 0.122   & 0.120   & 0.086   & 0.087   & 0.087   & 0.087   \\
			&                     & (1.261) & (1.275) & (1.282) & (1.261) & (1.281) & (1.298) & (1.296) & (1.296) \\
			& LTM                 & 0.111   & 0.112   & 0.113   & 0.113   & 0.075   & 0.076   & 0.076   & 0.076   \\
			&                     & (1.032) & (1.045) & (1.051) & (1.051) & (1.031) & (1.047) & (1.045) & (1.045)\\
			\bottomrule
			\multicolumn{10}{l}{Note: All simulations were conducted with 10,000 replications. The ratio of sample standard}\\
			\multicolumn{10}{l}{deviations to the average of estimated standard errors is given in parenthesis.}
		\end{tabularx}
	\end{table}
}

In most cases, \(\hat{\tau}^\ast\) shows the smallest standard deviations, confirming our theoretical result that \(\hat{\tau}^\ast\) is the most efficient combination of \(\hat{\tau}^{adj}\) and the unadjusted estimator. In Model 4, however, when \(\kappa_{\alpha, s}\) is above \(0.3\), LTM achieves a slightly smaller standard deviation. This outcome is consistent with our theory, as the Lasso estimator does not fall within the class of estimators based on \(\hat{\tau}^{adj}\) and \(\hat{\tau}^{unadj}\). The main challenge with using Lasso for adjustment is its inability to control size when the sparsity condition is not met, as shown in Table 1 for Model 4 in the main text. Results under other randomization schemes are similar and are therefore omitted.

\subsection{Simulation analysis of the sensitivity to violation of Assumption 2(iv)}

In this section, we consider two DGPs that violate  Assumption 2(iv). Similar to models 1-4, our data generating processes are as follows.	
For $a\in\{0,1\}$, we generate potential outcomes according to the equation
$$
Y_{i}(a)=\mu_{a}+m_{a}(Z_{i})+ \sigma_{a}(Z_{i})\epsilon_{i}(a),
$$
where $\mu_{a}$, and $m_{a}\left(Z_{i}\right)$ and $\sigma_a(Z_i)$ are specified as follows. In all specifications, $\{Z_i, \epsilon_{i}(1), \epsilon_{i}(0)\}$ are i.i.d.  The conditional mean functions $m_a(\cdot)$ are different for each model. 

\begin{description}
	\item [{Model 5:} Model with incorrect regressors] 
	The dimension of $Z_i$ is set to $d_n=0.2n/|\mathcal{S}|$. The first entry of $Z_i$ is uniformly distributed in $[-1,1]$, i.e., $Z_{1i}\sim U[-1,1]$. 
	The other entries of $Z_i$ are$[Z_{2i},\cdots,Z_{d_n,i}]^\top=\Sigma_\rho^{1/2}\mathbf{v}_i$, where   $\mathbf{v}_i\in \mathbb{R}^{d_n-1}$ containing i.i.d. $U[-1,1]$ entries and $\Sigma_\rho$ is a Toeplitz matrix with
	$
	[\Sigma_\rho]_{i,j}=\rho^{|i-j|}
	$ and $\rho=0.6$. 
	The conditional variances of idiosyncratic terms are given by 
	$\sigma_{0}(Z_{i})=\sigma_{1}(Z_{i})=c_{\epsilon}\left[1+\left(Z_{1i}+\sum_{j=2}^{d_n}Z_{ji}/\sqrt{d_n-1}\right)^2\right]^{1/2}.$ 
	We set $c_\eps$ as the normalizing constant such that $\mathbb E \sigma_a^2(Z_i) = 1$. Set $(\epsilon_{i}(1),\epsilon_{i}(0))\sim \mathcal{N}(0,\mathbf{I}_2)$ to be independent of $Z_i$. 
	
	Model 5 is the same as Model 3, but with the conditional means defined as:
	
	\begin{center}
		$m_1(Z_i)=2\sum_{j=2}^{d_n}Z_{ji}/\sqrt{d_n-1}$, and
		$m_0(Z_i)=2\sum_{j=2}^{d_n}Z_{ji}^3/\sqrt{d_n-1}.$
	\end{center}
	
	\textbf{We directly use $Z_{ji}$'s as regressors so that the regression adjustment for $m_{0}(\cdot)$ is not correctly specified.}
	
	\item [{Model 6:}] The setting is the same as Model 5, but with the conditional means defined as:
	
	\begin{center}
		$m_1(Z_i)=2\sum_{j=2}^{d_n}Z_{ji}/\sqrt{d_n-1}$, and 
		$m_0(Z_i)=2\sum_{j=2}^{d_n}\left(\frac{Z_{ji}^2-1/3}{3}+\frac{Z_{ji}}{10}\right)/\sqrt{d_n-1}$.
	\end{center}
	\textbf{We directly use $Z_{ji}$'s as regressors, so that the regression adjustment for $m_{0}(\cdot)$ is not correctly specified.}
\end{description}

Table \ref{Table_rate_H0_model56} presents the test sizes for the estimators, while Table \ref{Table_rate_H1_model56} shows their power under \( H_{1}: \mu_1 - \mu_0 = 0.2 \). The number of regressors is set to \( 0.4\frac{n}{2|\mathcal{S}|} \), corresponding to \( \kappa_{a, s} = 0.4 \). The results are similar to those of models 1-4. Specifically, under \( H_0 \), \(\hat{\tau}^{adj}\) and \(\hat{\tau}^{\ast}\) maintain good size control across both models, with accuracy improving as \( n \) increases. In contrast, YYS and LTM significantly over-reject the null hypothesis. BCS, which does not use any regressors, is unaffected by the size distortion that arises with many regressors. Under \( H_1 \), \(\hat{\tau}^{\ast}\) consistently shows higher power than \(\hat{\tau}^{adj}\) and BCS, as predicted by our theory. YYS and LTM exhibit higher rejection rates due to their use of an incorrect variance estimator.

Figures \ref{fig:SBR_dgp5}-\ref{fig:SBR_dgp6} show the rejection probabilities, bias, and standard deviations of the methods considered, with varying numbers of regressors under the SBR randomization scheme. We examine Models 5 and 6 for \( n = 800 \) and \( n = 400 \), setting the number of regressors to \( k_n = 0, 2, \ldots, 80 \) for \( n = 800 \) and \( k_n = 0, 1, \ldots, 40 \) for \( n = 400 \), corresponding to \( \kappa_{a, s} = 0, 0.02, \ldots, 0.4 \). The \( x \)-axis represents \( \kappa_{a, s} \). 

Despite Assumption 2(iv) not holding, \(\hat{\tau}^\ast\) and \(\hat{\tau}^{adj}\) perform well. First, panels (a) of the figures indicate that \(\hat{\tau}^{adj}\) and \(\hat{\tau}^\ast\) maintain uniform size control across different values of \( \kappa_{a, s} \), with close-to-nominal rejection rates under \( H_0 \), even as the number of regressors ranges from 0 to 80 (40) for \( n = 800 \) (\( n = 400 \)). In contrast, LTM and YYS lack uniform size control, as their rejection probabilities increase with the number of regressors. Second, \(\hat{\tau}^\ast\) consistently achieves higher power than both BCS and \(\hat{\tau}^{adj}\). Third, we observe that the bias is nearly zero in Model 5 and remains small in Model 6, with a slight increase as \( \kappa_{a, s} \) grows. However, this bias is asymptotically negligible relative to the standard deviations. Fourth, \(\hat{\tau}^\ast\) has the smallest standard deviation among all methods. Similar results were obtained under other randomization schemes and are omitted here for brevity.

\newcolumntype{L}{>{\raggedright\arraybackslash}X} \newcolumntype{C}{>{\centering\arraybackslash}X}
\begin{table}[tbp]\caption{Rejection rate (in percent) under $H_0:\mu_{1}-\mu_{0}=0$}  \label{Table_rate_H0_model56}
	\centering
	\begin{tabularx}{1\textwidth}{CCCCCCCCCC} 
		\toprule
		\multicolumn{1}{l}{} & \multicolumn{1}{c}{} &  \multicolumn{4}{c}{n=400, $k_n=40$} &\multicolumn{4}{c}{n=800, $k_n=80$} \\ 
		\cmidrule(lr){3-6} \cmidrule(lr){7-10}
		Model                & Method               & SRS   & BCD   & WEI   & SBR   & \multicolumn{1}{c}{SRS} & \multicolumn{1}{c}{BCD} & \multicolumn{1}{c}{WEI} & \multicolumn{1}{c}{SBR} \\ 
		\midrule
		5 & $\hat{\tau}^{adj}$  & 5.44  & 5.68  & 5.34  & 5.83  & 5.39 & 5.51  & 5.4   & 5.37  \\
		& $\hat{\tau}^{\ast}$ & 5.66  & 5.53  & 5.23  & 5.59  & 5.45 & 5.61  & 5.39  & 5.26  \\
		& BCS                 & 5.11  & 4.64  & 5.14  & 5.24  & 5.02 & 5.15  & 5.18  & 5.08  \\
		& YYS                 & 11.63 & 11.35 & 11.12 & 11.64 & 11.8 & 11.49 & 11.44 & 11.81 \\
		& LTM                 & 8.13  & 7.5   & 7.5   & 7.95  & 7.9  & 7.91  & 7.86  & 7.75  \\ \hline
		6 & $\hat{\tau}^{adj}$  & 5.52  & 5.51  & 5.18  & 5.34  & 5.41 & 5.53  & 5.11  & 5.45  \\
		& $\hat{\tau}^{\ast}$ & 5.47  & 5.39  & 5.25  & 5.29  & 5.59 & 5.58  & 5.08  & 5.67  \\
		& BCS                 & 5.24  & 4.86  & 5.03  & 5.31  & 5.41 & 5.01  & 5.18  & 5.63  \\
		& YYS                 & 8.11  & 8.05  & 7.78  & 8.08  & 8.61 & 8.46  & 8.15  & 8.37  \\
		& LTM                 & 9.41  & 9.02  & 8.89  & 9.14  & 9.88 & 9.77  & 9.43  & 9.97 \\
		\bottomrule		
	\end{tabularx}
\end{table}

\newcolumntype{L}{>{\raggedright\arraybackslash}X} \newcolumntype{C}{>{\centering\arraybackslash}X}
\begin{table}[tbp]\caption{Rejection rate (in percent) under $H_1:\mu_{1}-\mu_{0}=0.2$}  \label{Table_rate_H1_model56}
	\centering
	\begin{tabularx}{1\textwidth}{CCCCCCCCCC} 
		\toprule
		\multicolumn{1}{l}{} & \multicolumn{1}{c}{} &  \multicolumn{4}{c}{n=400, $k_n=40$} &\multicolumn{4}{c}{n=800, $k_n=80$} \\ 			
		\cmidrule(lr){3-6} \cmidrule(lr){7-10}
		Model                & Method               & SRS   & BCD   & WEI   & SBR   & \multicolumn{1}{c}{SRS} & \multicolumn{1}{c}{BCD} & \multicolumn{1}{c}{WEI} & \multicolumn{1}{c}{SBR} \\ 
		\midrule
		5 & $\hat{\tau}^{adj}$  & 39.88 & 39.47 & 39.06 & 39.56 & 66.08 & 66.16 & 66.31 & 66.04 \\
		& $\hat{\tau}^{\ast}$ & 41.68 & 41.81 & 41.20 & 41.76 & 68.86 & 69.19 & 68.64 & 68.47 \\
		& BCS                 & 27.15 & 27.36 & 27.20 & 27.05 & 47.01 & 48.24 & 47.73 & 47.90 \\
		& YYS                 & 52.95 & 53.13 & 52.79 & 53.00 & 78.10 & 78.12 & 78.14 & 78.03 \\
		& LTM                 & 45.67 & 45.89 & 45.94 & 45.65 & 71.08 & 71.85 & 71.35 & 71.25 \\ \hline
		6 & $\hat{\tau}^{adj}$  & 27.30 & 26.75 & 26.85 & 26.91 & 44.30 & 45.04 & 45.18 & 44.50 \\
		& $\hat{\tau}^{\ast}$ & 27.66 & 27.16 & 27.32 & 27.18 & 45.16 & 45.87 & 45.81 & 45.28 \\
		& BCS                 & 18.72 & 17.97 & 17.99 & 18.29 & 30.37 & 31.37 & 31.07 & 31.06 \\
		& YYS                 & 33.71 & 32.98 & 33.18 & 33.32 & 52.33 & 53.07 & 52.85 & 52.72 \\
		& LTM                 & 35.56 & 34.81 & 35.24 & 35.04 & 54.12 & 54.90 & 54.78 & 54.16   \\ 			
		\bottomrule			
	\end{tabularx}
\end{table}

\begin{figure}
	\centering
	\subfigure[$H_{0}:\mu_1-\mu_0=0$]{
		\includegraphics[width=\textwidth]{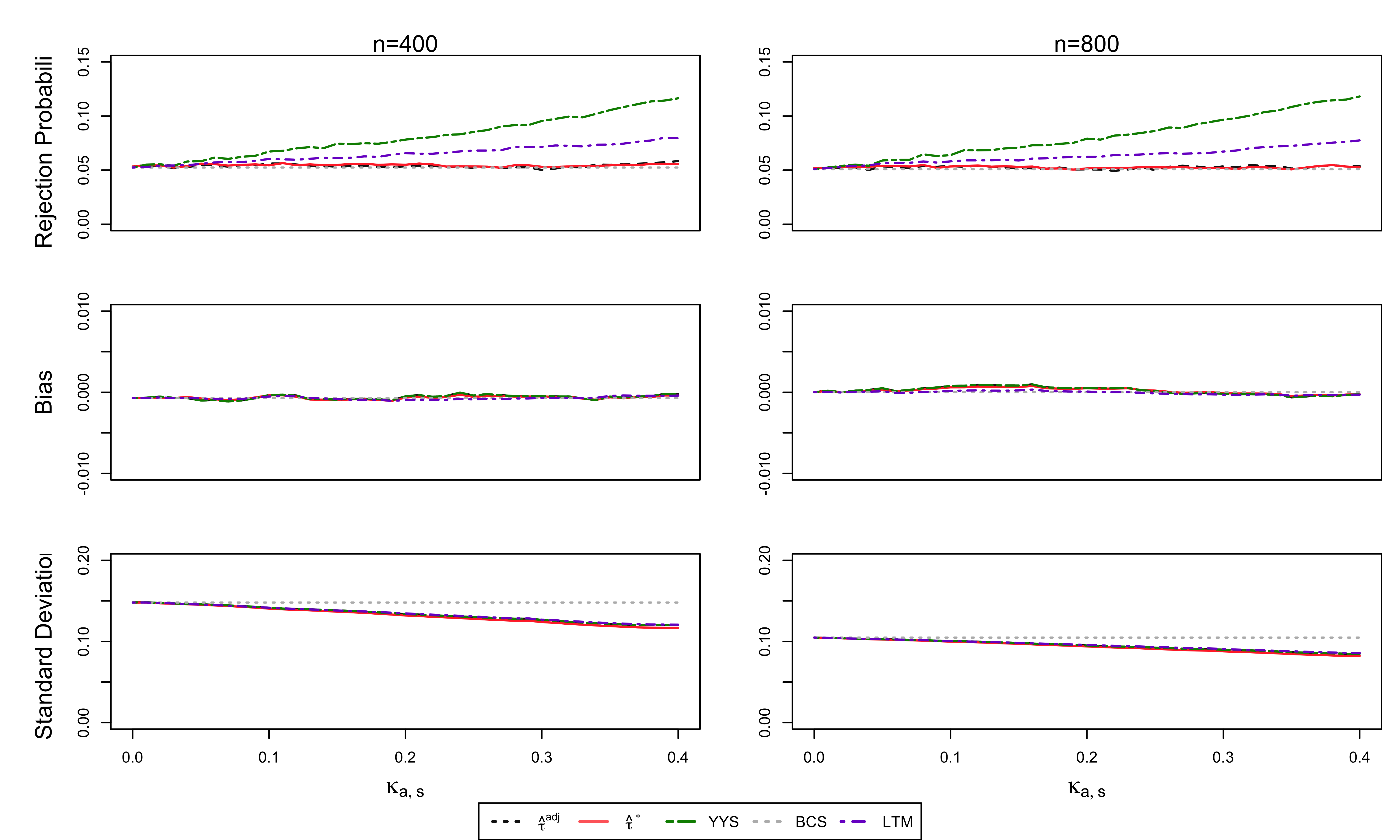}
	}
	\hfill
	\subfigure[$H_{1}:\mu_1-\mu_0=0.2$]{
		\includegraphics[width=\textwidth]{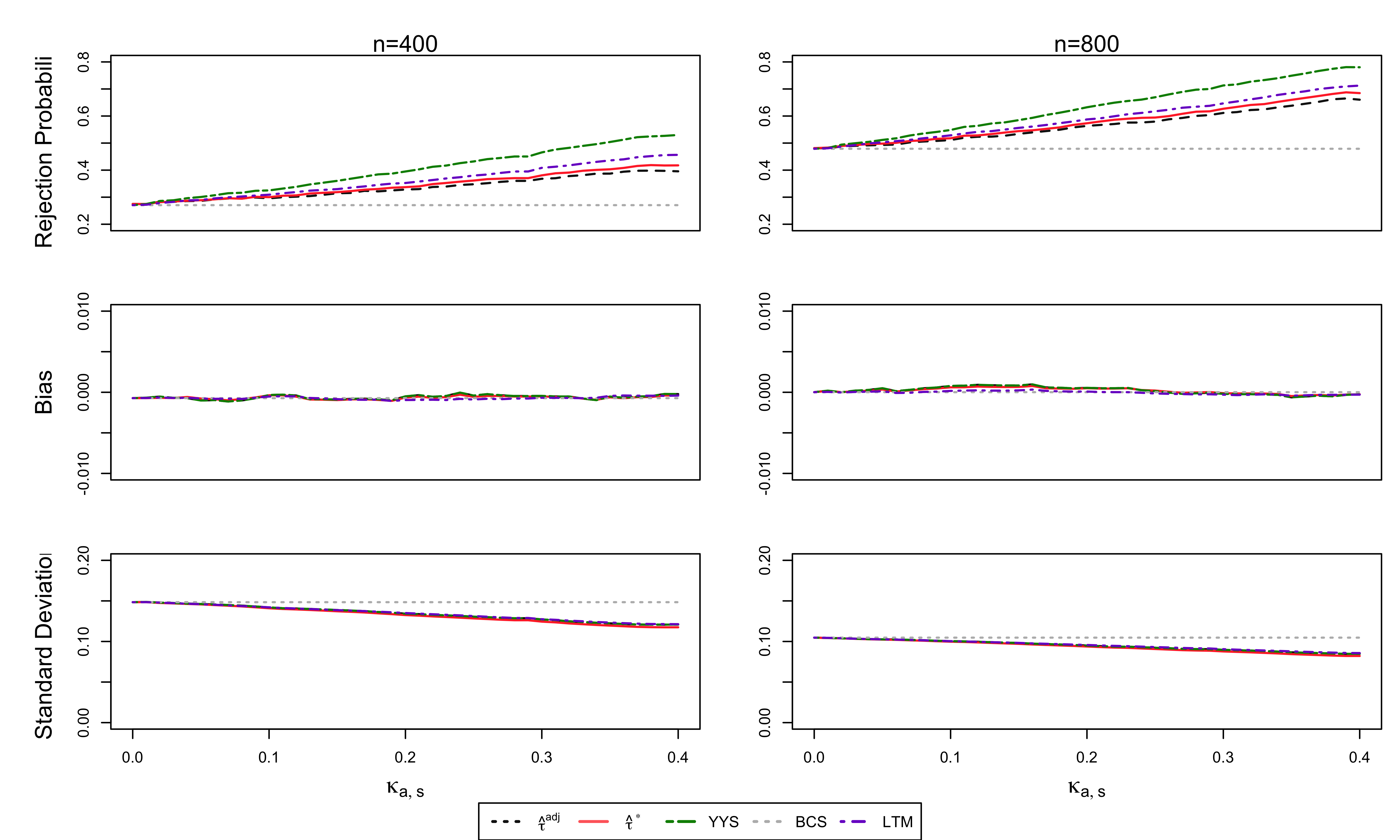}
	}
	\caption{Model 5 under SBR}\label{fig:SBR_dgp5}
\end{figure}

\begin{figure}
	\centering
	\subfigure[$H_{0}:\mu_1-\mu_0=0$]{
		\includegraphics[width=\textwidth]{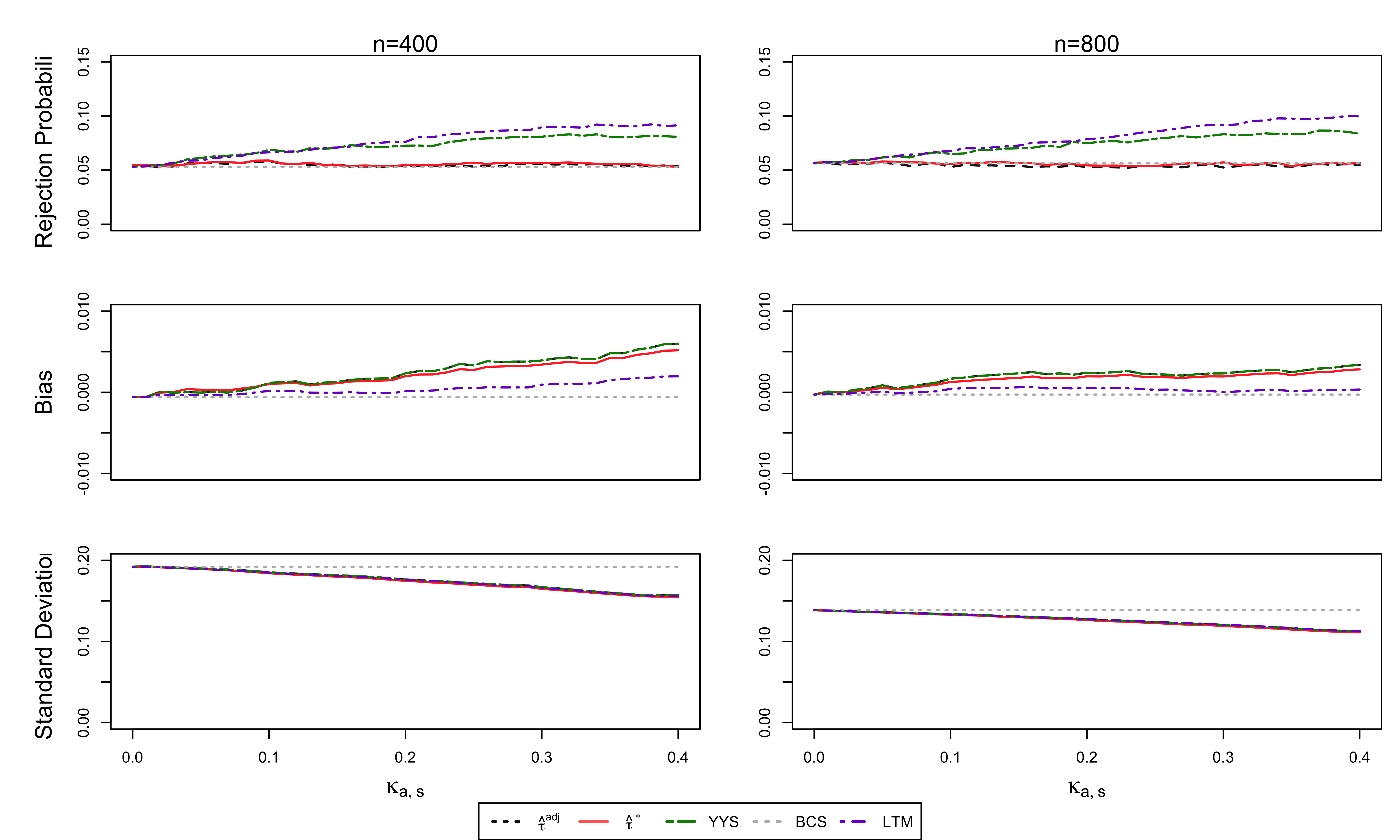}
	}
	\hfill
	\subfigure[$H_{1}:\mu_1-\mu_0=0.2$]{
		\includegraphics[width=\textwidth]{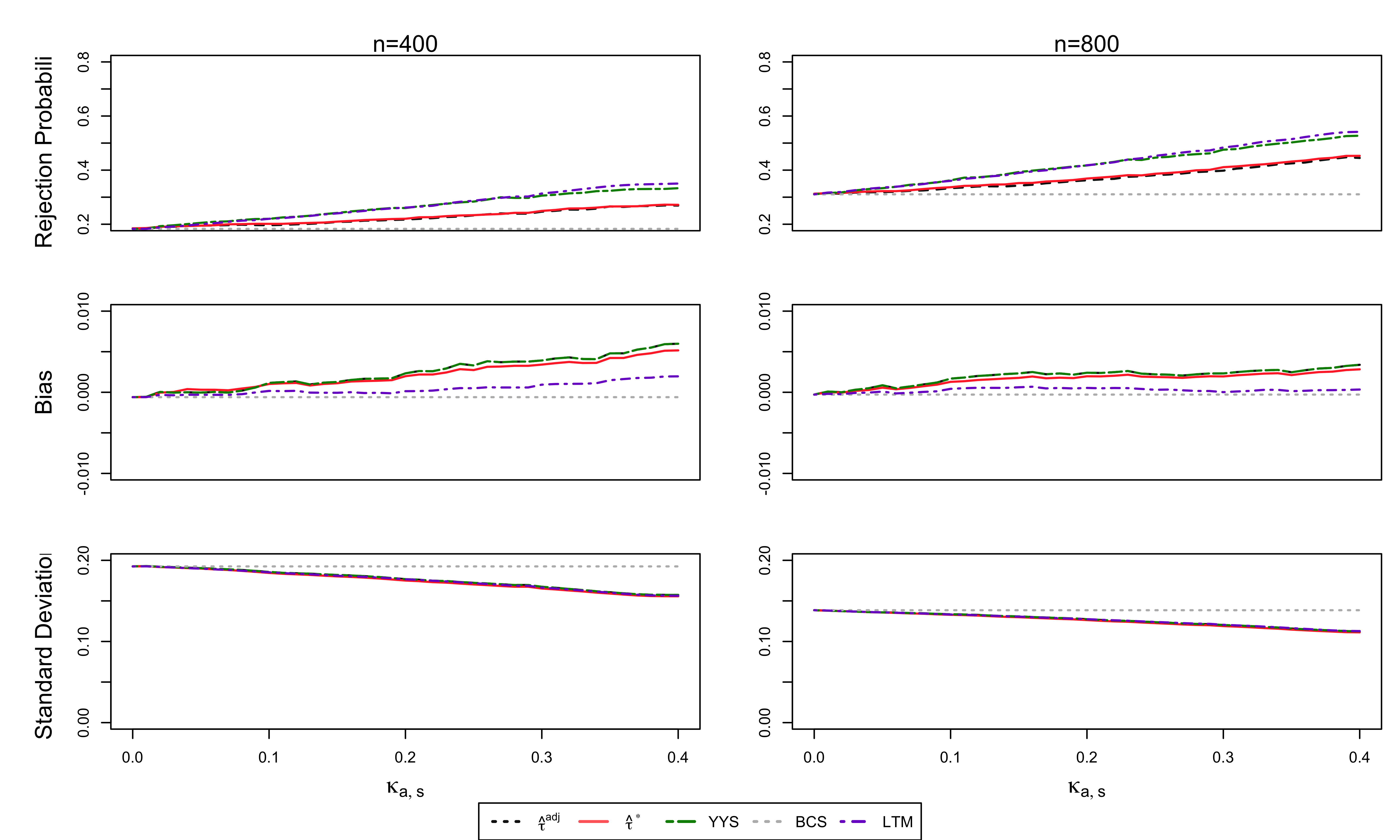}
	}
	\caption{Model 6 under SBR}\label{fig:SBR_dgp6}
\end{figure}

\newpage
\bibliographystyle{chicago}
\bibliography{BCAR}

\end{document}